%% file: The_Nature_of_Concurrency.tex
\documentclass[11pt]{article} 

\usepackage{textcomp, amssymb}  
\usepackage{anysize}            
\usepackage{amsfonts}
\usepackage{amsmath}
\usepackage{dsfont}
\usepackage{MnSymbol}
\usepackage{mathtools}
\usepackage{graphicx}
\usepackage{epsfig}
\usepackage{epstopdf}
\usepackage{lmerge}
\usepackage{tikz}
\usepackage{amsthm}
\usepackage{hyperref}
\usepackage{ifpdf}
\usepackage{thm-restate}

\usetikzlibrary{calc,decorations.pathmorphing,shapes,graphs}

\newcounter{sarrow}

\newcounter{sarrow1}
\newcommand\xnrsquigarrow[1]{%
\stepcounter{sarrow1}%
\mathrel{\begin{tikzpicture}[baseline= {( $ (current bounding box.south) + (0,-0.5ex) $ )}]
\node[inner sep=.5ex] (\thesarrow) {$\scriptstyle #1$};
\path[draw,<-,decorate,
  decoration={zigzag,amplitude=0.7pt,segment length=1.2mm,pre=lineto,pre length=4pt}]
    (\thesarrow1.south east) -- (\thesarrow1.south west);
    $\slashedarrowfill@\relbar\relbar/$
    \end{tikzpicture}}%
}

\makeatletter
\def\slashedarrowfill@#1#2#3#4#5{%
  $\m@th\thickmuskip0mu\medmuskip\thickmuskip\thinmuskip\thickmuskip
   \relax#5#1\mkern-7mu%
   \cleaders\hbox{$#5\mkern-2mu#2\mkern-2mu$}\hfill
   \mathclap{#3}\mathclap{#2}%
   \cleaders\hbox{$#5\mkern-2mu#2\mkern-2mu$}\hfill
   \mkern-7mu#4$%
}
\def\rightslashedarrowfillb@{%
  \slashedarrowfill@\relbar\relbar/\rightarrow}
\newcommand\xnrightarrow[2][]{%
  \ext@arrow 0055{\rightslashedarrowfillb@}{#1}{#2}}

\def\rightslashedarrowfille@{%
  \slashedarrowfill@\relbar\relbar/\twoheadrightarrow}
\newcommand\xntworightarrow[2][]{%
  \ext@arrow 0055{\rightslashedarrowfille@}{#1}{#2}}

\def\rightslashedarrowfillg@{%
  \slashedarrowfill@\relbar\relbar{\raisebox{.12em}{}}\twoheadrightarrow}
\newcommand\xtworightarrow[2][]{%
  \ext@arrow 0055{\rightslashedarrowfillg@}{#1}{#2}}

\def\rightslashedarrowfillx@{%
  \slashedarrowfill@\Relbar\Relbar/\rightrightarrows}
\newcommand\xnTworightarrow[2][]{%
  \ext@arrow 0055{\rightslashedarrowfillx@}{#1}{#2}}

\def\rightslashedarrowfilly@{%
  \slashedarrowfill@\Relbar\Relbar{\raisebox{.12em}{}}\rightrightarrows}
\newcommand\xTworightarrow[2][]{%
  \ext@arrow 0055{\rightslashedarrowfilly@}{#1}{#2}}

\pgfdeclareshape{slash underlined}
{
  \inheritsavedanchors[from=rectangle] 
  \inheritanchorborder[from=rectangle]
  \inheritanchor[from=rectangle]{north}
  \inheritanchor[from=rectangle]{north west}
  \inheritanchor[from=rectangle]{north east}
  \inheritanchor[from=rectangle]{center}
  \inheritanchor[from=rectangle]{west}
  \inheritanchor[from=rectangle]{east}
  \inheritanchor[from=rectangle]{mid}
  \inheritanchor[from=rectangle]{mid west}
  \inheritanchor[from=rectangle]{mid east}
  \inheritanchor[from=rectangle]{base}
  \inheritanchor[from=rectangle]{base west}
  \inheritanchor[from=rectangle]{base east}
  \inheritanchor[from=rectangle]{south}
  \inheritanchor[from=rectangle]{south west}
  \inheritanchor[from=rectangle]{south east}
  \inheritanchorborder[from=rectangle]
  \foregroundpath{
    \southwest \pgf@xa=\pgf@x \pgf@ya=\pgf@y
    \northeast \pgf@xb=\pgf@x \pgf@yb=\pgf@y
    \pgf@xc=\pgf@xa
    \advance\pgf@xc by .5\pgf@xb
    \pgf@yc=\pgf@ya
    \advance\pgf@xc by -1.3pt
    \advance\pgf@yc by -1.8pt
    \pgfpathmoveto{\pgfqpoint{\pgf@xc}{\pgf@yc}}
    \advance\pgf@xc by  2.6pt
    \advance\pgf@yc by  3.6pt
    \pgfpathlineto{\pgfqpoint{\pgf@xc}{\pgf@yc}}
    \pgfpathmoveto{\pgfqpoint{\pgf@xa}{\pgf@ya}}
    \pgfpathlineto{\pgfqpoint{\pgf@xb}{\pgf@ya}}
 }
}
\tikzset{nomorepostaction/.code=\let\tikz@postactions\pgfutil@empty}

\newcommand{\sembrack}[1]{[\![#1]\!]}
\newcommand{\fulbrack}[1]{\{\!|#1|\!\}}

\newtheorem{theorem}{Theorem}[section]
\newtheorem{definition}[theorem]{Definition}
\newtheorem{proposition}[theorem]{Proposition}

\begin{document}

\begin{titlepage}
\thispagestyle{empty}

\hrule
\begin{center}
{\bf\LARGE The Nature of Concurrency\\}
%
\vspace{0.5cm}
--- Yong Wang ---

\vspace{2cm}
\begin{figure}[!htbp]
 \centering
 \includegraphics[width=1.0\textwidth]{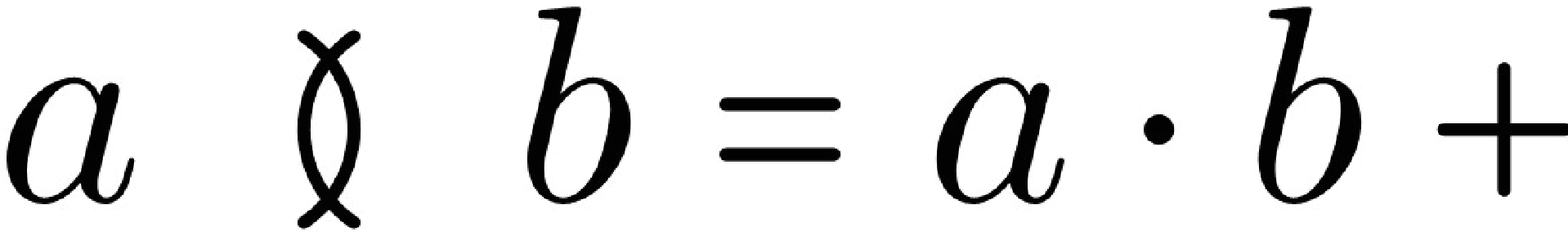}
\end{figure}

\end{center}
\end{titlepage}

\newpage 

\setcounter{page}{1}\pagenumbering{roman}

\tableofcontents

\newpage
\setcounter{page}{1}\pagenumbering{arabic}

        \input{section1.tex}
\newpage\input{section2.tex}
\newpage\input{section3.tex}
\newpage\input{section4.tex}
\newpage\input{section5.tex}
\bibliographystyle{elsarticle-num}
\newpage\bibliography{Refs-TNOC}

\end{document}

%% file: section1.tex
\section{Introduction}\label{intro}

There are mainly two kinds of models of concurrency \cite{MOC}: the models of interleaving concurrency and the models of true concurrency. Among the models of interleaving concurrency, the representatives are process algebras based on bisimilarity semantics, such as CCS \cite{CC} \cite{CCS} and ACP \cite{ACP}. And among the models of true concurrency, the representatives are event structure \cite{PED} \cite{ES} \cite{IES}, Petri net \cite{PN00} \cite{PN01} \cite{PN02} \cite{PN1} \cite{PN2} \cite{PN3}, and also automata and concurrent Kleene algebra \cite{CKA1} \cite{CKA2} \cite{CKA3} \cite{CKA4} \cite{CKA5} \cite{CKA6} \cite{CKA7}. The relationship between interleaving concurrency vs. true concurrency is not clarified, the main work is giving interleaving concurrency a semantics of true concurrency \cite{NISCCS} \cite{POSCCS} \cite{ESSCCS}.

As Chinese, we love "big" unification, i.e., the unification of interleaving concurrency vs. true concurrency. In concurrency theory, we refer to parallelism, denoted $a\parallel b$ for $a,b$ are atomic actions, which means that there are two parallel branches $a$ and $b$, they executed independently (without causality and confliction) and is captured exactly by the concurrency relation. But the whole thing, we prefer to use the word \emph{concurrency}, denoted $a\between b$, is that the actions in the two parallel branches may exist causalities or conflictions. The causalities between two parallel branches are usually not the sequence relation, but communications (the sending/receiving or writing/reading pairs). The conflictions can also be generalized to the ones between any actions in two parallel branches. Concurrency is made up of several parallel branches, in each branch which can be a model of concurrency, there exists communications or conflictions among these branches. This is well supported by computational systems in reality from the smaller ones to bigger ones: threads, cores, CPUs, processes, and communications and conflictions among them inner one computer system; distributed applications, communications via computer networks and distributed locks among them, constitute small or big scale distributed systems and the whole Internet.

Base on the above assumptions, we have done some work on the so-called truly concurrent process algebra CTC and APTC \cite{APTC} \cite{APTC2}, which are generalizations of CCS and ACP from interleaving concurrency to true concurrency. In this book, we deep the relationship between interleaving concurrency vs. true concurrency, especially, giving models of true concurrency, such as event structure, Petri net and concurrent Kleene algebra, (truly concurrent) process algebra foundations.

\subsection{Process Algebra vs. Event Structure}

Event structure \cite{PED} \cite{ES} \cite{IES} is a model of true concurrency. In an event structure, there are a set of atomic events and arbitrary causalities and conflictions among them and concurrency is implicitly defined. Based on the definition of an event structure, truly concurrent behaviours such as pomset bisimulation, step bisimulation, history-preserving (hp-) bisimulation and the finest hereditary history-preserving (hhp-) bisimulation \cite{HHPS1} \cite{HHPS2} can be introduced.

Since the relationship between process algebra vs. event structure (interleaving concurrency vs. true concurrency in nature) is not clarified before the introduction of truly concurrent process algebra \cite{APTC} \cite{APTC2}, the work on the relationship between process algebra vs. event structure usually gives traditional process algebra an event structure-based semantics, such as giving CCS a event structure-based semantics \cite{ESSCCS}. The work on giving event structure a process algebra-based foundation lacks.

In chapter \ref{pe}, we discuss the relationship between process algebra and event structure by establishing the relationship between prime event structures and processes based on the structurization of prime event structures and establishing structural operational semantics of prime event structure. We reproduce the truly concurrent process algebra APTC based on the structural operational semantics of prime event structure.

\subsection{Process Algebra vs. Petri Net}

Petri net \cite{PN00} \cite{PN01} \cite{PN02} \cite{PN1} \cite{PN2} \cite{PN3} is also a model of true concurrency. In a Petri net, there are two kinds of nodes: places (conditions) and transitions (actions), and causalities among them. On the relationship between process algebra and Petri net, one side is giving process algebra a Petri net semantics, the other side is giving Petri net a process algebra foundation \cite{BOXA1} \cite{BOXA2} \cite{PNA} \cite{PAPN}, among them, Petri net algebra \cite{PNA} gives Petri net a CCS-like foundation.

In chapter \ref{pp}, we discuss the relationship between process algebra and Petri net by establishing the relationship between Petri nets and processes based on the structurization of Petri nets. and establishing structural operational semantics of Petri net. We reproduce the guarded truly concurrent process algebra $APTC_G$ \cite{APTC} based on the structural operational semantics of Petri net.

\subsection{Process Algebra vs. Automata and Kleene Algebra}

Kleene algebra (KA) \cite{KA00} \cite{KA0} \cite{KA1} \cite{KA2} \cite{KA3} \cite{KA4} \cite{KA5} \cite{KA6} is an important algebraic structure with operators $+$, $\cdot$, $^*$, $0$ and $1$ to model computational properties of regular expressions. Kleene algebra can be used widely in computational areas, such as relational algebra, automata, algorithms, program logic and semantics, etc. A Kleene algebra of the family of regular set over a finite alphabet $\Sigma$ is called the algebra of regular events denoted $\mathbf{Reg}_{\Sigma}$, which was firstly studied as an open problem by Kleene \cite{KA00}. Then, Kleene algebra was widely studied and there existed several definitions on Kleene algebra \cite{KA0} \cite{KA1} \cite{KA2} \cite{KA3} \cite{KA4} \cite{KA5} for the almost same purpose of modelling regular expressions, and Kozen \cite{KA6} established the relationship among these definitions.

Then Kleene algebra has been extended in many ways to capture more computational properties, such as hypotheses \cite{EKA1} \cite{EKA2}, tests \cite{EKA3} \cite{EKA4} \cite{EKA5}, observations \cite{EKA6}, probabilistic KA \cite{EKA7}, etc. Among these extensions, a significant extension is concurrent KA (CKA) \cite{CKA1} \cite{CKA2} \cite{CKA3} \cite{CKA4} \cite{CKA5} \cite{CKA6} \cite{CKA7} and its extensions \cite{CKA70} \cite{CKA8} \cite{CKA9} \cite{CKA10} \cite{CKA11} \cite{CKA12} to capture the concurrent and parallel computations.

It is well-known that process algebras are theories to capture concurrent and parallel computations, for CCS \cite{CC} \cite{CCS} and ACP \cite{ACP} are with bisimilarity semantics. A natural question is that how automata theory is related to process algebra and how (concurrent) KA is related to process algebra? J. C. M. Baeten et al have done a lot of work on the relationship between automata theory and process algebra \cite{AP1} \cite{AP2} \cite{AP3} \cite{AP4} \cite{AP5} \cite{AP6}. It is essential of the work on introducing Kleene star into the process algebra based on bisimilarity semantics to answer this question, firstly initialized by Milner's proof system for regular expressions modulo bisimilarity (Mil) \cite{MF1}. Since the completeness of Milner's proof system remained open, some efforts were done, such as Redko's incompleteness proof for Klneene star modulo trace semantics \cite{MF10n}, completeness for BPA (basic process algebra) with Kleene star \cite{MF2} \cite{MF20}, work on ACP with iteration \cite{MF15} \cite{MF16} \cite{MF17}, completeness for prefix iteration \cite{MF10} \cite{MF11} \cite{MF12}, multi-exit iteration \cite{MF13}, flat iteration \cite{MF14}, 1-free regular expressions \cite{MF3} modulo bisimilarity. But these are not the full sense of regular expressions, most recently, Grabmayer \cite{MF4} \cite{MF5} \cite{MF6} \cite{MF7} has prepared to prove that Mil is complete with respect to a specific kind of process graphs called LLEE-1-charts which is equal to regular expressions, modulo the corresponding kind of bisimilarity called 1-bisimilarity.

But for concurrency and parallelism, the relationship between CKA and process algebra has remained open from Hoare \cite{CKA1} \cite{CKA2} to the recent work of CKA \cite{CKA7} \cite{CKA70}. Since most CKAs are based on the so-called true concurrency, we can draw the conclusion that the concurrency of CKA includes the interleaving one which the bisimilarity based process algebra captures, as the extended Milner's expansion law $a\parallel b=a\cdot b+b\cdot a+a\parallel b$ says, where $a,b$ are primitives (atomic actions), $\parallel$ is the parallel composition, $+$ is the alternative composition and $\cdot$ is the sequential composition with the background of computation. In contrast, Milner's expansion law is that $a\parallel b=a\cdot b+b\cdot a$ in bisimilarity based process algebras CCS and ACP.

Based on the work of truly concurrent process algebra APTC \cite{APTC} which is process algebra based on truly concurrent semantics, we can introduce Kleene star (and also parallel star) into APTC. Both for CKA with communications and APTC with Kleene star and parallel star, the extended Milner's expansion law $a\between b=a\cdot b+b\cdot a+a\parallel b +a\mid b$ with the concurrency operator $\between$ and communication merge $\mid$ holds. CKA and APTC are all the truly concurrent computation models, can have the same syntax (primitives and operators), the similar axiomatizations, and maybe have the same or different semantics. That's all.

The relationship between process algebra vs. automata and Kleene algebra is discussed in chapter \ref{pak}.

%% file: section2.tex
\section{Preliminaries}\label{pre} 

For self-satisfactory, in this chapter, we introduce the preliminaries on set, language, automata, equational logic, operational semantics, prime event structure and process algebras.

\subsection{Set, Language and Automata}

\begin{definition}[Set]
A set contains some objects, and let $\{-\}$ denote the contents of a set. For instance, $\mathbb{N}=\{1,2,3,\cdots\}$. Let $a\in A$ denote that $a$ is an element of the set $A$ and $a\notin A$ denote that $a$ is not an element of the set $A$. For all $a\in A$, if we can get $a\in B$, then we say that $A$ is a subset of $B$ denoted $A\subseteq B$. If $A\subseteq B$ and $B\subseteq A$, then $A=B$. We can define a new set by use of predicates on the existing sets, such that $\{n\in\mathbb{N}|\exists k\in\mathbb{N},n=2k\}$ for the set of even numbers. We can also specify a set to be the smallest set satisfy some inductive inference rules, for instance, we specify the set of even numbers $A$ satisfying the following rules:

$$\frac{}{0\in A}\quad\quad\frac{n\in A}{n+2\in A}$$
\end{definition}

\begin{definition}[Set composition]
The union of two sets $A$ and $B$, is denoted by $A\cup B=\{a|a\in A\textrm{ or }a\in B\}$, and the intersection of $A$ and $B$ by $A\cap B=\{a|a\in A\textrm{ and }a\in B\}$, the difference of $A$ and $B$ by $A\setminus B=\{a|a\in A\textrm{ and }a\notin B\}$. The empty set $\emptyset$ contains nothing. The set of all subsets of a set $A$ is called the powerset of $A$ denoted $2^A$.
\end{definition}

\begin{definition}[Tuple]
A tuple is a finite and ordered list of objects and denoted $\langle -\rangle$. For sets $A$ and $B$, the Cartesian product of $A$ and $B$ is denoted by $A\times B=\{\langle a,b\rangle|a\in A,b\in B\}$. $A^n$ is the $n$-fold Cartesian product of set $A$, for instance, $A^2=A\times A$. Tuples can be flattened as $A\times (B\times C)=(A\times B)\times C=A\times B\times C$ for sets $A$, $B$ and $C$.
\end{definition}

\begin{definition}[Relation]
A relation $R$ between sets $A$ and $B$ is a subset of $A\times B$, i.e., $R\subseteq A\times B$. We say that $R$ is a relation on set $A$ if $R$ is a relation between $A$ and itself, and,

\begin{itemize}
  \item $R$ is reflexive if for all $a\in A$, $aRa$ holds; it is irreflexive if for all $a\in A$, $aRa$ does not hold.
  \item $R$ is symmetric if for all $a,a'\in A$ with $aRa'$, then $a'Ra$ holds; it is antisymmetric if for all $a,a'\in A$ with $aRa'$ and $a'Ra$, then $a=a'$.
  \item $R$ is transitive if for all $a,a',a''\in A$ with $aRa'$ and $a'Ra''$, then $aRa''$ holds.
\end{itemize}
\end{definition}

\begin{definition}[Preorder, partial order, strict order]
If a relation $R$ is reflexive and transitive, we call that it is a preorder; When it is a preorder and antisymmetric, it is called a partial order, and a partially ordered set (poset) is a pair $\langle A,R\rangle$ with a set $A$ and a partial order $R$ on $A$; When it is irreflexive and transitive, it is called a strict order.
\end{definition}

\begin{definition}[Equivalence]
A relation $R$ is called an equivalence, if it is reflexive, symmetric and transitive. For an equivalent relation $R$ and a set $A$, $[a]_R=\{a'\in A|aRa'\}$ is called the equivalence class of $a\in A$.
\end{definition}

\begin{definition}[Relation composition]
For sets $A$, $B$ and $C$, and relations $R\subseteq A\times B$ and $R'\subseteq B\times C$, the relational composition denoted $R\circ R'$, is defined as the smallest relation $a(R\circ R')c$ satisfying $aRb$ and $bR'c$ with $a\in A$, $b\in B$ and $c\in C$. For a relation $R$ on set $A$, we denote $R^*$ for the reflexive and transitive closure of $R$, which is the least reflexive and transitive relation on $A$ that contains $R$.
\end{definition}

\begin{definition}[Function]
A function $f:A\rightarrow B$ from sets $A$ to $B$ is a relation between $A$ and $B$, i.e., for every $a\in A$, there exists one $b=f(a)\in B$, where $A$ is called the domain of $f$ and $B$ the codomain of $f$. $\sembrack{-}$ is also used as a function with $-$ a placeholder, i.e., $\sembrack{e}$ is the value of $\sembrack{-}$ for input $e$. A function $f$ is a bijection if for every $b\in B$, there exists exactly one $a\in A$ such that $b=f(a)$. For functions $f:A\rightarrow B$ and $g:B\rightarrow C$, the functional composition of $f$ and $g$ denoted $g\circ f$ such that $(g\circ f)(a)=g(f(a))$ for $a\in A$.
\end{definition}

\begin{definition}[Poset morphism]
For posets $\langle A,\leq\rangle$ and $\langle A',\leq'\rangle$ and function $f:A\rightarrow A'$, $f$ is called a poset morphism if for $a_0,a_1\in A$ with $a_0\leq a_1$, then $f(a_0)\leq' f(a_1)$ holds.
\end{definition}

\begin{definition}[Multiset]
A multiset is a kind of set of objects which may be repetitive denoted $\fulbrack{-}$, such that $\fulbrack{0,1,1}$ is significantly distinguishable from $\fulbrack{0,1}$.
\end{definition}

\begin{definition}[Alphabet, word, language]
An alphabet $\Sigma$ is a (maybe infinite) set of symbols. A word over some alphabet $\Sigma$ is a finite sequence of symbols from $\Sigma$. Words can be concatenated and the concatenation operator is denoted by $\cdot$, for instance $ab\cdot c=abc$. The empty word is denoted $1$ with $1\cdot w=w=w\cdot 1$ for word $w$. For $n\in\mathbb{N}$ and $a\in\Sigma$, $a^n$ is the $n$-fold concatenation of $a$ with $a^0=1$ and $a^{n+1}=a\cdot a^n$. A language if a set of words, and the language of all words over an alphabet $\Sigma$ is denoted $\Sigma^*$.
\end{definition}

\begin{definition}[Expressions]
Expressions are builded by function symbols and constants over a fixed alphabet inductively. For instance, the set of numerical expressions over some fixed set of variables $V$ are defined as the smallest set $E$ satisfying the following inference rules:

$$\frac{n\in\mathbb{N}}{n\in E}\quad \frac{v\in V}{v\in E}\quad \frac{e,f\in E}{e+f\in E}\quad \frac{e,f\in E}{e\times f\in E}\quad \frac{e\in E}{-e\in E}$$

The above inference rules are equal to the following Backus-Naur Form (BNF) grammer.

$$E\ni e,f::=n\in\mathbb{N}|v\in V|e+f|e\times f|-e$$
\end{definition}

\begin{definition}[Congruence, precongruence]
A relation $R$ on a set of expressions is a congruence if it is an equivalence compatible with the operators; and a relation $R$ on a set of expressions is a precongruence if it is a preorder compatible with the operators.
\end{definition}

\subsection{Equational Logic}\label{el}

We introduce some basic concepts about equational logic briefly, including signature, term, substitution, axiomatization, equality relation, model, term rewriting system, rewrite relation, normal form, termination, weak confluence and several conclusions. These concepts are coming from \cite{ACP}, and are introduced briefly as follows. About the details, please see \cite{ACP}.

\begin{definition}[Signature]
A signature $\Sigma$ consists of a finite set of function symbols (or operators) $f,g,\cdots$, where each function symbol $f$ has an arity $ar(f)$, being its number of arguments. A function symbol $a,b,c,\cdots$ of arity \emph{zero} is called a constant, a function symbol of arity one is called unary, and a function symbol of arity two is called binary.
\end{definition}

\begin{restatable}[Term]{definition}{term}
Let $\Sigma$ be a signature. The set $\mathbb{T}(\Sigma)$ of (open) terms $s,t,u,\cdots$ over $\Sigma$ is defined as the least set satisfying: (1)each variable is in $\mathbb{T}(\Sigma)$; (2) if $f\in \Sigma$ and $t_1,\cdots,t_{ar(f)}\in\mathbb{T}(\Sigma)$, then $f(t_1,\cdots,t_{ar(f)}\in\mathbb{T}(\Sigma))$. A term is closed if it does not contain free variables. The set of closed terms is denoted by $\mathcal{T}(\Sigma)$.
\end{restatable}

\begin{definition}[Substitution]
Let $\Sigma$ be a signature. A substitution is a mapping $\sigma$ from variables to the set $\mathbb{T}(\Sigma)$ of open terms. A substitution extends to a mapping from open terms to open terms: the term $\sigma(t)$ is obtained by replacing occurrences of variables $x$ in t by $\sigma(x)$. A substitution $\sigma$ is closed if $\sigma(x)\in\mathcal{T}(\Sigma)$ for all variables $x$.
\end{definition}

\begin{definition}[Axiomatization]
An axiomatization over a signature $\Sigma$ is a finite set of equations, called axioms, of the form $s=t$ with $s,t\in\mathbb{T}(\Sigma)$.
\end{definition}

\begin{definition}[Equality relation]
An axiomatization over a signature $\Sigma$ induces a binary equality relation $=$ on $\mathbb{T}(\Sigma)$ as follows. (1)(Substitution) If $s=t$ is an axiom and $\sigma$ a substitution, then $\sigma(s)=\sigma(t)$. (2)(Equivalence) The relation $=$ is closed under reflexivity, symmetry, and transitivity. (3)(Context) The relation $=$ is closed under contexts: if $t=u$ and $f$ is a function symbol with $ar(f)>0$, then $f(s_1,\cdots,s_{i-1},t,s_{i+1},\cdots,s_{ar(f)})=f(s_1,\cdots,s_{i-1},u,s_{i+1},\cdots,s_{ar(f)})$.
\end{definition}

\begin{definition}[Model]
Assume an axiomatization $\mathcal{E}$ over a signature $\Sigma$, which induces an equality relation $=$. A model for $\mathcal{E}$ consists of a set $\mathcal{M}$ together with a mapping $\phi: \mathcal{T}(\Sigma)\rightarrow\mathcal{M}$. (1) $(\mathcal{M},\phi)$ is sound for $\mathcal{E}$ if $s=t$ implies $\phi(s)\equiv\phi(t)$ for $s,t\in\mathcal{T}(\Sigma)$; (2) $(\mathcal{M},\phi)$ is complete for $\mathcal{E}$ if $\phi(s)\equiv\phi(t)$ implies $s=t$ for $s,t\in\mathcal{T}(\Sigma)$.
\end{definition}

\begin{definition}[Term rewriting system]
Assume a signature $\Sigma$. A rewrite rule is an expression $s\rightarrow t$ with $s,t\in\mathbb{T}(\Sigma)$, where: (1) the left-hand side $s$ is not a single variable; (2) all variables that occur at the right-hand side $t$ also occur in the left-hand side $s$. A term rewriting system (TRS) is a finite set of rewrite rules.
\end{definition}

\begin{definition}[Rewrite relation]
A TRS over a signature $\Sigma$ induces a one-step rewrite relation $\rightarrow$ on $\mathbb{T}(\Sigma)$ as follows. (1) (Substitution) If $s\rightarrow t$ is a rewrite rule and $\sigma$ a substitution, then $\sigma(s)\rightarrow\sigma(t)$. (2) (Context) The relation $\rightarrow$ is closed under contexts: if $t\rightarrow u$ and f is a function symbol with $ar(f)>0$, then $f(s_1,\cdots,s_{i-1},t,s_{i+1},\cdots,s_{ar(f)})\rightarrow f(s_1,\cdots,s_{i-1},u,s_{i+1},\cdots,s_{ar(f)})$. The rewrite relation $\rightarrow^*$ is the reflexive transitive closure of the one-step rewrite relation $\rightarrow$: (1) if $s\rightarrow t$, then $s\rightarrow^* t$; (2) $t\rightarrow^* t$; (3) if $s\rightarrow^* t$ and $t\rightarrow^* u$, then $s\rightarrow^* u$.
\end{definition}

\begin{definition}[Normal form]
A term is called a normal form for a TRS if it cannot be reduced by any of the rewrite rules.
\end{definition}

\begin{definition}[Termination]
A TRS is terminating if it does not induce infinite reductions $t_0\rightarrow t_1\rightarrow t_2\rightarrow \cdots$.
\end{definition}

\begin{definition}[Weak confluence]
A TRS is weakly confluent if for each pair of one-step reductions $s\rightarrow t_1$ and $s\rightarrow t_2$, there is a term $u$ such that $t_1\rightarrow^* u$ and $t_2\rightarrow^* u$.
\end{definition}

\begin{theorem}[Newman's lemma]
If a TRS is terminating and weakly confluent, then it reduces each term to a unique normal form.
\end{theorem}

\begin{definition}[Commutativity and associativity]
Assume an axiomatization $\mathcal{E}$. A binary function symbol $f$ is commutative if $\mathcal{E}$ contains an axiom $f(x,y)=f(y,x)$ and associative if $\mathcal{E}$ contains an axiom $f(f(x,y),z)=f(x,f(y,z))$.
\end{definition}

\begin{definition}[Convergence]
A pair of terms $s$ and $t$ is said to be convergent if there exists a term $u$ such that $s\rightarrow^* u$ and $t\rightarrow^* u$.
\end{definition}

Axiomatizations can give rise to TRSs that are not weakly confluent, which can be remedied by Knuth-Bendix completion \cite{KB}. It determines overlaps in left hand sides of rewrite rules, and introduces extra rewrite rules to join the resulting right hand sides, witch are called critical pairs.

\begin{theorem}[Weak confluence]
A TRS is weakly confluent if and only if all its critical pairs are convergent.
\end{theorem}

\begin{definition}[Elimination property]
Let a process algebra with a defined set of basic terms as a subset of the set of closed terms over the process algebra. Then the process algebra has the elimination to basic terms property if for every closed term $s$ of the algebra, there exists a basic term $t$ of the algebra such that the algebra$\vdash s=t$.
\end{definition}

\begin{definition}[Strongly normalizing]
A term $s_0$ is called strongly normalizing if does not an infinite series of reductions beginning in $s_0$.
\end{definition}

\begin{definition}
We write $s>_{lpo} t$ if $s\rightarrow^+ t$ where $\rightarrow^+$ is the transitive closure of the reduction relation defined by the transition rules of an algebra.
\end{definition}

\begin{theorem}[Strong normalization]\label{SN}
Let a term rewriting system (TRS) with finitely many rewriting rules and let $>$ be a well-founded ordering on the signature of the corresponding algebra. If $s>_{lpo} t$ for each rewriting rule $s\rightarrow t$ in the TRS, then the term rewriting system is strongly normalizing.
\end{theorem}

\subsection{Operational Semantics}\label{os}

We assume a non-empty set $S$ of states, a finite, non-empty set of transition labels $A$ and a finite set of predicate symbols.

\begin{definition}[Labeled transition system] 
A transition is a triple $(s,a,s')$ with $a\in A$, or a pair (s, P) with $P$ a predicate, where $s,s'\in S$. A labeled transition system (LTS) is possibly infinite set of transitions. An LTS is finitely branching if each of its states has only finitely many outgoing transitions.
\end{definition}

\begin{definition}[Transition system specification] 
A transition rule $\rho$ is an expression of the form $\frac{H}{\pi}$, with $H$ a set of expressions $t\xrightarrow{a}t'$ and $tP$ with $t,t'\in \mathbb{T}(\Sigma)$, called the (positive) premises of $\rho$, and $\pi$ an expression $t\xrightarrow{a}t'$ or $tP$ with $t,t'\in \mathbb{T}(\Sigma)$, called the conclusion of $\rho$. The left-hand side of $\pi$ is called the source of $\rho$. A transition rule is closed if it does not contain any variables. A transition system specification (TSS) is a (possible infinite) set of transition rules.
\end{definition}

\begin{restatable}[Process (graph)]{definition}{processgraph}
A process (graph) $p$ is an LTS in which one state $s$ is elected to be the root. If the LTS contains a transition $s\xrightarrow{a} s'$, then $p\xrightarrow{a} p'$ where $p'$ has root state $s'$. Moreover, if the LTS contains a transition $sP$, then $pP$. (1) A process $p_0$ is finite if there are only finitely many sequences $p_0\xrightarrow{a_1}p_1\xrightarrow{a_2}\cdots\xrightarrow{a_k} p_k$. (2) A process $p_0$ is regular if there are only finitely many processes $p_k$ such that $p_0\xrightarrow{a_1}p_1\xrightarrow{a_2}\cdots\xrightarrow{a_k} p_k$.
\end{restatable}

\begin{definition}[Bisimulation]
A bisimulation relation $R$ is a binary relation on processes such that: (1) if $p R q$ and $p\xrightarrow{a}p'$ then $q\xrightarrow{a}q'$ with $p' R q'$; (2) if $p R q$ and $q\xrightarrow{a}q'$ then $p\xrightarrow{a}p'$ with $p' R q'$; (3) if $p R q$ and $pP$, then $qP$; (4) if $p R q$ and $qP$, then $pP$. Two processes $p$ and $q$ are bisimilar, denoted by $p\sim_{HM} q$, if there is a bisimulation relation $R$ such that $p R q$. Note that $p,p',q,q'$ are processes, $a$ is an atomic action, and $P$ is a predicate.
\end{definition}

\begin{definition}[Congruence]
Let $\Sigma$ be a signature. An equivalence relation $R$ on $\mathcal{T}(\Sigma)$ is a congruence if for each $f\in\Sigma$, if $s_i R t_i$ for $i\in\{1,\cdots,ar(f)\}$, then $f(s_1,\cdots,s_{ar(f)}) R f(t_1,\cdots,t_{ar(f)})$.
\end{definition}

\begin{definition}[Branching bisimulation]
A branching bisimulation relation $R$ is a binary relation on the collection of processes such that: (1) if $p R q$ and $p\xrightarrow{a}p'$ then either $a\equiv \tau$ and $p' R q$ or there is a sequence of (zero or more) $\tau$-transitions $q\xrightarrow{\tau}\cdots\xrightarrow{\tau}q_0$ such that $p R q_0$ and $q_0\xrightarrow{a}q'$ with $p' R q'$; (2) if $p R q$ and $q\xrightarrow{a}q'$ then either $a\equiv \tau$ and $p R q'$ or there is a sequence of (zero or more) $\tau$-transitions $p\xrightarrow{\tau}\cdots\xrightarrow{\tau}p_0$ such that $p_0 R q$ and $p_0\xrightarrow{a}p'$ with $p' R q'$; (3) if $p R q$ and $pP$, then there is a sequence of (zero or more) $\tau$-transitions $q\xrightarrow{\tau}\cdots\xrightarrow{\tau}q_0$ such that $p R q_0$ and $q_0P$; (4) if $p R q$ and $qP$, then there is a sequence of (zero or more) $\tau$-transitions $p\xrightarrow{\tau}\cdots\xrightarrow{\tau}p_0$ such that $p_0 R q$ and $p_0P$. Two processes $p$ and $q$ are branching bisimilar, denoted by $p\approx_{bHM} q$, if there is a branching bisimulation relation $R$ such that $p R q$.
\end{definition}

\begin{definition}[Rooted branching bisimulation]
A rooted branching bisimulation relation $R$ is a binary relation on processes such that: (1) if $p R q$ and $p\xrightarrow{a}p'$ then $q\xrightarrow{a}q'$ with $p'\approx_{bHM} q'$; (2) if $p R q$ and $q\xrightarrow{a}q'$ then $p\xrightarrow{a}p'$ with $p'\approx_{bHM} q'$; (3) if $p R q$ and $pP$, then $qP$; (4) if $p R q$ and $qP$, then $pP$. Two processes $p$ and $q$ are rooted branching bisimilar, denoted by $p\approx_{rbHM} q$, if there is a rooted branching bisimulation relation $R$ such that $p R q$.
\end{definition}

\begin{definition}[Conservative extension]
Let $T_0$ and $T_1$ be TSSs (transition system specifications) over signatures $\Sigma_0$ and $\Sigma_1$, respectively. The TSS $T_0\oplus T_1$ is a conservative extension of $T_0$ if the LTSs (labeled transition systems) generated by $T_0$ and $T_0\oplus T_1$ contain exactly the same transitions $t\xrightarrow{a}t'$ and $tP$ with $t\in \mathcal{T}(\Sigma_0)$.
\end{definition}

\begin{definition}[Source-dependency]
The source-dependent variables in a transition rule of $\rho$ are defined inductively as follows: (1) all variables in the source of $\rho$ are source-dependent; (2) if
$t\xrightarrow{a}t'$ is a premise of $\rho$ and all variables in $t$ are source-dependent, then all variables in $t'$ are source-dependent. A transition rule is source-dependent if
all its variables are. A TSS is source-dependent if all its rules are.
\end{definition}

\begin{definition}[Freshness]
Let $T_0$ and $T_1$ be TSSs over signatures $\Sigma_0$ and $\Sigma_1$, respectively. A term in $\mathbb{T}(T_0\oplus T_1)$ is said to be fresh if it contains a function symbol from
$\Sigma_1\setminus\Sigma_0$. Similarly, a transition label or predicate symbol in $T_1$ is fresh if it does not occur in $T_0$.
\end{definition}

\begin{theorem}[Conservative extension]
Let $T_0$ and $T_1$ be TSSs over signatures $\Sigma_0$ and $\Sigma_1$, respectively, where $T_0$ and $T_0\oplus T_1$ are positive after reduction. Under the following conditions,
$T_0\oplus T_1$ is a conservative extension of $T_0$. (1) $T_0$ is source-dependent. (2) For each $\rho\in T_1$, either the source of $\rho$ is fresh, or $\rho$ has a premise of the
form $t\xrightarrow{a}t'$ or $tP$, where $t\in \mathbb{T}(\Sigma_0)$, all variables in $t$ occur in the source of $\rho$ and $t'$, and $a$ or $P$ is fresh.
\end{theorem}

\subsection{Process Algebras} 

In this subsection, we introduce the preliminaries on process algebra $CCS$, $ACP$, which are based on the interleaving bisimulation semantics.

A crucial initial observation that is at the heart of the notion of process algebra is due to Milner, who noticed that concurrent processes have an algebraic structure. $CCS$ \cite{CC} \cite{CCS} is a calculus of concurrent systems. It includes syntax and semantics:

\begin{enumerate}
  \item Its syntax includes actions, process constant, and operators acting between actions, like Prefix, Summation, Composition, Restriction, Relabelling.
  \item Its semantics is based on labeled transition systems, Prefix, Summation, Composition, Restriction, Relabelling have their transition rules. $CCS$ has good semantic properties based on the interleaving bisimulation. These properties include monoid laws, static laws, expansion law for strongly interleaving bisimulation, $\tau$ laws for weakly interleaving bisimulation, and full congruences for strongly and weakly interleaving bisimulations, and also unique solution for recursion.
\end{enumerate}

$CCS$ can be used widely in verification of computer systems with an interleaving concurrent flavor.

$ACP$ captures several computational properties in the form of algebraic laws, and proves the soundness and completeness modulo bisimulation/rooted branching bisimulation equivalences. These computational properties are organized in a modular way by use of the concept of conservational extension, which include the following modules, note that, every algebra are composed of constants and operators, the constants are the computational objects, while operators capture the computational properties.

\begin{enumerate}
  \item \textbf{$BPA$ (Basic Process Algebras)}. $BPA$ has sequential composition $\cdot$ and alternative composition $+$ to capture sequential computation and nondeterminacy. The constants are ranged over $A$, the set of atomic actions. The algebraic laws on $\cdot$ and $+$ are sound and complete modulo bisimulation equivalence.
  \item \textbf{$ACP$ (Algebra of Communicating Processes)}. $ACP$ uses the parallel operator $\parallel$, the auxiliary binary left merge $\leftmerge$ to model parallelism, and the communication merge $\mid$ to model communications among different parallel branches. Since a communication may be blocked, a new constant called deadlock $\delta$ is extended to $A$, and also a new unary encapsulation operator $\partial_H$ is introduced to eliminate $\delta$, which may exist in the processes. The algebraic laws on these operators are also sound and complete modulo bisimulation equivalence. Note that, these operators in a process can be eliminated by deductions on the process using axioms of $ACP$, and eventually be steadied by $\cdot$ and $+$, this is also why bisimulation is called an \emph{interleaving} semantics.
  \item \textbf{Recursion}. To model infinite computation, recursion is introduced into $ACP$. In order to obtain a sound and complete theory, guarded recursion and linear recursion are needed. The corresponding axioms are $RSP$ (Recursive Specification Principle) and $RDP$ (Recursive Definition Principle), $RDP$ says the solutions of a recursive specification can represent the behaviors of the specification, while $RSP$ says that a guarded recursive specification has only one solution, they are sound with respect to $ACP$ with guarded recursion modulo bisimulation equivalence, and they are complete with respect to $ACP$ with linear recursion modulo bisimulation equivalence.
  \item \textbf{Abstraction}. To abstract away internal implementations from the external behaviors, a new constant $\tau$ called silent step is added to $A$, and also a new unary abstraction operator $\tau_I$ is used to rename actions in $I$ into $\tau$ (the resulted $ACP$ with silent step and abstraction operator is called $ACP_{\tau}$). The recursive specification is adapted to guarded linear recursion to prevent infinite $\tau$-loops specifically. The axioms for $\tau$ and $\tau_I$ are sound modulo rooted branching bisimulation equivalence (a kind of weak bisimulation equivalence). To eliminate infinite $\tau$-loops caused by $\tau_I$ and obtain the completeness, $CFAR$ (Cluster Fair Abstraction Rule) is used to prevent infinite $\tau$-loops in a constructible way.
\end{enumerate}

$ACP$ can be used to verify the correctness of system behaviors, by deduction on the description of the system using the axioms of $ACP$. Base on the modularity of $ACP$, it can be extended easily and elegantly. For more details, please refer to the book of $ACP$ \cite{ACP}.

\subsection{Prime Event Structures}\label{pes2}

The definition of prime event structures are coming from \cite{PED} \cite{ES} \cite{IES}, and truly concurrent bisimulations are coming from \cite{HHPS1} \cite{HHPS2}.

\begin{definition}[Prime event structure with silent event]\label{pes}
Let $\Lambda$ be a fixed set of labels, ranged over $a,b,c,\cdots$ and $\tau$. A ($\Lambda$-labelled) prime event structure with silent event $\tau$ is a tuple $\mathcal{E}=\langle \mathbb{E}, \leq, \sharp, \lambda\rangle$, where $\mathbb{E}$ is a denumerable set of events, including the silent event $\tau$. Let $\hat{\mathbb{E}}=\mathbb{E}\backslash\{\tau\}$, exactly excluding $\tau$, it is obvious that $\hat{\tau^*}=\epsilon$, where $\epsilon$ is the empty event. Let $\lambda:\mathbb{E}\rightarrow\Lambda$ be a labelling function and let $\lambda(\tau)=\tau$. And $\leq$, $\sharp$ are binary relations on $\mathbb{E}$, called causality and conflict respectively, such that:

\begin{enumerate}
  \item $\leq$ is a partial order and $\lceil a \rceil = \{a'\in \mathbb{E}|a'\leq a\}$ is finite for all $a\in \mathbb{E}$. It is easy to see that $a\leq\tau^*\leq a'=a\leq\tau\leq\cdots\leq\tau\leq a'$, then $a\leq a'$.
  \item $\sharp$ is irreflexive, symmetric and hereditary with respect to $\leq$, that is, for all $a,a',a''\in \mathbb{E}$, if $a\sharp a'\leq a''$, then $a\sharp a''$.
\end{enumerate}

Then, the concepts of consistency and concurrency can be drawn from the above definition:

\begin{enumerate}
  \item $a,a'\in \mathbb{E}$ are consistent, denoted as $a\frown a'$, if $\neg(a\sharp a')$. A subset $X\subseteq \mathbb{E}$ is called consistent, if $a\frown a'$ for all $a,a'\in X$.
  \item $a,a'\in \mathbb{E}$ are concurrent, denoted as $a\parallel a'$, if $\neg(a\leq a')$, $\neg(a'\leq a)$, and $\neg(a\sharp a')$.
\end{enumerate}
\end{definition}

\begin{definition}[Configuration]
Let $\mathcal{E}$ be a PES. A (finite) configuration in $\mathcal{E}$ is a (finite) consistent subset of events $C\subseteq \mathcal{E}$, closed with respect to causality (i.e. $\lceil C\rceil=C$). The set of finite configurations of $\mathcal{E}$ is denoted by $\mathcal{C}(\mathcal{E})$. We let $\hat{C}=C\backslash\{\tau\}$.
\end{definition}

A consistent subset of $X\subseteq \mathbb{E}$ of events can be seen as a pomset. Given $X, Y\subseteq \mathbb{E}$, $\hat{X}\sim \hat{Y}$ if $\hat{X}$ and $\hat{Y}$ are isomorphic as pomsets. In the following of the paper, we say $C_1\sim C_2$, we mean $\hat{C_1}\sim\hat{C_2}$.

\begin{definition}[Pomset transitions and step]
Let $\mathcal{E}$ be a PES and let $C\in\mathcal{C}(\mathcal{E})$, and $\emptyset\neq X\subseteq \mathbb{E}$, if $C\cap X=\emptyset$ and $C'=C\cup X\in\mathcal{C}(\mathcal{E})$, then $C\xrightarrow{X} C'$ is called a pomset transition from $C$ to $C'$. When the events in $X$ are pairwise concurrent, we say that $C\xrightarrow{X}C'$ is a step.
\end{definition}

\begin{definition}[Weak pomset transitions and weak step]
Let $\mathcal{E}$ be a PES and let $C\in\mathcal{C}(\mathcal{E})$, and $\emptyset\neq X\subseteq \hat{\mathbb{E}}$, if $C\cap X=\emptyset$ and $\hat{C'}=\hat{C}\cup X\in\mathcal{C}(\mathcal{E})$, then $C\xRightarrow{X} C'$ is called a weak pomset transition from $C$ to $C'$, where we define $\xRightarrow{a}\triangleq\xrightarrow{\tau^*}\xrightarrow{a}\xrightarrow{\tau^*}$. And $\xRightarrow{X}\triangleq\xrightarrow{\tau^*}\xrightarrow{a}\xrightarrow{\tau^*}$, for every $a\in X$. When the events in $X$ are pairwise concurrent, we say that $C\xRightarrow{X}C'$ is a weak step.
\end{definition}

\begin{definition}[Pomset, step bisimulation]\label{psb}
Let $\mathcal{E}_1$, $\mathcal{E}_2$ be PESs. A pomset bisimulation is a relation $R\subseteq\mathcal{C}(\mathcal{E}_1)\times\mathcal{C}(\mathcal{E}_2)$, such that if $(C_1,C_2)\in R$, and $C_1\xrightarrow{X_1}C_1'$ then $C_2\xrightarrow{X_2}C_2'$, with $X_1\subseteq \mathbb{E}_1$, $X_2\subseteq \mathbb{E}_2$, $X_1\sim X_2$ and $(C_1',C_2')\in R$, and vice-versa. We say that $\mathcal{E}_1$, $\mathcal{E}_2$ are pomset bisimilar, written $\mathcal{E}_1\sim_p\mathcal{E}_2$, if there exists a pomset bisimulation $R$, such that $(\emptyset,\emptyset)\in R$. By replacing pomset transitions with steps, we can get the definition of step bisimulation. When PESs $\mathcal{E}_1$ and $\mathcal{E}_2$ are step bisimilar, we write $\mathcal{E}_1\sim_s\mathcal{E}_2$.
\end{definition}

\begin{definition}[Weak pomset, step bisimulation]\label{wpsb}
Let $\mathcal{E}_1$, $\mathcal{E}_2$ be PESs. A weak pomset bisimulation is a relation $R\subseteq\mathcal{C}(\mathcal{E}_1)\times\mathcal{C}(\mathcal{E}_2)$, such that if $(C_1,C_2)\in R$, and $C_1\xRightarrow{X_1}C_1'$ then $C_2\xRightarrow{X_2}C_2'$, with $X_1\subseteq \hat{\mathbb{E}_1}$, $X_2\subseteq \hat{\mathbb{E}_2}$, $X_1\sim X_2$ and $(C_1',C_2')\in R$, and vice-versa. We say that $\mathcal{E}_1$, $\mathcal{E}_2$ are weak pomset bisimilar, written $\mathcal{E}_1\approx_p\mathcal{E}_2$, if there exists a weak pomset bisimulation $R$, such that $(\emptyset,\emptyset)\in R$. By replacing weak pomset transitions with weak steps, we can get the definition of weak step bisimulation. When PESs $\mathcal{E}_1$ and $\mathcal{E}_2$ are weak step bisimilar, we write $\mathcal{E}_1\approx_s\mathcal{E}_2$.
\end{definition}

\begin{definition}[Posetal product]
Given two PESs $\mathcal{E}_1$, $\mathcal{E}_2$, the posetal product of their configurations, denoted $\mathcal{C}(\mathcal{E}_1)\overline{\times}\mathcal{C}(\mathcal{E}_2)$, is defined as

$$\{(C_1,f,C_2)|C_1\in\mathcal{C}(\mathcal{E}_1),C_2\in\mathcal{C}(\mathcal{E}_2),f:C_1\rightarrow C_2 \textrm{ isomorphism}\}.$$

A subset $R\subseteq\mathcal{C}(\mathcal{E}_1)\overline{\times}\mathcal{C}(\mathcal{E}_2)$ is called a posetal relation. We say that $R$ is downward closed when for any $(C_1,f,C_2),(C_1',f',C_2')\in \mathcal{C}(\mathcal{E}_1)\overline{\times}\mathcal{C}(\mathcal{E}_2)$, if $(C_1,f,C_2)\subseteq (C_1',f',C_2')$ pointwise and $(C_1',f',C_2')\in R$, then $(C_1,f,C_2)\in R$.

For $f:X_1\rightarrow X_2$, we define $f[x_1\mapsto x_2]:X_1\cup\{x_1\}\rightarrow X_2\cup\{x_2\}$, $z\in X_1\cup\{x_1\}$,(1)$f[x_1\mapsto x_2](z)=
x_2$,if $z=x_1$;(2)$f[x_1\mapsto x_2](z)=f(z)$, otherwise. Where $X_1\subseteq \mathbb{E}_1$, $X_2\subseteq \mathbb{E}_2$, $x_1\in \mathbb{E}_1$, $x_2\in \mathbb{E}_2$.
\end{definition}

\begin{definition}[Weakly posetal product]
Given two PESs $\mathcal{E}_1$, $\mathcal{E}_2$, the weakly posetal product of their configurations, denoted $\mathcal{C}(\mathcal{E}_1)\overline{\times}\mathcal{C}(\mathcal{E}_2)$, is defined as

$$\{(C_1,f,C_2)|C_1\in\mathcal{C}(\mathcal{E}_1),C_2\in\mathcal{C}(\mathcal{E}_2),f:\hat{C_1}\rightarrow \hat{C_2} \textrm{ isomorphism}\}.$$

A subset $R\subseteq\mathcal{C}(\mathcal{E}_1)\overline{\times}\mathcal{C}(\mathcal{E}_2)$ is called a weakly posetal relation. We say that $R$ is downward closed when for any $(C_1,f,C_2),(C_1',f,C_2')\in \mathcal{C}(\mathcal{E}_1)\overline{\times}\mathcal{C}(\mathcal{E}_2)$, if $(C_1,f,C_2)\subseteq (C_1',f',C_2')$ pointwise and $(C_1',f',C_2')\in R$, then $(C_1,f,C_2)\in R$.

For $f:X_1\rightarrow X_2$, we define $f[x_1\mapsto x_2]:X_1\cup\{x_1\}\rightarrow X_2\cup\{x_2\}$, $z\in X_1\cup\{x_1\}$,(1)$f[x_1\mapsto x_2](z)=
x_2$,if $z=x_1$;(2)$f[x_1\mapsto x_2](z)=f(z)$, otherwise. Where $X_1\subseteq \hat{\mathbb{E}_1}$, $X_2\subseteq \hat{\mathbb{E}_2}$, $x_1\in \hat{\mathbb{E}}_1$, $x_2\in \hat{\mathbb{E}}_2$. Also, we define $f(\tau^*)=f(\tau^*)$.
\end{definition}

\begin{definition}[(Hereditary) history-preserving bisimulation]
A history-preserving (hp-) bisimulation is a posetal relation $R\subseteq\mathcal{C}(\mathcal{E}_1)\overline{\times}\mathcal{C}(\mathcal{E}_2)$ such that if $(C_1,f,C_2)\in R$, and $C_1\xrightarrow{a_1} C_1'$, then $C_2\xrightarrow{a_2} C_2'$, with $(C_1',f[a_1\mapsto a_2],C_2')\in R$, and vice-versa. $\mathcal{E}_1,\mathcal{E}_2$ are history-preserving (hp-)bisimilar and are written $\mathcal{E}_1\sim_{hp}\mathcal{E}_2$ if there exists a hp-bisimulation $R$ such that $(\emptyset,\emptyset,\emptyset)\in R$.

A hereditary history-preserving (hhp-)bisimulation is a downward closed hp-bisimulation. $\mathcal{E}_1,\mathcal{E}_2$ are hereditary history-preserving (hhp-)bisimilar and are written $\mathcal{E}_1\sim_{hhp}\mathcal{E}_2$.
\end{definition}

\begin{definition}[Weak (hereditary) history-preserving bisimulation]
A weak history-preserving (hp-) bisimulation is a weakly posetal relation $R\subseteq\mathcal{C}(\mathcal{E}_1)\overline{\times}\mathcal{C}(\mathcal{E}_2)$ such that if $(C_1,f,C_2)\in R$, and $C_1\xRightarrow{a_1} C_1'$, then $C_2\xRightarrow{a_2} C_2'$, with $(C_1',f[a_1\mapsto a_2],C_2')\in R$, and vice-versa. $\mathcal{E}_1,\mathcal{E}_2$ are weak history-preserving (hp-)bisimilar and are written $\mathcal{E}_1\approx_{hp}\mathcal{E}_2$ if there exists a weak hp-bisimulation $R$ such that $(\emptyset,\emptyset,\emptyset)\in R$.

A weakly hereditary history-preserving (hhp-)bisimulation is a downward closed weak hp-bisimulation. $\mathcal{E}_1,\mathcal{E}_2$ are weakly hereditary history-preserving (hhp-)bisimilar and are written $\mathcal{E}_1\approx_{hhp}\mathcal{E}_2$.
\end{definition}

\begin{definition}[Branching pomset, step bisimulation]
Assume a special termination predicate $\downarrow$, and let $\surd$ represent a state with $\surd\downarrow$. Let $\mathcal{E}_1$, $\mathcal{E}_2$ be PESs. A branching pomset bisimulation is a relation $R\subseteq\mathcal{C}(\mathcal{E}_1)\times\mathcal{C}(\mathcal{E}_2)$, such that:
 \begin{enumerate}
   \item if $(C_1,C_2)\in R$, and $C_1\xrightarrow{X}C_1'$ then
   \begin{itemize}
     \item either $X\equiv \tau^*$, and $(C_1',C_2)\in R$;
     \item or there is a sequence of (zero or more) $\tau$-transitions $C_2\xrightarrow{\tau^*} C_2^0$, such that $(C_1,C_2^0)\in R$ and $C_2^0\xRightarrow{X}C_2'$ with $(C_1',C_2')\in R$;
   \end{itemize}
   \item if $(C_1,C_2)\in R$, and $C_2\xrightarrow{X}C_2'$ then
   \begin{itemize}
     \item either $X\equiv \tau^*$, and $(C_1,C_2')\in R$;
     \item or there is a sequence of (zero or more) $\tau$-transitions $C_1\xrightarrow{\tau^*} C_1^0$, such that $(C_1^0,C_2)\in R$ and $C_1^0\xRightarrow{X}C_1'$ with $(C_1',C_2')\in R$;
   \end{itemize}
   \item if $(C_1,C_2)\in R$ and $C_1\downarrow$, then there is a sequence of (zero or more) $\tau$-transitions $C_2\xrightarrow{\tau^*}C_2^0$ such that $(C_1,C_2^0)\in R$ and $C_2^0\downarrow$;
   \item if $(C_1,C_2)\in R$ and $C_2\downarrow$, then there is a sequence of (zero or more) $\tau$-transitions $C_1\xrightarrow{\tau^*}C_1^0$ such that $(C_1^0,C_2)\in R$ and $C_1^0\downarrow$.
 \end{enumerate}

We say that $\mathcal{E}_1$, $\mathcal{E}_2$ are branching pomset bisimilar, written $\mathcal{E}_1\approx_{bp}\mathcal{E}_2$, if there exists a branching pomset bisimulation $R$, such that $(\emptyset,\emptyset)\in R$.

By replacing pomset transitions with steps, we can get the definition of branching step bisimulation. When PESs $\mathcal{E}_1$ and $\mathcal{E}_2$ are branching step bisimilar, we write $\mathcal{E}_1\approx_{bs}\mathcal{E}_2$.
\end{definition}

\begin{definition}[Rooted branching pomset, step bisimulation]
Assume a special termination predicate $\downarrow$, and let $\surd$ represent a state with $\surd\downarrow$. Let $\mathcal{E}_1$, $\mathcal{E}_2$ be PESs. A rooted branching pomset bisimulation is a relation $R\subseteq\mathcal{C}(\mathcal{E}_1)\times\mathcal{C}(\mathcal{E}_2)$, such that:
 \begin{enumerate}
   \item if $(C_1,C_2)\in R$, and $C_1\xrightarrow{X}C_1'$ then $C_2\xrightarrow{X}C_2'$ with $C_1'\approx_{bp}C_2'$;
   \item if $(C_1,C_2)\in R$, and $C_2\xrightarrow{X}C_2'$ then $C_1\xrightarrow{X}C_1'$ with $C_1'\approx_{bp}C_2'$;
   \item if $(C_1,C_2)\in R$ and $C_1\downarrow$, then $C_2\downarrow$;
   \item if $(C_1,C_2)\in R$ and $C_2\downarrow$, then $C_1\downarrow$.
 \end{enumerate}

We say that $\mathcal{E}_1$, $\mathcal{E}_2$ are rooted branching pomset bisimilar, written $\mathcal{E}_1\approx_{rbp}\mathcal{E}_2$, if there exists a rooted branching pomset bisimulation $R$, such that $(\emptyset,\emptyset)\in R$.

By replacing pomset transitions with steps, we can get the definition of rooted branching step bisimulation. When PESs $\mathcal{E}_1$ and $\mathcal{E}_2$ are rooted branching step bisimilar, we write $\mathcal{E}_1\approx_{rbs}\mathcal{E}_2$.
\end{definition}

\begin{definition}[Branching (hereditary) history-preserving bisimulation]
Assume a special termination predicate $\downarrow$, and let $\surd$ represent a state with $\surd\downarrow$. A branching history-preserving (hp-) bisimulation is a weakly posetal relation $R\subseteq\mathcal{C}(\mathcal{E}_1)\overline{\times}\mathcal{C}(\mathcal{E}_2)$ such that:

 \begin{enumerate}
   \item if $(C_1,f,C_2)\in R$, and $C_1\xrightarrow{a_1}C_1'$ then
   \begin{itemize}
     \item either $a_1\equiv \tau$, and $(C_1',f[a_1\mapsto \tau^{a_1}],C_2)\in R$;
     \item or there is a sequence of (zero or more) $\tau$-transitions $C_2\xrightarrow{\tau^*} C_2^0$, such that $(C_1,f,C_2^0)\in R$ and $C_2^0\xrightarrow{a_2}C_2'$ with $(C_1',f[a_1\mapsto a_2],C_2')\in R$;
   \end{itemize}
   \item if $(C_1,f,C_2)\in R$, and $C_2\xrightarrow{a_2}C_2'$ then
   \begin{itemize}
     \item either $a_2\equiv \tau$, and $(C_1,f[a_2\mapsto \tau^{a_2}],C_2')\in R$;
     \item or there is a sequence of (zero or more) $\tau$-transitions $C_1\xrightarrow{\tau^*} C_1^0$, such that $(C_1^0,f,C_2)\in R$ and $C_1^0\xrightarrow{a_1}C_1'$ with $(C_1',f[a_2\mapsto a_1],C_2')\in R$;
   \end{itemize}
   \item if $(C_1,f,C_2)\in R$ and $C_1\downarrow$, then there is a sequence of (zero or more) $\tau$-transitions $C_2\xrightarrow{\tau^*}C_2^0$ such that $(C_1,f,C_2^0)\in R$ and $C_2^0\downarrow$;
   \item if $(C_1,f,C_2)\in R$ and $C_2\downarrow$, then there is a sequence of (zero or more) $\tau$-transitions $C_1\xrightarrow{\tau^*}C_1^0$ such that $(C_1^0,f,C_2)\in R$ and $C_1^0\downarrow$.
 \end{enumerate}

$\mathcal{E}_1,\mathcal{E}_2$ are branching history-preserving (hp-)bisimilar and are written $\mathcal{E}_1\approx_{bhp}\mathcal{E}_2$ if there exists a branching hp-bisimulation $R$ such that $(\emptyset,\emptyset,\emptyset)\in R$.

A branching hereditary history-preserving (hhp-)bisimulation is a downward closed branching hp-bisimulation. $\mathcal{E}_1,\mathcal{E}_2$ are branching hereditary history-preserving (hhp-)bisimilar and are written $\mathcal{E}_1\approx_{bhhp}\mathcal{E}_2$.
\end{definition}

\begin{definition}[Rooted branching (hereditary) history-preserving bisimulation]
Assume a special termination predicate $\downarrow$, and let $\surd$ represent a state with $\surd\downarrow$. A rooted branching history-preserving (hp-) bisimulation is a weakly posetal relation $R\subseteq\mathcal{C}(\mathcal{E}_1)\overline{\times}\mathcal{C}(\mathcal{E}_2)$ such that:

 \begin{enumerate}
   \item if $(C_1,f,C_2)\in R$, and $C_1\xrightarrow{a_1}C_1'$, then $C_2\xrightarrow{a_2}C_2'$ with $C_1'\approx_{bhp}C_2'$;
   \item if $(C_1,f,C_2)\in R$, and $C_2\xrightarrow{a_2}C_2'$, then $C_1\xrightarrow{a_1}C_1'$ with $C_1'\approx_{bhp}C_2'$;
   \item if $(C_1,f,C_2)\in R$ and $C_1\downarrow$, then $C_2\downarrow$;
   \item if $(C_1,f,C_2)\in R$ and $C_2\downarrow$, then $C_1\downarrow$.
 \end{enumerate}

$\mathcal{E}_1,\mathcal{E}_2$ are rooted branching history-preserving (hp-)bisimilar and are written $\mathcal{E}_1\approx_{rbhp}\mathcal{E}_2$ if there exists a rooted branching hp-bisimulation $R$ such that $(\emptyset,\emptyset,\emptyset)\in R$.

A rooted branching hereditary history-preserving (hhp-)bisimulation is a downward closed rooted branching hp-bisimulation. $\mathcal{E}_1,\mathcal{E}_2$ are rooted branching hereditary history-preserving (hhp-)bisimilar and are written $\mathcal{E}_1\approx_{rbhhp}\mathcal{E}_2$.
\end{definition}

\subsection{Truly Concurrent Process Algebras}\label{tcpa}

CTC \cite{APTC} is a calculus of truly concurrent systems. It includes syntax and semantics:

\begin{enumerate}
  \item Its syntax includes actions, process constant, and operators acting between actions, like Prefix, Summation, Composition, Restriction, Relabelling.
  \item Its semantics is based on labeled transition systems, Prefix, Summation, Composition, Restriction, Relabelling have their transition rules. CTC has good semantic properties based on the truly concurrent bisimulations. These properties include monoid laws, static laws, new expansion law for strongly truly concurrent bisimulations, $\tau$ laws for weakly truly concurrent bisimulations, and full congruences for strongly and weakly truly concurrent bisimulations, and also unique solution for recursion.
\end{enumerate}

$APTC$ \cite{APTC} captures several computational properties in the form of algebraic laws, and proves the soundness and completeness modulo truly concurrent bisimulation/rooted branching truly concurrent bisimulation equivalence. These computational properties are organized in a modular way by use of the concept of conservational extension, which include the following modules, note that, every algebra are composed of constants and operators, the constants are the computational objects, while operators capture the computational properties.

\begin{enumerate}
  \item \textbf{$BATC$ (Basic Algebras for True Concurrency)}. $BATC$ has sequential composition $\cdot$ and alternative composition $+$ to capture causality computation and conflict. The constants are ranged over $\mathbb{E}$, the set of atomic events. The algebraic laws on $\cdot$ and $+$ are sound and complete modulo truly concurrent bisimulation equivalences, such as pomset bisimulation $\sim_p$, step bisimulation $\sim_s$, history-preserving (hp-) bisimulation $\sim_{hp}$ and hereditary history-preserving (hhp-) bisimulation $\sim_{hhp}$.
  \item \textbf{$APTC$ (Algebra of Parallelism for True Concurrency)}. $APTC$ uses the whole parallel operator $\between$, the parallel operator $\parallel$ to model parallelism, and the communication merge $\mid$ to model causality (communication) among different parallel branches. Since a communication may be blocked, a new constant called deadlock $\delta$ is extended to $\mathbb{E}$, and also a new unary encapsulation operator $\partial_H$ is introduced to eliminate $\delta$, which may exist in the processes. And also a conflict elimination operator $\Theta$ to eliminate conflicts existing in different parallel branches. The algebraic laws on these operators are also sound and complete modulo truly concurrent bisimulation equivalences, such as pomset bisimulation $\sim_p$, step bisimulation $\sim_s$, history-preserving (hp-) bisimulation $\sim_{hp}$. Note that, these operators in a process except the parallel operator $\parallel$ can be eliminated by deductions on the process using axioms of $APTC$, and eventually be steadied by $\cdot$, $+$ and $\parallel$, this is also why bisimulations are called an \emph{truly concurrent} semantics.
  \item \textbf{Recursion}. To model infinite computation, recursion is introduced into $APTC$. In order to obtain a sound and complete theory, guarded recursion and linear recursion are needed. The corresponding axioms are $RSP$ (Recursive Specification Principle) and $RDP$ (Recursive Definition Principle), $RDP$ says the solutions of a recursive specification can represent the behaviors of the specification, while $RSP$ says that a guarded recursive specification has only one solution, they are sound with respect to $APTC$ with guarded recursion modulo truly concurrent bisimulation equivalences, such as pomset bisimulation $\sim_p$, step bisimulation $\sim_s$, history-preserving (hp-) bisimulation $\sim_{hp}$, and they are complete with respect to $APTC$ with linear recursion modulo truly concurrent bisimulation equivalence, such as pomset bisimulation $\sim_p$, step bisimulation $\sim_s$, history-preserving (hp-) bisimulation $\sim_{hp}$.
  \item \textbf{Abstraction}. To abstract away internal implementations from the external behaviors, a new constant $\tau$ called silent step is added to $\mathbb{E}$, and also a new unary abstraction operator $\tau_I$ is used to rename actions in $I$ into $\tau$ (the resulted $APTC$ with silent step and abstraction operator is called $APTC_{\tau}$). The recursive specification is adapted to guarded linear recursion to prevent infinite $\tau$-loops specifically. The axioms for $\tau$ and $\tau_I$ are sound modulo rooted branching truly concurrent bisimulation equivalences (a kind of weak truly concurrent bisimulation equivalence), such as rooted branching pomset bisimulation $\approx_p$, rooted branching step bisimulation $\approx_s$, rooted branching history-preserving (hp-) bisimulation $\approx_{hp}$. To eliminate infinite $\tau$-loops caused by $\tau_I$ and obtain the completeness, $CFAR$ (Cluster Fair Abstraction Rule) is used to prevent infinite $\tau$-loops in a constructible way.
\end{enumerate} 

CTC and APTC can be used widely in verification of computer systems with a truly concurrent flavor.

%% file: section3.tex
\section{Process Algebra vs. Event Structure}\label{pe}

In this chapter, we discuss the relationship between process algebra and event structure. Firstly, we establish the relationship between prime event structures and processes in section \ref{pesap}, based on the structurization of prime event structures. Then, we establish structural operational semantics of prime event structure in section \ref{sospes}. Finally, we reproduce the truly concurrent process algebra APTC based on the structural operational semantics of prime event structure, in section \ref{sospespa}.

\subsection{Prime Event Structures as Processes}\label{pesap}

Firstly, let we recall the basic definitions of term and process.

\term*

\processgraph*

Usually, a process graph describes the execution of a (closed) term. The closed term is elected as the initial root state of the process graph, then it executes some atomic actions and eventually terminates successfully. The definition of the initial term is following the structural way, but, in the definition of a PES (Definition \ref{pes}), there may exist unstructured causalities or conflicts. So, before treating a PES as a process, we must structurize the PES.

PESs in Definition \ref{pes} can be composed together into a bigger PES, and a PES can or cannot be decomposed into several smaller PESs. To compose PESs or decompose a PES, the first problem is to define the relations between PESs. We can see that there are only two kinds of relations called causality and confliction among the events (actions) in the definition of a PES, and the other two relations called consistency and concurrency are implicit. From a background of language and computation (and also Kleene algebra), especially the success of structured (sequential) programming \cite{GOTO} \cite{SP}, it is well-known that the three basic relations called sequence $\cdot$, choice $+$ and recursion are necessary and sufficient in sequential computation and a sequential program can be structured (we will make it rigorous in the following). The fresh thing is concurrency, and it must be added as a basic relation among atomic events (see the following analyses), so we adopt the four relations of causality, confliction, concurrency and recursion as the basic relations to compose PESs or decompose a PES. Since recursions are expressed by recursive equations, while recursive equations are mixtures of recursive variables and atomic actions by sequence $\cdot$, choice $+$ and parallelism $\parallel$ to form terms, so, causality, confliction and concurrency are three fundamental relations to compose or decompose PESs.

Let us analyze the three basic relations named causality, confliction, and concurrency. Firstly, sequence is a kind of causality, choice is a special kind of confliction (confliction between the beginning actions of different branches) in sequential computation, but they are not all the things in concurrent computation. We refer to parallelism, denoted $a\parallel b$ for $a,b\in\mathbb{E}$, which means that there are two parallel branches $a$ and $b$, they executed independently (without causality and confliction) and is captured exactly by the concurrency relation in the definition of PES. But the whole thing defined by PES, we prefer to use the word \emph{concurrency}, denoted $a\between b$, is that the actions in the two parallel branches may exist causalities or conflictions. The causalities between two parallel branches are usually not the sequence relation, but communications (the sending/receiving or writing/reading pairs). The conflictions can also be generalized to the ones between any actions in two parallel branches. Concurrency defined by PES is made up of several parallel branches, in each branch which can be a PES, there exists communications or conflictions among these branches. 

In concurrency theory, causality can be classified finely into sequence and communication, and choice can be generalized to confliction. 

Some PESs can be composed into a bigger PES in sequence, in choice, in parallel, and in concurrency.

\begin{definition}[PES composition in sequence]
Let PESs $\mathcal{E}_1=\langle \mathbb{E}_1, \leq_1, \sharp_1, \Lambda_1\rangle$, $\mathcal{E}_2=\langle \mathbb{E}_2, \leq_2, \sharp_2, \Lambda_2\rangle$ with $\mathbb{E}_i$, $\leq_i$, $\sharp_i$ and $\Lambda_i$ for $i\in\{1,2\}$ being the corresponding set of events, set of causality relations, set of confliction relations and set of labels respectively of PES $\mathcal{E}_i$ for $i\in\{1,2\}$ (with a little abuse of symbols), we write $\mathcal{E}_1\cdot\mathcal{E}_2=\langle \mathbb{E}_{\mathcal{E}_1\cdot\mathcal{E}_2}, \leq_{\mathcal{E}_1\cdot\mathcal{E}_2}, \sharp_{\mathcal{E}_1\cdot\mathcal{E}_2}, \Lambda_{\mathcal{E}_1\cdot\mathcal{E}_2}\rangle$ for the sequential composition of $\mathcal{E}_1$ and $\mathcal{E}_2$, where
$$\mathbb{E}_{\mathcal{E}_1\cdot\mathcal{E}_2}=\mathbb{E}_1\cup\mathbb{E}_2\quad \leq_{\mathcal{E}_1\cdot\mathcal{E}_2}=\leq_1\cup\leq_2\cup(\mathbb{E}_1\times\mathbb{E}_2) \quad\sharp_{\mathcal{E}_1\cdot\mathcal{E}_2}=\sharp_1\cup\sharp_2 \quad\Lambda_{\mathcal{E}_1\cdot\mathcal{E}_2}=\Lambda_1\cup\Lambda_2$$

where $\mathbb{E}_1\times\mathbb{E}_2$ is the Cartesian product of $\mathbb{E}_1$ and $\mathbb{E}_2$.
\end{definition}

\begin{definition}[PES composition in choice]
Let PESs $\mathcal{E}_1=\langle \mathbb{E}_1, \leq_1, \sharp_1, \Lambda_1\rangle$, $\mathcal{E}_2=\langle \mathbb{E}_2, \leq_2, \sharp_2, \Lambda_2\rangle$ with $\mathbb{E}_i$, $\leq_i$, $\sharp_i$ and $\Lambda_i$ for $i\in\{1,2\}$ being the corresponding set of events, set of causality relations, set of confliction relations and set of labels respectively of PES $\mathcal{E}_i$ for $i\in\{1,2\}$ (with a little abuse of symbols), we write $\mathcal{E}_1+\mathcal{E}_2=\langle \mathbb{E}_{\mathcal{E}_1+\mathcal{E}_2}, \leq_{\mathcal{E}_1+\mathcal{E}_2}, \sharp_{\mathcal{E}_1+\mathcal{E}_2}, \Lambda_{\mathcal{E}_1+\mathcal{E}_2}\rangle$ for the alternative composition of $\mathcal{E}_1$ and $\mathcal{E}_2$, where
$$\mathbb{E}_{\mathcal{E}_1+\mathcal{E}_2}=\mathbb{E}_1\cup\mathbb{E}_2\quad \leq_{\mathcal{E}_1+\mathcal{E}_2}=\leq_1\cup\leq_2 \quad\sharp_{\mathcal{E}_1+\mathcal{E}_2}=\sharp_1\cup\sharp_2\cup(\mathbb{E}_1\times\mathbb{E}_2) \quad\Lambda_{\mathcal{E}_1+\mathcal{E}_2}=\Lambda_1\cup\Lambda_2$$

with $\mathbb{E}_1\times\mathbb{E}_2$ is the Cartesian product of $\mathbb{E}_1$ and $\mathbb{E}_2$.
\end{definition}

\begin{definition}[PES composition in parallel]
Let PESs $\mathcal{E}_1=\langle \mathbb{E}_1, \leq_1, \sharp_1, \Lambda_1\rangle$, $\mathcal{E}_2=\langle \mathbb{E}_2, \leq_2, \sharp_2, \Lambda_2\rangle$ with $\mathbb{E}_i$, $\leq_i$, $\sharp_i$ and $\Lambda_i$ for $i\in\{1,2\}$ being the corresponding set of events, set of causality relations, set of confliction relations and set of labels respectively of PES $\mathcal{E}_i$ for $i\in\{1,2\}$ (with a little abuse of symbols), we write $\mathcal{E}_1\parallel\mathcal{E}_2=\langle \mathbb{E}_{\mathcal{E}_1\parallel\mathcal{E}_2}, \leq_{\mathcal{E}_1\parallel\mathcal{E}_2}, \sharp_{\mathcal{E}_1\parallel\mathcal{E}_2}, \Lambda_{\mathcal{E}_1\parallel\mathcal{E}_2}\rangle$ for the parallel composition of $\mathcal{E}_1$ and $\mathcal{E}_2$, where
$$\mathbb{E}_{\mathcal{E}_1\parallel\mathcal{E}_2}=\mathbb{E}_1\cup\mathbb{E}_2\quad \leq_{\mathcal{E}_1\parallel\mathcal{E}_2}=\leq_1\cup\leq_2 \quad\sharp_{\mathcal{E}_1\parallel\mathcal{E}_2}=\sharp_1\cup\sharp_2 \quad\Lambda_{\mathcal{E}_1\parallel\mathcal{E}_2}=\Lambda_1\cup\Lambda_2$$
\end{definition}

\begin{definition}[PES composition in concurrency]
Let PESs $\mathcal{E}_1=\langle \mathbb{E}_1, \leq_1, \sharp_1, \Lambda_1\rangle$, $\mathcal{E}_2=\langle \mathbb{E}_2, \leq_2, \sharp_2, \Lambda_2\rangle$ with $\mathbb{E}_i$, $\leq_i$, $\sharp_i$ and $\Lambda_i$ for $i\in\{1,2\}$ being the corresponding set of events, set of causality relations, set of confliction relations and set of labels respectively of PES $\mathcal{E}_i$ for $i\in\{1,2\}$ (with a little abuse of symbols), we write $\mathcal{E}_1\between\mathcal{E}_2=\langle \mathbb{E}_{\mathcal{E}_1\between\mathcal{E}_2}, \leq_{\mathcal{E}_1\between\mathcal{E}_2}, \sharp_{\mathcal{E}_1\between\mathcal{E}_2}, \Lambda_{\mathcal{E}_1\between\mathcal{E}_2}\rangle$ for the concurrent composition of $\mathcal{E}_1$ and $\mathcal{E}_2$, where
$$\mathbb{E}_{\mathcal{E}_1\between\mathcal{E}_2}=\mathbb{E}_1\cup\mathbb{E}_2\quad \leq_{\mathcal{E}_1\between\mathcal{E}_2}=\leq_1\cup\leq_2\cup\leq_{1,2} \quad\sharp_{\mathcal{E}_1\between\mathcal{E}_2}=\sharp_1\cup\sharp_2\cup\sharp_{1,2} \quad\Lambda_{\mathcal{E}_1\between\mathcal{E}_2}=\Lambda_1\cup\Lambda_2$$

where $\leq_{1,2}$ and $\sharp_{1,2}$ are the newly added causality and confliction relations among actions in $\mathbb{E}_1$ and $\mathbb{E}_2$, which are unstructured.
\end{definition}

Note that concurrent composition of PES is the common sense composition pattern, and other compositions are all special cases of concurrent composition, such that sequential composition is a concurrent composition with newly added causalities from the ending actions (actions without outgoing causalities) of the first PES to the beginning actions (actions without incoming causalities) of the second PES, alternative composition is a concurrent composition with newly added conflictions between the beginning actions of the two PESs, parallel composition is a concurrent composition without newly added causalities and conflictions.

Then we discuss the decomposition of a PES. We say that a PES is structured, we mean that a PES can be decomposed into several sub-PESs (the sub-PESs can composed into the original PES by the above composition patterns) without unstructured causalities and conflictions among them. Structured PES can capture the above meanings inductively.

\begin{definition}[Structured PES]
A Structured PES $\mathcal{SE}$ which is a PES $\mathcal{E}=\langle \mathbb{E}, \leq, \sharp, \lambda\rangle$, is inductively defined as follows:

\begin{enumerate}
  \item $\mathbb{E}\subset \mathcal{SE}$;
  \item If $\mathcal{E}_1$ is an $\mathcal{SE}$ and $\mathcal{E}_2$ is an $\mathcal{SE}$, then $\mathcal{E}_1\cdot\mathcal{E}_2$ is an $\mathcal{SE}$;
  \item If $\mathcal{E}_1$ is an $\mathcal{SE}$ and $\mathcal{E}_2$ is an $\mathcal{SE}$, then $\mathcal{E}_1+\mathcal{E}_2$ is an $\mathcal{SE}$;
  \item If $\mathcal{E}_1$ is an $\mathcal{SE}$ and $\mathcal{E}_2$ is an $\mathcal{SE}$, then $\mathcal{E}_1\parallel\mathcal{E}_2$ is an $\mathcal{SE}$.
\end{enumerate}
\end{definition}

Actually, a PES defines an unstructured graph with a truly concurrent flavor and can not be structured usually.

\begin{figure}[h]
  \centering
  \includegraphics{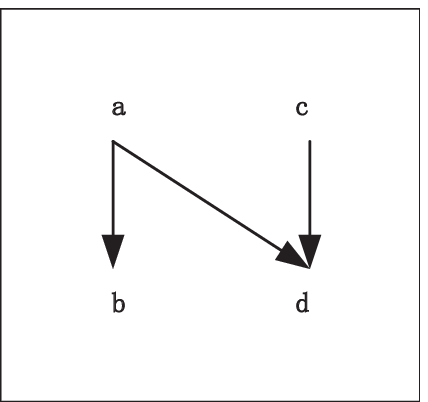}
  \caption{N-shape}
  \label{nshape}
\end{figure}

\begin{definition}[N-shape]
A PES is an N-shape, if it has at least four events in $\mathbb{E}$ with labels such as $a$, $b$, $c$ and $d$, and causality relations $a\leq b$, $c\leq d$ and $a\leq d$, as illustrated in Fig. \ref{nshape}.
\end{definition}

\begin{proposition}[Structurization of N-shape]
An N-shape can not be structurized.
\end{proposition}

\begin{proof}
$a$ and $c$ are in parallel, $b$ after $a$, so $b$ and $a$ are in the same parallel branch; $d$ after $c$, so $d$ and $c$ are in the same parallel branch; so $a$ and $d$ are in different parallel branches. But, $d$ after $a$ means that $d$ and $a$ are in the same parallel branch. This causes contradictions.
\end{proof}

\begin{figure}[h]
  \centering
  \includegraphics{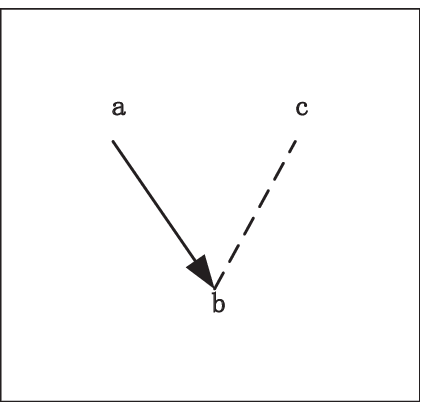}
  \caption{V-shape}
  \label{vshape}
\end{figure}

\begin{definition}[V-shape]
A PES is a V-shape, if it has at least three events in $\mathbb{E}$ with labels such as $a$, $b$, $c$, and causality relation $a\leq b$, confliction $b\sharp c$, as illustrated in Fig. \ref{vshape}.
\end{definition}

To be moore intuitive, in Fig. \ref{conflictions}, we show another example of unstructured conflictions between parallel branches called H-shape. There are two parallel branches, one is $a_1\cdot a_2\cdot a_3$ and the other is $a_4\cdot a_5\cdot a_6$, and there is an unstructured confliction between event $a_2$ and $a_5$ denoted $a_2\sharp a_5$, which is illustrated by the dashed line between $a_2$ and $a_5$.

\begin{figure}[h]
  \centering
  \includegraphics{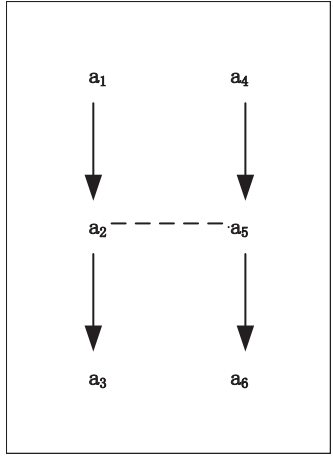}
  \caption{An example of unstructured conflictions called H-shape}
  \label{conflictions}
\end{figure}

Through the above analyses on the composition of PESs, it is reasonable to assume that a PES is composed by parallel branches and the unstructured causalities and conflictions always exist between actions in different parallel branches, usually the unstructured causalities are communications. Now, let us discuss the structurization of PES. Firstly, we only consider the synchronous communications. In a synchronous communication, two atomic event pair $a,b$ shakes hands denoted $a\mid b$ and merges as a communication action $\gamma(a,b)$ if the communication exists, otherwise, will cause deadlock (the special constant 0). 

$$a\mid b=\begin{cases}\gamma(a,b),& \textrm{if }\gamma(a,b)\textrm{ is defined;} \\ 0, & \textrm{otherwise.}\end{cases}$$

As Fig. \ref{snshape}-a) illustrated, the unstructured causalities are usually synchronous communications which is denoted by different arrows with respect to sequential composition. Fig. \ref{snshape}-b) shows that the communicating action pair merges and Fig. \ref{snshape} illustrates the structured PES denoted $c\cdot \gamma(a,d)\cdot b$ after elimination of unstructured communications.

\begin{figure}[h]
  \centering
  \includegraphics{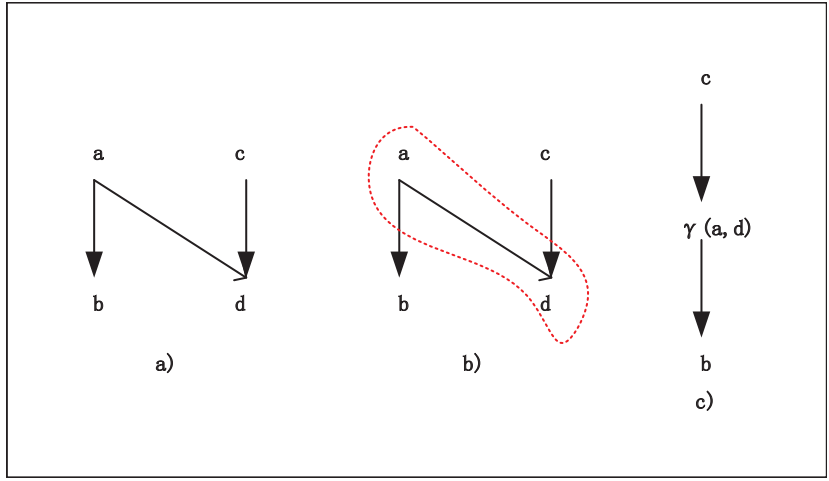}
  \caption{Structurization of N-shape}
  \label{snshape}
\end{figure}

It is the turn of structurization of unstructured conflictions between actions in different parallel branches called V-shape in Fig. \ref{svshape}. The PES $(a\cdot b)\parallel c$ with $b\sharp c$ is equal to the PES $a\cdot b+a\parallel c$ in Fig. \ref{svshape} modulo the truly concurrent bisimulation equivalence relations in section \ref{pes2}, we introduce a unitary operator $\Theta$ and an auxiliary binary operator $\triangleleft$ to eliminate the unstructured conflictions, that is, 

$$\Theta((a_1\cdot a_2\cdot a_3)\parallel (a_4\cdot a_5\cdot a_6),(a_2\sharp a_5))=(a_1\cdot a_2\cdot a_3)\parallel a_4+a_1\parallel(a_4\cdot a_5\cdot a_6)$$

In the first summand of PES in Fig. \ref{svshape} $a\cdot b$, the action $c$ is renamed to empty event (the special constant 1), and in the second one $a\parallel c$, the action $b$ is renamed to empty event.

\begin{figure}[h]
  \centering
  \includegraphics{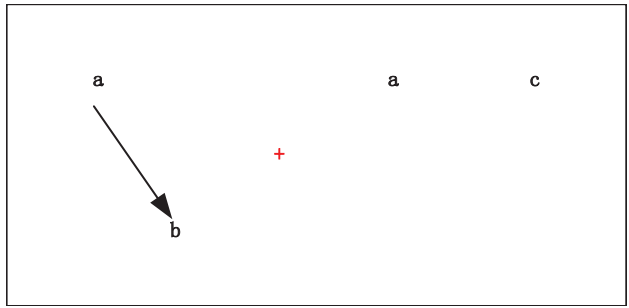}
  \caption{Structurization of V-shape}
  \label{svshape}
\end{figure}

To be more intuitive, let us structurize the unstructured conflictions between actions in different parallel branches called H-shape in Fig. \ref{conflictions}. The PES $(a_1\cdot a_2\cdot a_3)\parallel (a_4\cdot a_5\cdot a_6)$ with $a_2\sharp a_5$ is equal to the PES $(a_1\cdot a_2\cdot a_3)\parallel a_4+a_1\parallel(a_4\cdot a_5\cdot a_6)$ in Fig. \ref{sconflictions} modulo the truly concurrent bisimulation equivalence relations in section \ref{pes2}, that is, 

$$\Theta((a_1\cdot a_2\cdot a_3)\parallel (a_4\cdot a_5\cdot a_6),(a_2\sharp a_5))=(a_1\cdot a_2\cdot a_3)\parallel a_4+a_1\parallel(a_4\cdot a_5\cdot a_6)$$

In the first summand of PES in Fig. \ref{sconflictions} $(a_1\cdot a_2\cdot a_3)\parallel a_4$, the actions $a_5$ and $a_6$ are renamed to empty events (the special constant 1), and in the second one $a_1\parallel(a_4\cdot a_5\cdot a_6)$, the actions $a_2$ and $a_3$ are renamed to empty events.

\begin{figure}[h]
  \centering
  \includegraphics{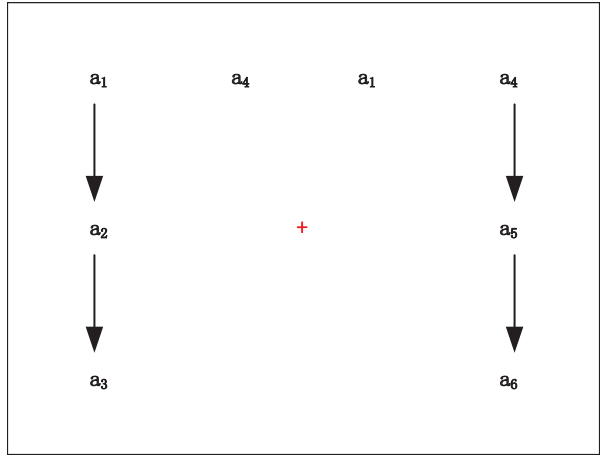}
  \caption{Structurization of H-shape}
  \label{sconflictions}
\end{figure}

Thus, with the assumption of parallel branches, and unstructured causalities and conflictions among them, by use of the elimination methods in Fig. \ref{snshape} and \ref{sconflictions}, a PES $\mathcal{E}$ can be structurizated to a structured one $\mathcal{SE}$. Actually, the unstructured PESs and their corresponding structured ones are equivalent modulo truly concurrent bisimulation equivalences, such as pomset bisimulation $\sim_p$, step bisimulation $\sim_s$, hp-bisimulation $\sim_{hp}$ and hhp-bisimulation $\sim_{hhp}$ in section \ref{pes2}. A PES can be expressed by a term of atomic actions (including 0 and 1), binary operators $\cdot$, $+$, $\between$, $\parallel$, $\mid$, the unary operator $\Theta$ and the auxiliary binary $\triangleleft$. This means that a PES can be treated as a process, so, in the following, we will not distinguish atomic actions and atomic events, PESs and processes.

\subsection{Structural Operational Semantics of Prime Event Structure}\label{sospes}

In this section, we use some concepts of structural operational semantics, and we do not refer to the original reference, please refer to \cite{SOS} for details.

By replacing the single atomic action $a$ to a pomset of atomic actions $a_1,\cdots, a_n$ for $n\in\mathbb{N}$, we adapt the concepts of Labelled Transition System (LTS) and Transition System Specification (TSS) in section \ref{os} to be able to capture pomset transitions. 

\begin{definition}[Stratification]
A stratification for a TSS is a weight function $\phi$ which maps transitions to ordinal numbers, such that for each transition rule $\rho$ with conclusion $\pi$ and for each closed substitution $\sigma$:

\begin{enumerate}
  \item for positive premises $t\xrightarrow{\{a_1,\cdots,a_n\}}t'$ and $tP$ of $\rho$, $\phi(\sigma(t)\xrightarrow{\{a_1,\cdots,a_n\}}\sigma(t'))\leq\phi(\sigma(\pi))$ and $\phi(\sigma(t)P)\leq\phi(\sigma(\pi))$, respectively;
  \item for negative premises $t\xnrightarrow{\{a_1,\cdots,a_n\}}$ and $t\neg P$ of $\rho$, $\phi(\sigma(t)\xrightarrow{\{a_1,\cdots,a_n\}}t')<\phi(\sigma(\pi))$ for all closed terms $t'$ and $\phi(\sigma(t)P)<\phi(\sigma(\pi))$, respectively.
\end{enumerate}
\end{definition}

\begin{theorem}[Positivity after reduction]
If a TSS allows a stratification, then it is positive after reduction.
\end{theorem}

\subsubsection{Truly Concurrent Bisimulations as Congruences}

\begin{definition}[Panth format]
A transition rule $\rho$ is in panth format if it satisfies the following three restrictions:

\begin{enumerate}
  \item for each positive premise $t\xrightarrow{\{a_1,\cdots,a_n\}} t'$ of $\rho$, the right-hand side $t'$ is a single variable;
  \item the source of $\rho$ contains no more than one function symbols;
  \item there are no multiple occurrences of the same variable at the right-hand sides of positive premises and in the source of $\rho$.
\end{enumerate}

A TSS is in panth format if it consists of panth rules only.
\end{definition}

\begin{theorem}[Truly concurrent bisimulations as congruences]\label{tcbac}
If a TSS is positive after reduction and in panth format, then the truly concurrent bisimulation equivalences, including pomset bisimulation equivalence $\sim_p$, step bisimulation equivalence $\sim_s$, hp-bisimulation equivalence $\sim_{hp}$ and hhp-bisimulation equivalence $\sim_{hhp}$, that it induces are all congruences.
\end{theorem}

\subsubsection{Branching Truly Concurrent Bisimulations as Congruences}

\begin{definition}[Lookahead]
A transition rule contains lookahead if a variable occurs at the left-hand side of a premise and at the right-hand side of a premise of this rule.
\end{definition}

\begin{definition}[Patience rule]
A patience rule for the i-th argument of a function symbol $f$ is a panth rule of the form

$$\frac{x_i\xrightarrow{\tau}y}{f(x_1,\cdots,x_{ar(f)})\xrightarrow{\tau}f(x_1,\cdots,x_{i-1},y,x_{i+1},\cdots,x_{ar(f)})}$$
\end{definition}

\begin{definition}[RBB cool format]
A TSS $T$ is in RBB cool format if the following restrictions are satisfied:

\begin{enumerate}
  \item it consists of panth rules that do not contain lookahead;
  \item suppose a function symbol $f$ occurs at the right-hand side of the conclusion of some transition rule in $T$. Let $\rho\in T$ be a non-patience rule with source $f(x_1,\cdots,x_{ar(f)})$, then for $i\in\{1,\cdots,ar(f)\}$, $x_i$ occurs in no more than one premise of $\rho$, where this premise is of the form $x_i P$ or $x_i\xrightarrow{\{a_1,\cdots,a_n\}}y$ with $a_1\nequiv\tau,\cdots,a_n\nequiv\tau$. Moreover, if there is such a premise in $\rho$, then there is a patience rule for the i-th argument of $f$ in $T$.
\end{enumerate}
\end{definition}

\begin{theorem}[Rooted branching truly concurrent bisimulations as congruences]\label{rbtcbac}
If a TSS is positive after reduction and in RBB cool format, then the rooted branching truly concurrent bisimulation equivalences, including rooted branching pomset bisimulation equivalence $\approx_{rbp}$, rooted branching step bisimulation equivalence $\approx_{rbs}$, rooted branching hp-bisimulation equivalence $\approx_{rbhp}$ and rooted branching hhp-bisimulation equivalence $\approx_{rbhhp}$, that it induces are all congruences.
\end{theorem}

\subsection{Prime Event Structure as Operational Semantics of Truly Concurrent Process Algebra} \label{sospespa} 

Based on the structurization and operational semantics of PES, we can establish an axiomatization of truly concurrent process APTC, which is almost the same as that in \cite{APTC}, but we reorganize it.

\subsubsection{Basic Algebra for True Concurrency}

In this section, we will discuss the algebraic laws for prime event structure $\mathcal{E}$, exactly for causality $\leq$ and conflict $\sharp$. We will follow the conventions of process algebra, using $\cdot$ instead of $\leq$ and $+$ instead of $\sharp$, also including the set of atomic actions $\mathbb{E}$, the deadlock constant 0 and the empty action 1, and let $a\in\mathbb{E}\cup\{0\}\cup\{1\}$. The resulted algebra is called Basic Algebra for True Concurrency, abbreviated BATC.

In the following, let $a,b,a',b'\in \mathbb{E}$, and let variables $x,y,z$ range over the set of terms for true concurrency, $p,q,s$ range over the set of closed terms. The set of axioms of BATC consists of the laws given in Table \ref{AxiomsForBATC}.

\begin{center}
    \begin{table}
        \begin{tabular}{@{}ll@{}}
            \hline No. &Axiom\\
            $A1$ & $x+ y = y+ x$\\
            $A2$ & $(x+ y)+ z = x+ (y+ z)$\\
            $A3$ & $x+ x = x$\\
            $A4$ & $(x+ y)\cdot z = x\cdot z + y\cdot z$\\
            $A5$ & $(x\cdot y)\cdot z = x\cdot(y\cdot z)$\\
            $A6$ & $x+0=x$\\
            $A7$ & $0\cdot x=0$\\
            $A8$ & $1\cdot x=x$\\
            $A9$ & $x\cdot 1=x$\\
        \end{tabular}
        \caption{Axioms of BATC}
        \label{AxiomsForBATC}
    \end{table}
\end{center}

\begin{definition}[Basic terms of BATC]\label{BTBATC}
The set of basic terms of BATC, $\mathcal{B}(BATC)$, is inductively defined as follows:

\begin{enumerate}
  \item $0,1\in\mathcal{B}(BATC)$;
  \item $\mathbb{E}\subset\mathcal{B}(BATC)$;
  \item if $a\in \mathbb{E}, t\in\mathcal{B}(BATC)$ then $a\cdot t\in\mathcal{B}(BATC)$;
  \item if $t,s\in\mathcal{B}(BATC)$ then $t+ s\in\mathcal{B}(BATC)$.
\end{enumerate}
\end{definition}

\begin{theorem}[Elimination theorem of BATC]\label{ETBATC}
Let $p$ be a closed BATC term. Then there is a basic BATC term $q$ such that $BATC\vdash p=q$.
\end{theorem}

\begin{proof}
(1) Firstly, suppose that the following ordering on the signature of BATC is defined: $\cdot > +$ and the symbol $\cdot$ is given the lexicographical status for the first argument, then for each rewrite rule $p\rightarrow q$ in Table \ref{TRSForBATC} relation $p>_{lpo} q$ can easily be proved. We obtain that the term rewrite system shown in Table \ref{TRSForBATC} is strongly normalizing, for it has finitely many rewriting rules, and $>$ is a well-founded ordering on the signature of BATC, and if $s>_{lpo} t$, for each rewriting rule $s\rightarrow t$ is in Table \ref{TRSForBATC} (see Theorem \ref{SN}).

\begin{center}
    \begin{table}
        \begin{tabular}{@{}ll@{}}
            \hline No. &Rewriting Rule\\
            $RA3$ & $x+ x \rightarrow x$\\
            $RA4$ & $(x+ y)\cdot z \rightarrow x\cdot z + y\cdot z$\\
            $RA5$ & $(x\cdot y)\cdot z \rightarrow x\cdot(y\cdot z)$\\
            $RA6$ & $x+0\rightarrow x$\\
            $RA7$ & $0\cdot x\rightarrow 0$\\
            $RA8$ & $1\cdot x\rightarrow x$\\
            $RA9$ & $x\cdot 1\rightarrow x$\\
        \end{tabular}
        \caption{Term rewrite system of BATC}
        \label{TRSForBATC}
    \end{table}
\end{center}

(2) Then we prove that the normal forms of closed BATC terms are basic BATC terms.

Suppose that $p$ is a normal form of some closed BATC term and suppose that $p$ is not a basic term. Let $p'$ denote the smallest sub-term of $p$ which is not a basic term. It implies that each sub-term of $p'$ is a basic term. Then we prove that $p$ is not a term in normal form. It is sufficient to induct on the structure of $p'$:

\begin{itemize}
  \item Case $p'\equiv 0, 0\in \mathbb{B}(BATC)$. $p'$ is a basic term, which contradicts the assumption that $p'$ is not a basic term, so this case should not occur.
  \item Case $p'\equiv 1, 1\in \mathbb{B}(BATC)$. $p'$ is a basic term, which contradicts the assumption that $p'$ is not a basic term, so this case should not occur.
  \item Case $p'\equiv a, a\in \mathbb{E}$. $p'$ is a basic term, which contradicts the assumption that $p'$ is not a basic term, so this case should not occur.
  \item Case $p'\equiv p_1\cdot p_2$. By induction on the structure of the basic term $p_1$:
      \begin{itemize}
        \item Subcase $p_1\equiv 0$. $p'$ would be a basic term, which contradicts the assumption that $p'$ is not a basic term;
        \item Subcase $p_1\equiv 1$. $p'$ would be a basic term, which contradicts the assumption that $p'$ is not a basic term;
        \item Subcase $p_1\in \mathbb{E}$. $p'$ would be a basic term, which contradicts the assumption that $p'$ is not a basic term;
        \item Subcase $p_1\equiv 0\cdot p_1'$. $RA7$ rewriting rule can be applied. So $p$ is not a normal form;
        \item Subcase $p_1\equiv 1\cdot p_1'$. $RA8$ rewriting rule can be applied. So $p$ is not a normal form;
        \item Subcase $p_1\equiv a\cdot p_1'$. $RA5$ rewriting rule can be applied. So $p$ is not a normal form;
        \item Subcase $p_1\equiv p_1'+ p_1''$. $RA4$ rewriting rule can be applied. So $p$ is not a normal form.
      \end{itemize}
  \item Case $p'\equiv p_1+ p_2$. By induction on the structure of the basic terms both $p_1$ and $p_2$, all subcases will lead to that $p'$ would be a basic term, which contradicts the assumption that $p'$ is not a basic term.
\end{itemize}
\end{proof}

We give the operational transition rules for 1,atomic action $a\in\mathbb{E}$, and operators $\cdot$ and $+$ as Table \ref{TRForBATC} shows. And the predicate $\xrightarrow{\{a_1,\cdots,a_n\}}\surd$ ($n\in\mathbb{N}$) represents successful termination after execution of the events $a_1,\cdots,a_n$.

\begin{center}
    \begin{table}
        $$\frac{}{1\rightarrow\surd}$$
        $$\frac{}{a_1,\cdots,a_n\xrightarrow{\{a_1,\cdots,a_n\}}\surd}$$
        $$\frac{x\xrightarrow{\{a_1,\cdots,a_n\}}\surd}{x+ y\xrightarrow{\{a_1,\cdots,a_n\}}\surd} \quad\frac{x\xrightarrow{\{a_1,\cdots,a_n\}}x'}{x+ y\xrightarrow{\{a_1,\cdots,a_n\}}x'} \quad\frac{y\xrightarrow{\{a_1,\cdots,a_n\}}\surd}{x+ y\xrightarrow{\{a_1,\cdots,a_n\}}\surd} \quad\frac{y\xrightarrow{\{a_1,\cdots,a_n\}}y'}{x+ y\xrightarrow{\{a_1,\cdots,a_n\}}y'}$$
        $$\frac{x\xrightarrow{\{a_1,\cdots,a_n\}}\surd}{x\cdot y\xrightarrow{\{a_1,\cdots,a_n\}} y} \quad\frac{x\xrightarrow{\{a_1,\cdots,a_n\}}x'}{x\cdot y\xrightarrow{\{a_1,\cdots,a_n\}}x'\cdot y}$$
        \caption{Transition rules of BATC}
        \label{TRForBATC}
    \end{table}
\end{center}

\begin{theorem}[Congruence of BATC with respect to truly concurrent bisimulation equivalences]
Truly concurrent bisimulation equivalences, including pomset bisimulation equivalence $\sim_{p}$, step bisimulation equivalence $\sim_s$, hp-bisimulation equivalence $\sim_{hp}$ and hhp-bisimulation equivalence $\sim_{hhp}$, are all congruences with respect to BATC.
\end{theorem}

\begin{proof}
Since the TSS of BATC in Table \ref{TRForBATC} is positive after reduction and in panth format, according to Theorem \ref{tcbac}, truly concurrent bisimulation equivalences, including pomset bisimulation equivalence $\sim_{p}$, step bisimulation equivalence $\sim_s$, hp-bisimulation equivalence $\sim_{hp}$ and hhp-bisimulation equivalence $\sim_{hhp}$, are all congruences with respect to BATC.
\end{proof}

\begin{theorem}[Soundness of BATC modulo truly concurrent bisimulation equivalences]\label{SBATC}
The axiomatization of BATC is sound modulo truly concurrent bisimulation equivalences, i.e.,

\begin{enumerate}
  \item let $x$ and $y$ be BATC terms. If $BATC\vdash x=y$, then $x\sim_{p} y$;
  \item let $x$ and $y$ be BATC terms. If $BATC\vdash x=y$, then $x\sim_{s} y$;
  \item let $x$ and $y$ be BATC terms. If $BATC\vdash x=y$, then $x\sim_{hp} y$;
  \item let $x$ and $y$ be BATC terms. If $BATC\vdash x=y$, then $x\sim_{hhp} y$.
\end{enumerate}
\end{theorem}

\begin{proof}
Since truly concurrent bisimulations $\sim_p$, $\sim_s$, $\sim_{hp}$ and $\sim_{hhp}$ are all both equivalent and congruent relations, we only need to check if each axiom in Table \ref{AxiomsForBATC} is sound modulo truly concurrent bisimulation equivalences $\sim_p$, $\sim_s$, $\sim_{hp}$ and $\sim_{hhp}$. The proof is trivial and we omit it.
\end{proof}

\begin{theorem}[Completeness of BATC modulo pomset bisimulation equivalence]\label{CBATC}
The axiomatization of BATC is complete modulo truly concurrent bisimulation equivalences, i.e.,

\begin{enumerate}
  \item let $p$ and $q$ be closed BATC terms, if $p\sim_{p} q$ then $p=q$;
  \item let $p$ and $q$ be closed BATC terms, if $p\sim_{s} q$ then $p=q$;
  \item let $p$ and $q$ be closed BATC terms, if $p\sim_{hp} q$ then $p=q$;
  \item let $p$ and $q$ be closed BATC terms, if $p\sim_{hhp} q$ then $p=q$.
\end{enumerate}
\end{theorem}

\begin{proof}
(1) Let $p$ and $q$ be closed BATC terms, if $p\sim_{p} q$ then $p=q$.

Firstly, by the elimination theorem of BATC, we know that for each closed BATC term $p$, there exists a closed basic BATC term $p'$, such that $BATC\vdash p=p'$, so, we only need to consider closed basic BATC terms.

The basic terms (see Definition \ref{BTBATC}) modulo associativity and commutativity (AC) of conflict $+$ (defined by axioms $A1$ and $A2$ in Table \ref{AxiomsForBATC}), and this equivalence is denoted by $=_{AC}$. Then, each equivalence class $s$ modulo AC of $+$ has the following normal form

$$s_1+\cdots+ s_k$$

with each $s_i$ either an atomic event or of the form $t_1\cdot t_2$, and each $s_i$ is called the summand of $s$.

Now, we prove that for normal forms $n$ and $n'$, if $n\sim_{p} n'$ then $n=_{AC}n'$. It is sufficient to induct on the sizes of $n$ and $n'$.

\begin{itemize}
  \item Consider a summand $a$ of $n$. Then $n\xrightarrow{a}\surd$, so $n\sim_p n'$ implies $n'\xrightarrow{a}\surd$, meaning that $n'$ also contains the summand $a$.
  \item Consider a summand $t_1\cdot t_2$ of $n$. Then $n\xrightarrow{t_1}t_2$, so $n\sim_p n'$ implies $n'\xrightarrow{t_1}t_2'$ with $t_2\sim_p t_2'$, meaning that $n'$ contains a summand $t_1\cdot t_2'$. Since $t_2$ and $t_2'$ are normal forms and have sizes smaller than $n$ and $n'$, by the induction hypotheses $t_2\sim_p t_2'$ implies $t_2=_{AC} t_2'$.
\end{itemize}

So, we get $n=_{AC} n'$.

Finally, let $s$ and $t$ be basic terms, and $s\sim_p t$, there are normal forms $n$ and $n'$, such that $s=n$ and $t=n'$. The soundness theorem of BATC modulo pomset bisimulation equivalence (see Theorem \ref{SBATC}) yields $s\sim_p n$ and $t\sim_p n'$, so $n\sim_p s\sim_p t\sim_p n'$. Since if $n\sim_p n'$ then $n=_{AC}n'$, $s=n=_{AC}n'=t$, as desired.

(2) Let $p$ and $q$ be closed BATC terms, if $p\sim_{s} q$ then $p=q$.

It can be proven similarly to (1).

(3) Let $p$ and $q$ be closed BATC terms, if $p\sim_{hp} q$ then $p=q$.

It can be proven similarly to (1).

(4) Let $p$ and $q$ be closed BATC terms, if $p\sim_{hhp} q$ then $p=q$.

It can be proven similarly to (1).
\end{proof}

\subsubsection{Algebra of Parallelism for True Concurrency}

In this section, we will discuss parallelism in true concurrency. We know that parallelism can be modeled by left merge and communication merge in ACP \cite{ACP} with an interleaving bisimulation semantics. Parallelism in true concurrency is quite different to that in interleaving bisimulation: it is a fundamental computational pattern (modeled by parallel operators $\parallel$ and $\leftmerge$) and cannot be merged (replaced by other operators $+$ and $\cdot$). As mentioned in the above sections, the unary operator $\Theta$ is the confliction elimination operator and the auxiliary binary operator $\triangleleft$ is the unless operator. The resulted algebra is called Algebra of Parallelism for True Concurrency, abbreviated APTC.

Firstly, we give the transition rules for parallelism as Table \ref{TRForParallel} shows, it is suitable for all truly concurrent behavioral equivalences, including pomset bisimulation, step bisimulation, hp-bisimulation and hhp-bisimulation.

\begin{center}
    \begin{table}
        $$\frac{x\xrightarrow{a_1}\surd\quad y\xrightarrow{a_2}\surd}{x\parallel y\xrightarrow{\{a_1,a_2\}}\surd} \quad\frac{x\xrightarrow{a_1}x'\quad y\xrightarrow{a_2}\surd}{x\parallel y\xrightarrow{\{a_1,a_2\}}x'}$$
        $$\frac{x\xrightarrow{a_1}\surd\quad y\xrightarrow{a_2}y'}{x\parallel y\xrightarrow{\{a_1,a_2\}}y'} \quad\frac{x\xrightarrow{a_1}x'\quad y\xrightarrow{a_2}y'}{x\parallel y\xrightarrow{\{a_1,a_2\}}x'\between y'}$$
        $$\frac{x\xrightarrow{a_1}\surd\quad y\xrightarrow{a_2}\surd \quad(a_1\leq a_2)}{x\leftmerge y\xrightarrow{\{a_1,a_2\}}\surd} \quad\frac{x\xrightarrow{a_1}x'\quad y\xrightarrow{a_2}\surd \quad(a_1\leq a_2)}{x\leftmerge y\xrightarrow{\{a_1,a_2\}}x'}$$
        $$\frac{x\xrightarrow{a_1}\surd\quad y\xrightarrow{a_2}y' \quad(a_1\leq a_2)}{x\leftmerge y\xrightarrow{\{a_1,a_2\}}y'} \quad\frac{x\xrightarrow{a_1}x'\quad y\xrightarrow{a_2}y' \quad(a_1\leq a_2)}{x\leftmerge y\xrightarrow{\{a_1,a_2\}}x'\between y'}$$
        $$\frac{x\xrightarrow{a_1}\surd\quad y\xrightarrow{a_2}\surd}{x\mid y\xrightarrow{\gamma(a_1,a_2)}\surd} \quad\frac{x\xrightarrow{a_1}x'\quad y\xrightarrow{a_2}\surd}{x\mid y\xrightarrow{\gamma(a_1,a_2)}x'}$$
        $$\frac{x\xrightarrow{a_1}\surd\quad y\xrightarrow{a_2}y'}{x\mid y\xrightarrow{\gamma(a_1,a_2)}y'} \quad\frac{x\xrightarrow{a_1}x'\quad y\xrightarrow{a_2}y'}{x\mid y\xrightarrow{\gamma(a_1,a_2)}x'\between y'}$$
        $$\frac{x\xrightarrow{a_1}\surd\quad (\sharp(a_1,a_2))}{\Theta(x)\xrightarrow{a_1}\surd} \quad\frac{x\xrightarrow{a_2}\surd\quad (\sharp(a_1,a_2))}{\Theta(x)\xrightarrow{a_2}\surd}$$
        $$\frac{x\xrightarrow{a_1}x'\quad (\sharp(a_1,a_2))}{\Theta(x)\xrightarrow{a_1}\Theta(x')} \quad\frac{x\xrightarrow{a_2}x'\quad (\sharp(a_1,a_2))}{\Theta(x)\xrightarrow{a_2}\Theta(x')}$$
        $$\frac{x\xrightarrow{a_1}\surd \quad y\nrightarrow^{a_2}\quad (\sharp(a_1,a_2))}{x\triangleleft y\rightarrow \surd}
        \quad\frac{x\xrightarrow{a_1}x' \quad y\nrightarrow^{a_2}\quad (\sharp(a_1,a_2))}{x\triangleleft y\rightarrow x'}$$
        $$\frac{x\xrightarrow{a_3}\surd \quad y\nrightarrow^{a_2}\quad (\sharp(a_1,a_2),a_3\leq a_1)}{x\triangleleft y\xrightarrow{a_3}\surd}
        \quad\frac{x\xrightarrow{a_3}x' \quad y\nrightarrow^{a_2}\quad (\sharp(a_1,a_2),a_3\leq a_1)}{x\triangleleft y\xrightarrow{a_3}x'}$$
        $$\frac{x\xrightarrow{a_1}\surd \quad y\nrightarrow^{a_3}\quad (\sharp(a_1,a_2),a_2\leq a_3)}{x\triangleleft y\rightarrow \surd}
        \quad\frac{x\xrightarrow{a_1}x' \quad y\nrightarrow^{a_3}\quad (\sharp(a_1,a_2),a_2\leq a_3)}{x\triangleleft y\rightarrow x'}$$
        $$\frac{x\xrightarrow{a_3}\surd \quad y\nrightarrow^{a_2}\quad (\sharp(a_1,a_2),a_1\leq a_3)}{x\triangleleft y\rightarrow \surd}
        \quad\frac{x\xrightarrow{a_3}x' \quad y\nrightarrow^{a_2}\quad (\sharp(a_1,a_2),a_1\leq a_3)}{x\triangleleft y\rightarrow x'}$$
        \caption{Transition rules of APTC}
        \label{TRForParallel}
    \end{table}
\end{center}

\begin{theorem}[Congruence theorem of APTC]
Truly concurrent bisimulation equivalences $\sim_{p}$, $\sim_s$, $\sim_{hp}$ and $\sim_{hhp}$ are all congruences with respect to APTC.
\end{theorem}

\begin{proof}
Since the TSS of APTC in Table \ref{TRForParallel} is positive after reduction and in panth format, according to Theorem \ref{tcbac}, truly concurrent bisimulation equivalences, including pomset bisimulation equivalence $\sim_{p}$, step bisimulation equivalence $\sim_s$, hp-bisimulation equivalence $\sim_{hp}$ and hhp-bisimulation equivalence $\sim_{hhp}$, are all congruences with respect to APTC.
\end{proof}

We design the axioms of parallelism in Table \ref{AxiomsForParallelism}, including algebraic laws for parallel operator $\parallel$, communication operator $\mid$, conflict elimination operator $\Theta$ and unless operator $\triangleleft$, and also the whole parallel operator $\between$. The communication between two communicating events in different parallel branches may cause deadlock 0 (a state of inactivity), which is caused by mismatch of two communicating events or the imperfectness of the communication channel.

\begin{center}
    \begin{table}
        \begin{tabular}{@{}ll@{}}
            \hline No. &Axiom\\
            $P1$ & $x\between y = x\parallel y + x\mid y$\\
            $P2$ & $x\parallel y = x\leftmerge y + y\leftmerge x$\\
            $P3$ & $(a_1\leq a_2)\quad a_1\leftmerge (a_2\cdot y) = (a_1\leftmerge a_2)\cdot y$\\
            $P4$ & $(a_1\leq a_2)\quad (a_1\cdot x)\leftmerge a_2 = (a_1\leftmerge a_2)\cdot x$\\
            $P5$ & $(a_1\leq a_2)\quad (a_1\cdot x)\leftmerge (a_2\cdot y) = (a_1\leftmerge a_2)\cdot (x\between y)$\\
            $P6$ & $(x+ y)\leftmerge z = (x\leftmerge z)+ (y\leftmerge z)$\\
            $P7$ & $0\leftmerge x = 0$\\
            $P8$ & $1\leftmerge x = x$\\
            $P9$ & $x\leftmerge 1 = x$\\
            $C10$ & $a_1\mid a_2 = \gamma(a_1,a_2)$\\
            $C11$ & $a_1\mid (a_2\cdot y) = \gamma(a_1,a_2)\cdot y$\\
            $C12$ & $(a_1\cdot x)\mid a_2 = \gamma(a_1,a_2)\cdot x$\\
            $C13$ & $(a_1\cdot x)\mid (a_2\cdot y) = \gamma(a_1,a_2)\cdot (x\between y)$\\
            $C14$ & $(x+ y)\mid z = (x\mid z) + (y\mid z)$\\
            $C15$ & $x\mid (y+ z) = (x\mid y)+ (x\mid z)$\\
            $C16$ & $0\mid x = 0$\\
            $C17$ & $x\mid0 = 0$\\
            $C18$ & $1\mid x = 0$\\
            $C19$ & $x\mid 1 = 0$\\
            $CE20$ & $\Theta(a) = a$\\
            $CE21$ & $\Theta(0) = 0$\\
            $CE22$ & $\Theta(x+ y) = \Theta(x) + \Theta(y)$\\
            $CE23$ & $\Theta(x\cdot y)=\Theta(x)\cdot\Theta(y)$\\
            $CE24$ & $\Theta(x\leftmerge y) = ((\Theta(x)\triangleleft y)\leftmerge y)+ ((\Theta(y)\triangleleft x)\leftmerge x)$\\
            $CE25$ & $\Theta(x\mid y) = ((\Theta(x)\triangleleft y)\mid y)+ ((\Theta(y)\triangleleft x)\mid x)$\\
            $U26$ & $(\sharp(a_1,a_2))\quad a_1\triangleleft a_2 = 1$\\
            $U27$ & $(\sharp(a_1,a_2),a_3\leq a_1)\quad a_3\triangleleft a_2 = a_3$\\
            $U28$ & $(\sharp(a_1,a_2),a_2\leq a_3)\quad a_1\triangleleft a_3 = 1$\\
            $U29$ & $(\sharp(a_1,a_2),a_2\leq a_3)\quad a_3\triangleleft a_1 = 1$\\
            $U30$ & $a\triangleleft 0 = a$\\
            $U31$ & $0 \triangleleft a = 0$\\
            $U32$ & $a\triangleleft 1 = a$\\
            $U33$ & $1 \triangleleft a = 1$\\
            $U34$ & $(x+ y)\triangleleft z = (x\triangleleft z)+ (y\triangleleft z)$\\
            $U35$ & $(x\cdot y)\triangleleft z = (x\triangleleft z)\cdot (y\triangleleft z)$\\
            $U36$ & $(x\leftmerge y)\triangleleft z = (x\triangleleft z)\leftmerge (y\triangleleft z)$\\
            $U37$ & $(x\mid y)\triangleleft z = (x\triangleleft z)\mid (y\triangleleft z)$\\
            $U38$ & $x\triangleleft (y+ z) = (x\triangleleft y)\triangleleft z$\\
            $U39$ & $x\triangleleft (y\cdot z)=(x\triangleleft y)\triangleleft z$\\
            $U40$ & $x\triangleleft (y\leftmerge z) = (x\triangleleft y)\triangleleft z$\\
            $U41$ & $x\triangleleft (y\mid z) = (x\triangleleft y)\triangleleft z$\\
        \end{tabular}
        \caption{Axioms of parallelism}
        \label{AxiomsForParallelism}
    \end{table}
\end{center}

\begin{definition}[Basic terms of APTC]\label{BTAPTC}
The set of basic terms of APTC, $\mathcal{B}(APTC)$, is inductively defined as follows:
\begin{enumerate}
  \item $0,1\in\mathcal{B}(APTC)$;
  \item $\mathbb{E}\subset\mathcal{B}(APTC)$;
  \item if $a\in \mathbb{E}, t\in\mathcal{B}(APTC)$ then $a\cdot t\in\mathcal{B}(APTC)$;
  \item if $t,s\in\mathcal{B}(APTC)$ then $t+ s\in\mathcal{B}(APTC)$;
  \item if $t,s\in\mathcal{B}(APTC)$ then $t\leftmerge s\in\mathcal{B}(APTC)$.
\end{enumerate}
\end{definition}

Based on the definition of basic terms for APTC (see Definition \ref{BTAPTC}) and axioms of parallelism (see Table \ref{AxiomsForParallelism}), we can prove the elimination theorem of parallelism.

\begin{theorem}[Elimination theorem of parallelism]\label{ETParallelism}
Let $p$ be a closed APTC term. Then there is a basic APTC term $q$ such that $APTC\vdash p=q$.
\end{theorem}

\begin{proof}
(1) Firstly, suppose that the following ordering on the signature of APTC is defined: $\parallel > \cdot > +$ and the symbol $\cdot$ is given the lexicographical status for the first argument, then for each rewrite rule $p\rightarrow q$ in Table \ref{TRSForAPTC} relation $p>_{lpo} q$ can easily be proved. We obtain that the term rewrite system shown in Table \ref{TRSForAPTC} is strongly normalizing, for it has finitely many rewriting rules, and $>$ is a well-founded ordering on the signature of APTC, and if $s>_{lpo} t$, for each rewriting rule $s\rightarrow t$ is in Table \ref{TRSForAPTC} (see Theorem \ref{SN}).

\begin{center}
    \begin{table}
        \begin{tabular}{@{}ll@{}}
            \hline No. &Rewriting Rule\\
            $RP1$ & $x\between y \rightarrow x\parallel y + x\mid y$\\
            $RP2$ & $x\parallel y \rightarrow x\leftmerge y + y\leftmerge x$\\
            $RP3$ & $(a_1\leq a_2)\quad a_1\leftmerge (a_2\cdot y) \rightarrow (a_1\leftmerge a_2)\cdot y$\\
            $RP4$ & $(a_1\leq a_2)\quad (a_1\cdot x)\leftmerge a_2 \rightarrow (a_1\leftmerge a_2)\cdot x$\\
            $RP5$ & $(a_1\leq a_2)\quad (a_1\cdot x)\leftmerge (a_2\cdot y) \rightarrow (a_1\leftmerge a_2)\cdot (x\between y)$\\
            $RP6$ & $(x+ y)\leftmerge z \rightarrow (x\leftmerge z)+ (y\leftmerge z)$\\
            $RP7$ & $0\leftmerge x \rightarrow 0$\\
            $RP8$ & $1\leftmerge x \rightarrow x$\\
            $RP9$ & $x\leftmerge 1 \rightarrow x$\\
            $RC10$ & $a_1\mid a_2 \rightarrow \gamma(a_1,a_2)$\\
            $RC11$ & $a_1\mid (a_2\cdot y) \rightarrow \gamma(a_1,a_2)\cdot y$\\
            $RC12$ & $(a_1\cdot x)\mid a_2 \rightarrow \gamma(a_1,a_2)\cdot x$\\
            $RC13$ & $(a_1\cdot x)\mid (a_2\cdot y) \rightarrow \gamma(a_1,a_2)\cdot (x\between y)$\\
            $RC14$ & $(x+ y)\mid z \rightarrow (x\mid z) + (y\mid z)$\\
            $RC15$ & $x\mid (y+ z) \rightarrow (x\mid y)+ (x\mid z)$\\
            $RC16$ & $0\mid x \rightarrow 0$\\
            $RC17$ & $x\mid0 \rightarrow 0$\\
            $RC18$ & $1\mid x \rightarrow 0$\\
            $RC19$ & $x\mid 1 \rightarrow 0$\\
            $RCE20$ & $\Theta(a) \rightarrow a$\\
            $RCE21$ & $\Theta(0) \rightarrow 0$\\
            $RCE22$ & $\Theta(x+ y) \rightarrow \Theta(x) + \Theta(y)$\\
            $RCE23$ & $\Theta(x\cdot y)\rightarrow\Theta(x)\cdot\Theta(y)$\\
            $RCE24$ & $\Theta(x\leftmerge y) \rightarrow ((\Theta(x)\triangleleft y)\leftmerge y)+ ((\Theta(y)\triangleleft x)\leftmerge x)$\\
            $RCE25$ & $\Theta(x\mid y) \rightarrow ((\Theta(x)\triangleleft y)\mid y)+ ((\Theta(y)\triangleleft x)\mid x)$\\
            $RU26$ & $(\sharp(a_1,a_2))\quad a_1\triangleleft a_2 = 1$\\
            $RU27$ & $(\sharp(a_1,a_2),a_3\leq a_1)\quad a_3\triangleleft a_2 \rightarrow a_3$\\
            $RU28$ & $(\sharp(a_1,a_2),a_2\leq a_3)\quad a_1\triangleleft a_3 \rightarrow 1$\\
            $RU29$ & $(\sharp(a_1,a_2),a_2\leq a_3)\quad a_3\triangleleft a_1 \rightarrow 1$\\
            $RU30$ & $a\triangleleft 0 \rightarrow a$\\
            $RU31$ & $0 \triangleleft a \rightarrow 0$\\
            $RU32$ & $a\triangleleft 1 \rightarrow a$\\
            $RU33$ & $1 \triangleleft a \rightarrow 1$\\
            $RU34$ & $(x+ y)\triangleleft z \rightarrow (x\triangleleft z)+ (y\triangleleft z)$\\
            $RU35$ & $(x\cdot y)\triangleleft z \rightarrow (x\triangleleft z)\cdot (y\triangleleft z)$\\
            $RU36$ & $(x\leftmerge y)\triangleleft z \rightarrow (x\triangleleft z)\leftmerge (y\triangleleft z)$\\
            $RU37$ & $(x\mid y)\triangleleft z \rightarrow (x\triangleleft z)\mid (y\triangleleft z)$\\
            $RU38$ & $x\triangleleft (y+ z) \rightarrow (x\triangleleft y)\triangleleft z$\\
            $RU39$ & $x\triangleleft (y\cdot z)\rightarrow(x\triangleleft y)\triangleleft z$\\
            $RU40$ & $x\triangleleft (y\leftmerge z) \rightarrow (x\triangleleft y)\triangleleft z$\\
            $RU41$ & $x\triangleleft (y\mid z) \rightarrow (x\triangleleft y)\triangleleft z$\\
        \end{tabular}
        \caption{Term rewrite system of APTC}
        \label{TRSForAPTC}
    \end{table}
\end{center}

(2) Then we prove that the normal forms of closed APTC terms are basic APTC terms.

Suppose that $p$ is a normal form of some closed APTC term and suppose that $p$ is not a basic APTC term. Let $p'$ denote the smallest sub-term of $p$ which is not a basic APTC term. It implies that each sub-term of $p'$ is a basic APTC term. Then we prove that $p$ is not a term in normal form. It is sufficient to induct on the structure of $p'$:

\begin{itemize}
  \item Case $p'\equiv 0, 0\in \mathbb{B}(BATC)$. $p'$ is a basic term, which contradicts the assumption that $p'$ is not a basic term, so this case should not occur.
  \item Case $p'\equiv 1, 1\in \mathbb{B}(BATC)$. $p'$ is a basic term, which contradicts the assumption that $p'$ is not a basic term, so this case should not occur.
  \item Case $p'\equiv a, a\in \mathbb{E}$. $p'$ is a basic APTC term, which contradicts the assumption that $p'$ is not a basic APTC term, so this case should not occur.
  \item Case $p'\equiv p_1\cdot p_2$. By induction on the structure of the basic APTC term $p_1$:
      \begin{itemize}
        \item Subcase $p_1\equiv 0$. $p'$ would be a basic term, which contradicts the assumption that $p'$ is not a basic term;
        \item Subcase $p_1\equiv 1$. $p'$ would be a basic term, which contradicts the assumption that $p'$ is not a basic term;
        \item Subcase $p_1\in \mathbb{E}$. $p'$ would be a basic APTC term, which contradicts the assumption that $p'$ is not a basic APTC term;
        \item Subcase $p_1\equiv a\cdot p_1'$. $RA5$ rewriting rule in Table \ref{TRSForBATC} can be applied. So $p$ is not a normal form;
        \item Subcase $p_1\equiv p_1'+ p_1''$. $RA4$ rewriting rule in Table \ref{TRSForBATC} can be applied. So $p$ is not a normal form;
        \item Subcase $p_1\equiv p_1'\leftmerge p_1''$. $p'$ would be a basic APTC term, which contradicts the assumption that $p'$ is not a basic APTC term;
        \item Subcase $p_1\equiv p_1'\mid p_1''$. $p'$ would be a basic APTC term, which contradicts the assumption that $p'$ is not a basic APTC term;
        \item Subcase $p_1\equiv \Theta(p_1')$. $RCE20-RCE25$ rewrite rules in Table \ref{TRSForAPTC} can be applied. So $p$ is not a normal form.
      \end{itemize}
  \item Case $p'\equiv p_1+ p_2$. By induction on the structure of the basic APTC terms both $p_1$ and $p_2$, all subcases will lead to that $p'$ would be a basic APTC term, which contradicts the assumption that $p'$ is not a basic APTC term.
  \item Case $p'\equiv p_1\leftmerge p_2$. By induction on the structure of the basic APTC terms both $p_1$ and $p_2$, all subcases will lead to that $p'$ would be a basic APTC term, which contradicts the assumption that $p'$ is not a basic APTC term.
  \item Case $p'\equiv p_1\mid p_2$. By induction on the structure of the basic APTC terms both $p_1$ and $p_2$, all subcases will lead to that $p'$ would be a basic APTC term, which contradicts the assumption that $p'$ is not a basic APTC term.
  \item Case $p'\equiv \Theta(p_1)$. By induction on the structure of the basic APTC term $p_1$, $RCE20-RCE25$ rewrite rules in Table \ref{TRSForAPTC} can be applied. So $p$ is not a normal form.
  \item Case $p'\equiv p_1\triangleleft p_2$. By induction on the structure of the basic APTC terms both $p_1$ and $p_2$, all subcases will lead to that $p'$ would be a basic APTC term, which contradicts the assumption that $p'$ is not a basic APTC term.
\end{itemize}
\end{proof}

\begin{theorem}[Generalization of the algebra for parallelism with respect to BATC]
The algebra for parallelism is a generalization of BATC.
\end{theorem}

\begin{proof}
It follows from the following three facts.

\begin{enumerate}
  \item The transition rules of BATC Table \ref{TRForBATC} are all source-dependent;
  \item The sources of the transition rules for the algebra for parallelism in Table \ref{TRForParallel} contain an occurrence of $\between$, or $\leftmerge$, or $\mid$, or $\Theta$, or $\triangleleft$;
  \item The transition rules of APTC are all source-dependent.
\end{enumerate}

So, the algebra for parallelism is a generalization of BATC, that is, BATC is an embedding of the algebra for parallelism, as desired.
\end{proof}

\begin{theorem}[Soundness of parallelism modulo truly concurrent bisimulation equivalences]\label{SAPTC}
The axiomatization of APTC is sound modulo truly concurrent bisimulation equivalences, i.e.,

\begin{enumerate}
  \item let $x$ and $y$ be APTC terms. If $APTC\vdash x=y$, then $x\sim_{p} y$;
  \item let $x$ and $y$ be APTC terms. If $APTC\vdash x=y$, then $x\sim_{s} y$;
  \item let $x$ and $y$ be APTC terms. If $APTC\vdash x=y$, then $x\sim_{hp} y$;
  \item let $x$ and $y$ be APTC terms. If $APTC\vdash x=y$, then $x\sim_{hhp} y$.
\end{enumerate}
\end{theorem}

\begin{proof}
Since truly concurrent bisimulations $\sim_p$, $\sim_s$, $\sim_{hp}$ and $\sim_{hhp}$ are all both equivalent and congruent relations, we only need to check if each axiom in Table \ref{AxiomsForParallelism} is sound modulo truly concurrent bisimulation equivalences $\sim_p$, $\sim_s$, $\sim_{hp}$ and $\sim_{hhp}$. The proof is trivial and we omit it.
\end{proof}

\begin{theorem}[Completeness of parallelism modulo truly concurrent bisimulation equivalences]\label{CAPTC}
The axiomatization of APTC is complete modulo truly concurrent bisimulation equivalences, i.e.,

\begin{enumerate}
  \item let $p$ and $q$ be closed APTC terms, if $p\sim_{p} q$ then $p=q$;
  \item let $p$ and $q$ be closed APTC terms, if $p\sim_{s} q$ then $p=q$;
  \item let $p$ and $q$ be closed APTC terms, if $p\sim_{hp} q$ then $p=q$;
  \item let $p$ and $q$ be closed APTC terms, if $p\sim_{hhp} q$ then $p=q$.
\end{enumerate}
\end{theorem}

\begin{proof}
(1) Let $p$ and $q$ be closed APTC terms, if $p\sim_{p} q$ then $p=q$.

Firstly, by the elimination theorem of APTC (see Theorem \ref{ETParallelism}), we know that for each closed APTC term $p$, there exists a closed basic APTC term $p'$, such that $APTC\vdash p=p'$, so, we only need to consider closed basic APTC terms.

The basic terms (see Definition \ref{BTAPTC}) modulo associativity and commutativity (AC) of conflict $+$ (defined by axioms $A1$ and $A2$ in Table \ref{AxiomsForBATC}), and these equivalences is denoted by $=_{AC}$. Then, each equivalence class $s$ modulo AC of $+$ has the following normal form

$$s_1+\cdots+ s_k$$

with each $s_i$ either an atomic event or of the form

$$t_1\cdot\cdots\cdot t_m$$

with each $t_j$ either an atomic event or of the form

$$u_1\leftmerge\cdots\leftmerge u_l$$

with each $u_l$ an atomic event, and each $s_i$ is called the summand of $s$.

Now, we prove that for normal forms $n$ and $n'$, if $n\sim_{p} n'$ then $n=_{AC}n'$. It is sufficient to induct on the sizes of $n$ and $n'$.

\begin{itemize}
  \item Consider a summand $a$ of $n$. Then $n\xrightarrow{a}\surd$, so $n\sim_p n'$ implies $n'\xrightarrow{a}\surd$, meaning that $n'$ also contains the summand $a$.
  \item Consider a summand $t_1\cdot t_2$ of $n$,
  \begin{itemize}
    \item if $t_1\equiv a'$, then $n\xrightarrow{a'}t_2$, so $n\sim_p n'$ implies $n'\xrightarrow{a'}t_2'$ with $t_2\sim_p t_2'$, meaning that $n'$ contains a summand $a'\cdot t_2'$. Since $t_2$ and $t_2'$ are normal forms and have sizes smaller than $n$ and $n'$, by the induction hypotheses if $t_2\sim_p t_2'$ then $t_2=_{AC} t_2'$;
    \item if $t_1\equiv a_1\leftmerge\cdots\leftmerge a_l$, then $n\xrightarrow{\{a_1,\cdots,a_l\}}t_2$, so $n\sim_p n'$ implies $n'\xrightarrow{\{a_1,\cdots,a_l\}}t_2'$ with $t_2\sim_p t_2'$, meaning that $n'$ contains a summand $(a_1\leftmerge\cdots\leftmerge a_l)\cdot t_2'$. Since $t_2$ and $t_2'$ are normal forms and have sizes smaller than $n$ and $n'$, by the induction hypotheses if $t_2\sim_p t_2'$ then $t_2=_{AC} t_2'$.
  \end{itemize}
\end{itemize}

So, we get $n=_{AC} n'$.

Finally, let $s$ and $t$ be basic APTC terms, and $s\sim_p t$, there are normal forms $n$ and $n'$, such that $s=n$ and $t=n'$. The soundness theorem of parallelism modulo pomset bisimulation equivalence (see Theorem \ref{SAPTC}) yields $s\sim_p n$ and $t\sim_p n'$, so $n\sim_p s\sim_p t\sim_p n'$. Since if $n\sim_p n'$ then $n=_{AC}n'$, $s=n=_{AC}n'=t$, as desired.

(2) Let $p$ and $q$ be closed APTC terms, if $p\sim_{s} q$ then $p=q$.

It can be proven similarly to (1).

(3) Let $p$ and $q$ be closed APTC terms, if $p\sim_{hp} q$ then $p=q$.

It can be proven similarly to (1).

(4) Let $p$ and $q$ be closed APTC terms, if $p\sim_{hhp} q$ then $p=q$.

It can be proven similarly to (1).
\end{proof}

The mismatch of two communicating events in different parallel branches can cause deadlock 0, so the deadlocks in the concurrent processes should be eliminated. Like $ACP$ \cite{ACP}, we also introduce the unary encapsulation operator $\partial_H$ for set $H$ of atomic events, which renames all atomic events in $H$ into $0$. The whole algebra including parallelism for true concurrency in the above subsections, deadlock $0$ and encapsulation operator $\partial_H$, is called Algebra of Parallelism for True Concurrency, abbreviated APTC.

The transition rules of encapsulation operator $\partial_H$ are shown in Table \ref{TRForEncapsulation}.

\begin{center}
    \begin{table}
        $$\frac{x\xrightarrow{a}\surd}{\partial_H(x)\xrightarrow{a}\surd}\quad (a\notin H)\quad\quad\frac{x\xrightarrow{a}x'}{\partial_H(x)\xrightarrow{a}\partial_H(x')}\quad(a\notin H)$$
        \caption{Transition rules of encapsulation operator $\partial_H$}
        \label{TRForEncapsulation}
    \end{table}
\end{center}

Based on the transition rules for encapsulation operator $\partial_H$ in Table \ref{TRForEncapsulation}, we design the axioms as Table \ref{AxiomsForEncapsulation} shows.

\begin{center}
    \begin{table}
        \begin{tabular}{@{}ll@{}}
            \hline No. &Axiom\\
            $D1$ & $a\notin H\quad\partial_H(a) = a$\\
            $D2$ & $a\in H\quad \partial_H(a) = 0$\\
            $D3$ & $\partial_H(0) = 0$\\
            $D4$ & $\partial_H(1) = 1$\\
            $D5$ & $\partial_H(x+ y) = \partial_H(x)+\partial_H(y)$\\
            $D6$ & $\partial_H(x\cdot y) = \partial_H(x)\cdot\partial_H(y)$\\
            $D7$ & $\partial_H(x\leftmerge y) = \partial_H(x)\leftmerge\partial_H(y)$\\
        \end{tabular}
        \caption{Axioms of encapsulation operator}
        \label{AxiomsForEncapsulation}
    \end{table}
\end{center}

\begin{theorem}[Conservativity of APTC with respect to the algebra for parallelism]
APTC is a conservative extension of the algebra for parallelism.
\end{theorem}

\begin{proof}
It follows from the following two facts:

\begin{enumerate}
  \item The transition rules of the algebra for parallelism in Table \ref{TRForParallel} are all source-dependent;
  \item The sources of the transition rules for the encapsulation operator in Table \ref{TRForEncapsulation} contain an occurrence of $\partial_H$.
\end{enumerate}

So, APTC is a conservative extension of the algebra for parallelism, as desired.
\end{proof}

\begin{theorem}[Congruence theorem of encapsulation operator $\partial_H$]
Truly concurrent bisimulation equivalences $\sim_{p}$, $\sim_s$, $\sim_{hp}$ and $\sim_{hhp}$ are all congruences with respect to encapsulation operator $\partial_H$.
\end{theorem}

\begin{proof}
Since the TSS of encapsulation operator in Table \ref{TRForEncapsulation} is positive after reduction and in panth format, according to Theorem \ref{tcbac}, truly concurrent bisimulation equivalences, including pomset bisimulation equivalence $\sim_{p}$, step bisimulation equivalence $\sim_s$, hp-bisimulation equivalence $\sim_{hp}$ and hhp-bisimulation equivalence $\sim_{hhp}$, are all congruences with respect to APTC.
\end{proof}

\begin{theorem}[Elimination theorem of APTC]\label{ETEncapsulation}
Let $p$ be a closed APTC term including the encapsulation operator $\partial_H$. Then there is a basic APTC term $q$ such that $APTC\vdash p=q$.
\end{theorem}

\begin{proof}
(1) Firstly, suppose that the following ordering on the signature of APTC is defined: $\leftmerge > \cdot > +$ and the symbol $\cdot$ is given the lexicographical status for the first argument, then for each rewrite rule $p\rightarrow q$ in Table \ref{TRSForEncapsulation} relation $p>_{lpo} q$ can easily be proved. We obtain that the term rewrite system shown in Table \ref{TRSForEncapsulation} is strongly normalizing, for it has finitely many rewriting rules, and $>$ is a well-founded ordering on the signature of APTC, and if $s>_{lpo} t$, for each rewriting rule $s\rightarrow t$ is in Table \ref{TRSForEncapsulation} (see Theorem \ref{SN}).

\begin{center}
    \begin{table}
        \begin{tabular}{@{}ll@{}}
            \hline No. &Rewriting Rule\\
            $RD1$ & $a\notin H\quad\partial_H(a) \rightarrow a$\\
            $RD2$ & $a\in H\quad \partial_H(a) \rightarrow 0$\\
            $RD3$ & $\partial_H(0) \rightarrow 0$\\
            $RD4$ & $\partial_H(1) \rightarrow 1$\\
            $RD5$ & $\partial_H(x+ y) \rightarrow \partial_H(x)+\partial_H(y)$\\
            $RD6$ & $\partial_H(x\cdot y) \rightarrow \partial_H(x)\cdot\partial_H(y)$\\
            $RD7$ & $\partial_H(x\leftmerge y) \rightarrow \partial_H(x)\leftmerge\partial_H(y)$\\
        \end{tabular}
        \caption{Term rewrite system of encapsulation operator $\partial_H$}
        \label{TRSForEncapsulation}
    \end{table}
\end{center}

(2) Then we prove that the normal forms of closed APTC terms including encapsulation operator $\partial_H$ are basic APTC terms.

Suppose that $p$ is a normal form of some closed APTC term and suppose that $p$ is not a basic APTC term. Let $p'$ denote the smallest sub-term of $p$ which is not a basic APTC term. It implies that each sub-term of $p'$ is a basic APTC term. Then we prove that $p$ is not a term in normal form. It is sufficient to induct on the structure of $p'$, we only prove the new case $p'\equiv \partial_H(p_1)$:

\begin{itemize}
  \item Case $p_1\equiv a$. The transition rules $RD1$ or $RD2$ can be applied, so $p$ is not a normal form;
  \item Case $p_1\equiv 0$. The transition rules $RD3$ can be applied, so $p$ is not a normal form;
  \item Case $p_1\equiv 1$. The transition rules $RD4$ can be applied, so $p$ is not a normal form;
  \item Case $p_1\equiv p_1'+ p_1''$. The transition rules $RD5$ can be applied, so $p$ is not a normal form;
  \item Case $p_1\equiv p_1'\cdot p_1''$. The transition rules $RD6$ can be applied, so $p$ is not a normal form;
  \item Case $p_1\equiv p_1'\leftmerge p_1''$. The transition rules $RD7$ can be applied, so $p$ is not a normal form.
\end{itemize}
\end{proof}

\begin{theorem}[Soundness of APTC modulo truly concurrent bisimulation equivalences]\label{SAPTC2}
The axiomatization of APTC is sound modulo truly concurrent bisimulation equivalences, i.e.,

\begin{enumerate}
  \item let $x$ and $y$ be APTC terms including encapsulation operator $\partial_H$. If $APTC\vdash x=y$, then $x\sim_{p} y$;
  \item let $x$ and $y$ be APTC terms including encapsulation operator $\partial_H$. If $APTC\vdash x=y$, then $x\sim_{s} y$;
  \item let $x$ and $y$ be APTC terms including encapsulation operator $\partial_H$. If $APTC\vdash x=y$, then $x\sim_{hp} y$;
  \item let $x$ and $y$ be APTC terms including encapsulation operator $\partial_H$. If $APTC\vdash x=y$, then $x\sim_{hhp} y$.
\end{enumerate}
\end{theorem}

\begin{proof}
Since truly concurrent bisimulations $\sim_p$, $\sim_s$, $\sim_{hp}$ and $\sim_{hhp}$ are all both equivalent and congruent relations, we only need to check if each axiom in Table \ref{AxiomsForEncapsulation} is sound modulo truly concurrent bisimulation equivalences $\sim_p$, $\sim_s$, $\sim_{hp}$ and $\sim_{hhp}$. The proof is trivial and we omit it.
\end{proof}

\begin{theorem}[Completeness of APTC modulo truly concurrent bisimulation equivalences]\label{CAPTC2}
The axiomatization of APTC is complete modulo truly concurrent bisimulation equivalences, i.e.,

\begin{enumerate}
  \item let $p$ and $q$ be closed APTC terms including encapsulation operator $\partial_H$, if $p\sim_{p} q$ then $p=q$;
  \item let $p$ and $q$ be closed APTC terms including encapsulation operator $\partial_H$, if $p\sim_{s} q$ then $p=q$;
  \item let $p$ and $q$ be closed APTC terms including encapsulation operator $\partial_H$, if $p\sim_{hp} q$ then $p=q$;
  \item let $p$ and $q$ be closed APTC terms including encapsulation operator $\partial_H$, if $p\sim_{hhp} q$ then $p=q$.
\end{enumerate}
\end{theorem}

\begin{proof}
(1) Let $p$ and $q$ be closed APTC terms including encapsulation operator $\partial_H$, if $p\sim_{p} q$ then $p=q$.

Firstly, by the elimination theorem of APTC (see Theorem \ref{ETEncapsulation}), we know that the normal form of APTC does not contain $\partial_H$, and for each closed APTC term $p$, there exists a closed basic APTC term $p'$, such that $APTC\vdash p=p'$, so, we only need to consider closed basic APTC terms.

Similarly to Theorem \ref{CAPTC}, we can prove that for normal forms $n$ and $n'$, if $n\sim_{p} n'$ then $n=_{AC}n'$.

Finally, let $s$ and $t$ be basic APTC terms, and $s\sim_p t$, there are normal forms $n$ and $n'$, such that $s=n$ and $t=n'$. The soundness theorem of APTC modulo pomset bisimulation equivalence (see Theorem \ref{SAPTC2}) yields $s\sim_p n$ and $t\sim_p n'$, so $n\sim_p s\sim_p t\sim_p n'$. Since if $n\sim_p n'$ then $n=_{AC}n'$, $s=n=_{AC}n'=t$, as desired.

(2) Let $p$ and $q$ be closed APTC terms including encapsulation operator $\partial_H$, if $p\sim_{s} q$ then $p=q$.

It can be proven similarly to (1).

(3) Let $p$ and $q$ be closed APTC terms including encapsulation operator $\partial_H$, if $p\sim_{hp} q$ then $p=q$.

It can be proven similarly to (1).

(4) Let $p$ and $q$ be closed APTC terms including encapsulation operator $\partial_H$, if $p\sim_{hhp} q$ then $p=q$.

It can be proven similarly to (1).
\end{proof}

\subsubsection{Recursion}

In this section, we introduce recursion to capture infinite processes based on APTC. Since in APTC, there are three basic operators $\cdot$, $+$ and $\leftmerge$, the recursion must be adapted this situation to include $\leftmerge$.

In the following, $E,F,G$ are recursion specifications, $X,Y,Z$ are recursive variables.

\begin{definition}[Recursive specification]
A recursive specification is a finite set of recursive equations

$$X_1=t_1(X_1,\cdots,X_n)$$
$$\cdots$$
$$X_n=t_n(X_1,\cdots,X_n)$$

where the left-hand sides of $X_i$ are called recursion variables, and the right-hand sides $t_i(X_1,\cdots,X_n)$ are process terms in APTC with possible occurrences of the recursion variables $X_1,\cdots,X_n$.
\end{definition}

\begin{definition}[Solution]
Processes $p_1,\cdots,p_n$ are solutions for a recursive specification $\{X_i=t_i(X_1,\cdots,X_n)|i\in\{1,\cdots,n\}\}$ (with respect to truly concurrent bisimulation equivalences $\sim_s$($\sim_p$, $\sim_{hp}$)) if $p_i\sim_s (\sim_p, \sim_{hp})t_i(p_1,\cdots,p_n)$ for $i\in\{1,\cdots,n\}$.
\end{definition}

\begin{definition}[Guarded recursive specification]
A recursive specification

$$X_1=t_1(X_1,\cdots,X_n)$$
$$...$$
$$X_n=t_n(X_1,\cdots,X_n)$$

is guarded if the right-hand sides of its recursive equations can be adapted to the form by applications of the axioms in APTC and replacing recursion variables by the right-hand sides of their recursive equations, and there does not exist an infinite sequence of $1$-transitions $\langle X|E\rangle\rightarrow\langle X'|E\rangle\rightarrow\langle X''|E\rangle\rightarrow\cdots$,

$$(a_{11}\leftmerge\cdots\leftmerge a_{1i_1})\cdot s_1(X_1,\cdots,X_n)+\cdots+(a_{k1}\leftmerge\cdots\leftmerge a_{ki_k})\cdot s_k(X_1,\cdots,X_n)+(b_{11}\leftmerge\cdots\leftmerge b_{1j_1})+\cdots+(b_{1j_1}\leftmerge\cdots\leftmerge b_{lj_l})$$

where $a_{11},\cdots,a_{1i_1},a_{k1},\cdots,a_{ki_k},b_{11},\cdots,b_{1j_1},b_{1j_1},\cdots,b_{lj_l}\in \mathbb{E}$, and the sum above is allowed to be empty, in which case it represents the deadlock $0$.
\end{definition}

\begin{definition}[Linear recursive specification]\label{LRS}
A recursive specification is linear if its recursive equations are of the form

$$(a_{11}\leftmerge\cdots\leftmerge a_{1i_1})X_1+\cdots+(a_{k1}\leftmerge\cdots\leftmerge a_{ki_k})X_k+(b_{11}\leftmerge\cdots\leftmerge b_{1j_1})+\cdots+(b_{1j_1}\leftmerge\cdots\leftmerge b_{lj_l})$$

where $a_{11},\cdots,a_{1i_1},a_{k1},\cdots,a_{ki_k},b_{11},\cdots,b_{1j_1},b_{1j_1},\cdots,b_{lj_l}\in \mathbb{E}$, and the sum above is allowed to be empty, in which case it represents the deadlock $0$.
\end{definition}

For a guarded recursive specifications $E$ with the form

$$X_1=t_1(X_1,\cdots,X_n)$$
$$\cdots$$
$$X_n=t_n(X_1,\cdots,X_n)$$

the behavior of the solution $\langle X_i|E\rangle$ for the recursion variable $X_i$ in $E$, where $i\in\{1,\cdots,n\}$, is exactly the behavior of their right-hand sides $t_i(X_1,\cdots,X_n)$, which is captured by the two transition rules in Table \ref{TRForGR}.

\begin{center}
    \begin{table}
        $$\frac{t_i(\langle X_1|E\rangle,\cdots,\langle X_n|E\rangle)\xrightarrow{\{a_1,\cdots,a_k\}}\surd}{\langle X_i|E\rangle\xrightarrow{\{a_1,\cdots,a_k\}}\surd}$$
        $$\frac{t_i(\langle X_1|E\rangle,\cdots,\langle X_n|E\rangle)\xrightarrow{\{a_1,\cdots,a_k\}} y}{\langle X_i|E\rangle\xrightarrow{\{a_1,\cdots,a_k\}} y}$$
        \caption{Transition rules of guarded recursion}
        \label{TRForGR}
    \end{table}
\end{center}

\begin{theorem}[Conservativity of APTC with guarded recursion]
APTC with guarded recursion is a conservative extension of APTC.
\end{theorem}

\begin{proof}
Since the transition rules of APTC are source-dependent, and the transition rules for guarded recursion in Table \ref{TRForGR} contain only a fresh constant in their source, so the transition rules of APTC with guarded recursion are conservative extensions of those of APTC.
\end{proof}

\begin{theorem}[Congruence theorem of APTC with guarded recursion]
Truly concurrent bisimulation equivalences $\sim_{p}$, $\sim_s$ and $\sim_{hp}$ are all congruences with respect to APTC with guarded recursion.
\end{theorem}

\begin{proof}
Since the TSS of APTC with guarded recursion in Table \ref{TRForGR} is positive after reduction and in panth format, according to Theorem \ref{tcbac}, truly concurrent bisimulation equivalences, including pomset bisimulation equivalence $\sim_{p}$, step bisimulation equivalence $\sim_s$, hp-bisimulation equivalence $\sim_{hp}$ and hhp-bisimulation equivalence $\sim_{hhp}$, are all congruences with respect to APTC with guarded recursion.
\end{proof}

The $RDP$ (Recursive Definition Principle) and the $RSP$ (Recursive Specification Principle) are shown in Table \ref{RDPRSP}.

\begin{center}
\begin{table}
  \begin{tabular}{@{}ll@{}}
\hline No. &Axiom\\
  $RDP$ & $\langle X_i|E\rangle = t_i(\langle X_1|E\rangle,\cdots,\langle X_n|E\rangle)\quad (i\in\{1,\cdots,n\})$\\
  $RSP$ & if $y_i=t_i(y_1,\cdots,y_n)$ for $i\in\{1,\cdots,n\}$, then $y_i=\langle X_i|E\rangle \quad(i\in\{1,\cdots,n\})$\\
\end{tabular}
\caption{Recursive definition and specification principle}
\label{RDPRSP}
\end{table}
\end{center}

$RDP$ follows immediately from the two transition rules for guarded recursion, which express that $\langle X_i|E\rangle$ and $t_i(\langle X_1|E\rangle,\cdots,\langle X_n|E\rangle)$ have the same initial transitions for $i\in\{1,\cdots,n\}$. $RSP$ follows from the fact that guarded recursive specifications have only one solution.

\begin{theorem}[Elimination theorem of APTC with linear recursion]\label{ETRecursion}
Each process term in APTC with linear recursion is equal to a process term $\langle X_1|E\rangle$ with $E$ a linear recursive specification.
\end{theorem}

\begin{proof}
By applying structural induction with respect to term size, each process term $t_1$ in APTC with linear recursion generates a process can be expressed in the form of equations

$$t_i=(a_{i11}\leftmerge\cdots\leftmerge a_{i1i_1})t_{i1}+\cdots+(a_{ik_i1}\leftmerge\cdots\leftmerge a_{ik_ii_k})t_{ik_i}+(b_{i11}\leftmerge\cdots\leftmerge b_{i1i_1})+\cdots+(b_{il_i1}\leftmerge\cdots\leftmerge b_{il_ii_l})$$

for $i\in\{1,\cdots,n\}$. Let the linear recursive specification $E$ consist of the recursive equations

$$X_i=(a_{i11}\leftmerge\cdots\leftmerge a_{i1i_1})X_{i1}+\cdots+(a_{ik_i1}\leftmerge\cdots\leftmerge a_{ik_ii_k})X_{ik_i}+(b_{i11}\leftmerge\cdots\leftmerge b_{i1i_1})+\cdots+(b_{il_i1}\leftmerge\cdots\leftmerge b_{il_ii_l})$$

for $i\in\{1,\cdots,n\}$. Replacing $X_i$ by $t_i$ for $i\in\{1,\cdots,n\}$ is a solution for $E$, $RSP$ yields $t_1=\langle X_1|E\rangle$.
\end{proof}

\begin{theorem}[Soundness of APTC with guarded recursion]\label{SAPTCR}
Let $x$ and $y$ be APTC with guarded recursion terms. If $APTC\textrm{ with guarded recursion}\vdash x=y$, then
\begin{enumerate}
  \item $x\sim_{p} y$;
  \item $x\sim_{s} y$;
  \item $x\sim_{hp} y$;
  \item $x\sim_{hhp} y$.
\end{enumerate}
\end{theorem}

\begin{proof}
Since truly concurrent bisimulation equivalences, including pomset bisimulation $\sim_{p}$, step bisimulation $\sim_s$, hp-bisimulation $\sim_{hp}$ and hhp-bisimulation $\sim_{hhp}$, are all both equivalent and congruent relations with respect to APTC with guarded recursion, we only need to check if each axiom in Table \ref{RDPRSP} is sound modulo truly concurrent bisimulation equivalences. The proof is trivial and we omit it.
\end{proof}

\begin{theorem}[Completeness of APTC with linear recursion]\label{CAPTCR}
Let $p$ and $q$ be closed APTC with linear recursion terms, then,

\begin{enumerate}
  \item if $p\sim_{s} q$ then $p=q$;
  \item if $p\sim_{p} q$ then $p=q$;
  \item if $p\sim_{hp} q$ then $p=q$;
  \item if $p\sim_{hhp} q$ then $p=q$.
\end{enumerate}
\end{theorem}

\begin{proof}
Firstly, by the elimination theorem of APTC with guarded recursion (see Theorem \ref{ETRecursion}), we know that each process term in APTC with linear recursion is equal to a process term $\langle X_1|E\rangle$ with $E$ a linear recursive specification.

It remains to prove the following cases.

(1) If $\langle X_1|E_1\rangle \sim_s \langle Y_1|E_2\rangle$ for linear recursive specification $E_1$ and $E_2$, then $\langle X_1|E_1\rangle = \langle Y_1|E_2\rangle$.

Let $E_1$ consist of recursive equations $X=t_X$ for $X\in \mathcal{X}$ and $E_2$
consists of recursion equations $Y=t_Y$ for $Y\in\mathcal{Y}$. Let the linear recursive specification $E$ consist of recursion equations $Z_{XY}=t_{XY}$, and $\langle X|E_1\rangle\sim_s\langle Y|E_2\rangle$, and $t_{XY}$ consists of the following summands:

\begin{enumerate}
  \item $t_{XY}$ contains a summand $(a_1\leftmerge\cdots\leftmerge a_m)Z_{X'Y'}$ iff $t_X$ contains the summand $(a_1\leftmerge\cdots\leftmerge a_m)X'$ and $t_Y$ contains the summand $(a_1\leftmerge\cdots\leftmerge a_m)Y'$ such that $\langle X'|E_1\rangle\sim_s\langle Y'|E_2\rangle$;
  \item $t_{XY}$ contains a summand $b_1\leftmerge\cdots\leftmerge b_n$ iff $t_X$ contains the summand $b_1\leftmerge\cdots\leftmerge b_n$ and $t_Y$ contains the summand $b_1\leftmerge\cdots\leftmerge b_n$.
\end{enumerate}

Let $\sigma$ map recursion variable $X$ in $E_1$ to $\langle X|E_1\rangle$, and let $\psi$ map recursion variable $Z_{XY}$ in $E$ to $\langle X|E_1\rangle$. So, $\sigma((a_1\leftmerge\cdots\leftmerge a_m)X')\equiv(a_1\leftmerge\cdots\leftmerge a_m)\langle X'|E_1\rangle\equiv\psi((a_1\leftmerge\cdots\leftmerge a_m)Z_{X'Y'})$, so by $RDP$, we get $\langle X|E_1\rangle=\sigma(t_X)=\psi(t_{XY})$. Then by $RSP$, $\langle X|E_1\rangle=\langle Z_{XY}|E\rangle$, particularly, $\langle X_1|E_1\rangle=\langle Z_{X_1Y_1}|E\rangle$. Similarly, we can obtain $\langle Y_1|E_2\rangle=\langle Z_{X_1Y_1}|E\rangle$. Finally, $\langle X_1|E_1\rangle=\langle Z_{X_1Y_1}|E\rangle=\langle Y_1|E_2\rangle$, as desired.

(2) If $\langle X_1|E_1\rangle \sim_p \langle Y_1|E_2\rangle$ for linear recursive specification $E_1$ and $E_2$, then $\langle X_1|E_1\rangle = \langle Y_1|E_2\rangle$.

It can be proven similarly to (1), we omit it.

(3) If $\langle X_1|E_1\rangle \sim_{hp} \langle Y_1|E_2\rangle$ for linear recursive specification $E_1$ and $E_2$, then $\langle X_1|E_1\rangle = \langle Y_1|E_2\rangle$.

It can be proven similarly to (1), we omit it.

(4) If $\langle X_1|E_1\rangle \sim_{hhp} \langle Y_1|E_2\rangle$ for linear recursive specification $E_1$ and $E_2$, then $\langle X_1|E_1\rangle = \langle Y_1|E_2\rangle$.

It can be proven similarly to (1), we omit it.
\end{proof}

Then, we introduce approximation induction principle ($AIP$) and try to explain that $AIP$ is still valid in true concurrency. $AIP$ can be used to try and equate truly concurrent bisimilar guarded recursive specifications. $AIP$ says that if two process terms are truly concurrent bisimilar up to any finite depth, then they are truly concurrent bisimilar.

Also, we need the auxiliary unary projection operator $\Pi_n$ for $n\in\mathbb{N}$ and $\mathbb{N}\triangleq\{0,1,2,\cdots\}$. The transition rules of $\Pi_n$ are expressed in Table \ref{TRForProjection}.

\begin{center}
    \begin{table}
        $$\frac{x\xrightarrow{\{a_1,\cdots,a_k\}}\surd}{\Pi_{n+1}(x)\xrightarrow{\{a_1,\cdots,a_k\}}\surd}
        \quad\frac{x\xrightarrow{\{a_1,\cdots,a_k\}}x'}{\Pi_{n+1}(x)\xrightarrow{\{a_1,\cdots,a_k\}}\Pi_n(x')}$$
        \caption{Transition rules of projection operator $\Pi_n$}
        \label{TRForProjection}
    \end{table}
\end{center}

Based on the transition rules for projection operator $\Pi_n$ in Table \ref{TRForProjection}, we design the axioms as Table \ref{AxiomsForProjection} shows.

\begin{center}
    \begin{table}
        \begin{tabular}{@{}ll@{}}
            \hline No. &Axiom\\
            $PR1$ & $\Pi_n(x+y)=\Pi_n(x)+\Pi_n(y)$\\
            $PR2$ & $\Pi_n(x\leftmerge y)=\Pi_n(x)\leftmerge \Pi_n(y)$\\
            $PR3$ & $\Pi_{n+1}(a_1\leftmerge\cdots\leftmerge a_k)=a_1\leftmerge\cdots\leftmerge a_k$\\
            $PR4$ & $\Pi_{n+1}((a_1\leftmerge\cdots\leftmerge a_k)\cdot x)=(a_1\leftmerge\cdots\leftmerge a_k)\cdot\Pi_n(x)$\\
            $PR5$ & $\Pi_0(x)=0$\\
            $PR6$ & $\Pi_n(0)=0$\\
            $PR7$ & $\Pi_n(1)=1$\\
        \end{tabular}
        \caption{Axioms of projection operator}
        \label{AxiomsForProjection}
    \end{table}
\end{center}

\begin{theorem}[Conservativity of APTC with projection operator and guarded recursion]
APTC with projection operator and guarded recursion is a conservative extension of APTC with guarded recursion.
\end{theorem}

\begin{proof}
It follows from the following two facts:

\begin{enumerate}
  \item The transition rules of APTC with guarded recursion are all source-dependent;
  \item The sources of the transition rules for the projection operator contain an occurrence of $\Pi_n$.
\end{enumerate}
\end{proof}

\begin{theorem}[Congruence theorem of projection operator $\Pi_n$]
Truly concurrent bisimulation equivalences $\sim_{p}$, $\sim_s$, $\sim_{hp}$ and $\sim_{hhp}$ are all congruences with respect to projection operator $\Pi_n$.
\end{theorem}

\begin{proof}
Since the TSS of APTC with projection operator in Table \ref{TRForProjection} is positive after reduction and in panth format, according to Theorem \ref{tcbac}, truly concurrent bisimulation equivalences, including pomset bisimulation equivalence $\sim_{p}$, step bisimulation equivalence $\sim_s$, hp-bisimulation equivalence $\sim_{hp}$ and hhp-bisimulation equivalence $\sim_{hhp}$, are all congruences with respect to APTC with projection operator.
\end{proof}

\begin{theorem}[Elimination theorem of APTC with linear recursion and projection operator]\label{ETProjection}
Each process term in APTC with linear recursion and projection operator is equal to a process term $\langle X_1|E\rangle$ with $E$ a linear recursive specification.
\end{theorem}

\begin{proof}
By applying structural induction with respect to term size, each process term $t_1$ in APTC with linear recursion and projection operator $\Pi_n$ generates a process can be expressed in the form of equations

$$t_i=(a_{i11}\leftmerge\cdots\leftmerge a_{i1i_1})t_{i1}+\cdots+(a_{ik_i1}\leftmerge\cdots\leftmerge a_{ik_ii_k})t_{ik_i}+(b_{i11}\leftmerge\cdots\leftmerge b_{i1i_1})+\cdots+(b_{il_i1}\leftmerge\cdots\leftmerge b_{il_ii_l})$$

for $i\in\{1,\cdots,n\}$. Let the linear recursive specification $E$ consist of the recursive equations

$$X_i=(a_{i11}\leftmerge\cdots\leftmerge a_{i1i_1})X_{i1}+\cdots+(a_{ik_i1}\leftmerge\cdots\leftmerge a_{ik_ii_k})X_{ik_i}+(b_{i11}\leftmerge\cdots\leftmerge b_{i1i_1})+\cdots+(b_{il_i1}\leftmerge\cdots\leftmerge b_{il_ii_l})$$

for $i\in\{1,\cdots,n\}$. Replacing $X_i$ by $t_i$ for $i\in\{1,\cdots,n\}$ is a solution for $E$, $RSP$ yields $t_1=\langle X_1|E\rangle$.

That is, in $E$, there is not the occurrence of projection operator $\Pi_n$.
\end{proof}

\begin{theorem}[Soundness of APTC with projection operator and guarded recursion]\label{SAPTCRP}
Let $x$ and $y$ be APTC with projection operator and guarded recursion terms. If APTC with projection operator and guarded recursion $\vdash x=y$, then
\begin{enumerate}
  \item $x\sim_{p} y$;
  \item $x\sim_{s} y$;
  \item $x\sim_{hp} y$;
  \item $x\sim_{hhp} y$.
\end{enumerate}
\end{theorem}

\begin{proof}
Since truly concurrent bisimulation equivalences, including pomset bisimulation $\sim_{p}$, step bisimulation $\sim_s$, hp-bisimulation $\sim_{hp}$ and hhp-bisimulation $\sim_{hhp}$, are all both equivalent and congruent relations with respect to APTC with projection operator and guarded recursion, we only need to check if each axiom in Table \ref{AxiomsForProjection} is sound modulo truly concurrent bisimulation equivalences. The proof is trivial and we omit it.
\end{proof}

Then $AIP$ is given in Table \ref{AIP}.

\begin{center}
    \begin{table}
        \begin{tabular}{@{}ll@{}}
            \hline No. &Axiom\\
            $AIP$ & if $\Pi_n(x)=\Pi_n(y)$ for $n\in\mathbb{N}$, then $x=y$\\
        \end{tabular}
        \caption{$AIP$}
        \label{AIP}
    \end{table}
\end{center}

\begin{theorem}[Soundness of $AIP$]\label{SAIP}
Let $x$ and $y$ be APTC with projection operator and guarded recursion terms.

\begin{enumerate}
  \item If $\Pi_n(x)\sim_s\Pi_n(y)$ for $n\in\mathbb{N}$, then $x\sim_s y$;
  \item If $\Pi_n(x)\sim_p\Pi_n(y)$ for $n\in\mathbb{N}$, then $x\sim_p y$;
  \item If $\Pi_n(x)\sim_{hp}\Pi_n(y)$ for $n\in\mathbb{N}$, then $x\sim_{hp} y$;
  \item If $\Pi_n(x)\sim_{hhp}\Pi_n(y)$ for $n\in\mathbb{N}$, then $x\sim_{hhp} y$
\end{enumerate}
\end{theorem}

\begin{proof}
(1) If $\Pi_n(x)\sim_s\Pi_n(y)$ for $n\in\mathbb{N}$, then $x\sim_s y$.

Since step bisimulation $\sim_{s}$ is both an equivalent and a congruent relation with respect to APTC with guarded recursion and projection operator, we only need to check if $AIP$ in Table \ref{AIP} is sound modulo step bisimulation equivalence.

Let $p,p_0$ and $q,q_0$ be closed APTC with projection operator and guarded recursion terms such that $\Pi_n(p_0)\sim_s\Pi_n(q_0)$ for $n\in\mathbb{N}$. We define a relation $R$ such that $p R q$ iff $\Pi_n(p)\sim_s \Pi_n(q)$. Obviously, $p_0 R q_0$, next, we prove that $R\in\sim_s$.

Let $p R q$ and $p\xrightarrow{\{a_1,\cdots,a_k\}}\surd$, then $\Pi_1(p)\xrightarrow{\{a_1,\cdots,a_k\}}\surd$, $\Pi_1(p)\sim_s\Pi_1(q)$ yields $\Pi_1(q)\xrightarrow{\{a_1,\cdots,a_k\}}\surd$. Similarly, $q\xrightarrow{\{a_1,\cdots,a_k\}}\surd$ implies $p\xrightarrow{\{a_1,\cdots,a_k\}}\surd$.

Let $p R q$ and $p\xrightarrow{\{a_1,\cdots,a_k\}}p'$. We define the set of process terms

$$S_n\triangleq\{q'|q\xrightarrow{\{a_1,\cdots,a_k\}}q'\textrm{ and }\Pi_n(p')\sim_s\Pi_n(q')\}$$

\begin{enumerate}
  \item Since $\Pi_{n+1}(p)\sim_s\Pi_{n+1}(q)$ and $\Pi_{n+1}(p)\xrightarrow{\{a_1,\cdots,a_k\}}\Pi_n(p')$, there exist $q'$ such that $\Pi_{n+1}(q)\xrightarrow{\{a_1,\cdots,a_k\}}\Pi_n(q')$ and $\Pi_{n}(p')\sim_s\Pi_{n}(q')$. So, $S_n$ is not empty.
  \item There are only finitely many $q'$ such that $q\xrightarrow{\{a_1,\cdots,a_k\}}q'$, so, $S_n$ is finite.
  \item $\Pi_{n+1}(p)\sim_s\Pi_{n+1}(q)$ implies $\Pi_{n}(p')\sim_s\Pi_{n}(q')$, so $S_n\supseteq S_{n+1}$.
\end{enumerate}

So, $S_n$ has a non-empty intersection, and let $q'$ be in this intersection, then $q\xrightarrow{\{a_1,\cdots,a_k\}}q'$ and $\Pi_n(p')\sim_s\Pi_n(q')$, so $p' R q'$. Similarly, let $p\mathcal{q}q$, we can obtain $q\xrightarrow{\{a_1,\cdots,a_k\}}q'$ implies $p\xrightarrow{\{a_1,\cdots,a_k\}}p'$ such that $p' R q'$.

Finally, $R\in\sim_s$ and $p_0\sim_s q_0$, as desired.

(2) If $\Pi_n(x)\sim_p\Pi_n(y)$ for $n\in\mathbb{N}$, then $x\sim_p y$.

Since pomset bisimulation $\sim_{p}$ is both an equivalent and a congruent relation with respect to APTC with guarded recursion and projection operator, we only need to check if $AIP$ in Table \ref{AIP} is sound modulo pomset bisimulation equivalence.
The proof is trivial and we omit it.

(3) If $\Pi_n(x)\sim_{hp}\Pi_n(y)$ for $n\in\mathbb{N}$, then $x\sim_{hp} y$.

Since hp-bisimulation $\sim_{hp}$ is both an equivalent and a congruent relation with respect to APTC with guarded recursion and projection operator, we only need to check if $AIP$ in Table \ref{AIP} is sound modulo hp-bisimulation equivalence.
The proof is trivial and we omit it.

(4) If $\Pi_n(x)\sim_{hhp}\Pi_n(y)$ for $n\in\mathbb{N}$, then $x\sim_{hhp} y$.

Since hhp-bisimulation $\sim_{hhp}$ is both an equivalent and a congruent relation with respect to APTC with guarded recursion and projection operator, we only need to check if $AIP$ in Table \ref{AIP} is sound modulo hhp-bisimulation equivalence.
The proof is trivial and we omit it.
\end{proof}

\begin{theorem}[Completeness of $AIP$]\label{CAIP}
Let $p$ and $q$ be closed APTC with linear recursion and projection operator terms, then,
\begin{enumerate}
  \item if $p\sim_{s} q$ then $\Pi_n(p)=\Pi_n(q)$;
  \item if $p\sim_{p} q$ then $\Pi_n(p)=\Pi_n(q)$;
  \item if $p\sim_{hp} q$ then $\Pi_n(p)=\Pi_n(q)$;
  \item if $p\sim_{hhp} q$ then $\Pi_n(p)=\Pi_n(q)$.
\end{enumerate}
\end{theorem}

\begin{proof}
Firstly, by the elimination theorem of APTC with guarded recursion and projection operator (see Theorem \ref{ETProjection}), we know that each process term in APTC with linear recursion and projection operator is equal to a process term $\langle X_1|E\rangle$ with $E$ a linear recursive specification:

$$X_i=(a_{i11}\leftmerge\cdots\leftmerge a_{i1i_1})X_{i1}+\cdots+(a_{ik_i1}\leftmerge\cdots\leftmerge a_{ik_ii_k})X_{ik_i}+(b_{i11}\leftmerge\cdots\leftmerge b_{i1i_1})+\cdots+(b_{il_i1}\leftmerge\cdots\leftmerge b_{il_ii_l})$$

for $i\in\{1,\cdots,n\}$.

It remains to prove the following cases.

(1) if $p\sim_{s} q$ then $\Pi_n(p)=\Pi_n(q)$.

Let $p\sim_s q$, and fix an $n\in\mathbb{N}$, there are $p',q'$ in basic APTC terms such that $p'=\Pi_n(p)$ and $q'=\Pi_n(q)$. Since $\sim_s$ is a congruence with respect to APTC, if $p\sim_s q$ then $\Pi_n(p)\sim_s\Pi_n(q)$. The soundness theorem yields $p'\sim_s\Pi_n(p)\sim_s\Pi_n(q)\sim_s q'$. Finally, the completeness of APTC modulo $\sim_s$ (see Theorem \ref{SAPTC2}) ensures $p'=q'$, and $\Pi_n(p)=p'=q'=\Pi_n(q)$, as desired.

(2) if $p\sim_{p} q$ then $\Pi_n(p)=\Pi_n(q)$.

Let $p\sim_p q$, and fix an $n\in\mathbb{N}$, there are $p',q'$ in basic APTC terms such that $p'=\Pi_n(p)$ and $q'=\Pi_n(q)$. Since $\sim_p$ is a congruence with respect to APTC, if $p\sim_p q$ then $\Pi_n(p)\sim_p\Pi_n(q)$. The soundness theorem yields $p'\sim_p\Pi_n(p)\sim_p\Pi_n(q)\sim_p q'$. Finally, the completeness of APTC modulo $\sim_p$ (see Theorem \ref{SAPTC2}) ensures $p'=q'$, and $\Pi_n(p)=p'=q'=\Pi_n(q)$, as desired.

(3) if $p\sim_{hp} q$ then $\Pi_n(p)=\Pi_n(q)$.

Let $p\sim_{hp} q$, and fix an $n\in\mathbb{N}$, there are $p',q'$ in basic APTC terms such that $p'=\Pi_n(p)$ and $q'=\Pi_n(q)$. Since $\sim_{hp}$ is a congruence with respect to APTC, if $p\sim_{hp} q$ then $\Pi_n(p)\sim_{hp}\Pi_n(q)$. The soundness theorem yields $p'\sim_{hp}\Pi_n(p)\sim_{hp}\Pi_n(q)\sim_{hp} q'$. Finally, the completeness of APTC modulo $\sim_{hp}$ (see Theorem \ref{SAPTC2}) ensures $p'=q'$, and $\Pi_n(p)=p'=q'=\Pi_n(q)$, as desired.

(3) if $p\sim_{hhp} q$ then $\Pi_n(p)=\Pi_n(q)$.

Let $p\sim_{hhp} q$, and fix an $n\in\mathbb{N}$, there are $p',q'$ in basic APTC terms such that $p'=\Pi_n(p)$ and $q'=\Pi_n(q)$. Since $\sim_{hhp}$ is a congruence with respect to APTC, if $p\sim_{hhp} q$ then $\Pi_n(p)\sim_{hhp}\Pi_n(q)$. The soundness theorem yields $p'\sim_{hhp}\Pi_n(p)\sim_{hhp}\Pi_n(q)\sim_{hhp} q'$. Finally, the completeness of APTC modulo $\sim_{hhp}$ (see Theorem \ref{SAPTC2}) ensures $p'=q'$, and $\Pi_n(p)=p'=q'=\Pi_n(q)$, as desired.
\end{proof}

\subsubsection{Abstraction}

To abstract away from the internal implementations of a program, and verify that the program exhibits the desired external behaviors, the silent step $\tau$ and abstraction operator $\tau_I$ are introduced, where $I\subseteq \mathbb{E}$ denotes the internal events. The silent step $\tau$ represents the internal events, when we consider the external behaviors of a process, $\tau$ events can be removed, that is, $\tau$ events must keep silent. The transition rule of $\tau$ is shown in Table \ref{TRForTau}. In the following, let the atomic event $a$ range over $\mathbb{E}\cup\{0\}\cup\{1\}\cup\{\tau\}$, and let the communication function $\gamma:\mathbb{E}\cup\{\tau\}\times \mathbb{E}\cup\{\tau\}\rightarrow \mathbb{E}\cup\{0\}$, with each communication involved $\tau$ resulting into $0$.

\begin{center}
    \begin{table}
        $$\frac{}{\tau\xrightarrow{\tau}\surd}$$
        \caption{Transition rule of the silent step}
        \label{TRForTau}
    \end{table}
\end{center}

The silent step $\tau$ as an atomic event, is introduced into $E$. Considering the recursive specification $X=\tau X$, $\tau s$, $\tau\tau s$, and $\tau\cdots s$ are all its solutions, that is, the solutions make the existence of $\tau$-loops which cause unfairness. To prevent $\tau$-loops, we extend the definition of linear recursive specification (Definition \ref{LRS}) to the guarded one.

\begin{definition}[Guarded linear recursive specification]\label{GLRS}
A recursive specification is linear if its recursive equations are of the form

$$(a_{11}\leftmerge\cdots\leftmerge a_{1i_1})X_1+\cdots+(a_{k1}\leftmerge\cdots\leftmerge a_{ki_k})X_k+(b_{11}\leftmerge\cdots\leftmerge b_{1j_1})+\cdots+(b_{1j_1}\leftmerge\cdots\leftmerge b_{lj_l})$$

where $a_{11},\cdots,a_{1i_1},a_{k1},\cdots,a_{ki_k},b_{11},\cdots,b_{1j_1},b_{1j_1},\cdots,b_{lj_l}\in \mathbb{E}\cup\{\tau\}$, and the sum above is allowed to be empty, in which case it represents the deadlock $0$.

A linear recursive specification $E$ is guarded if there does not exist an infinite sequence of $1$-transitions $\langle X|E\rangle\rightarrow\langle X'|E\rangle\rightarrow\langle X''|E\rangle\rightarrow\cdots$ and $\tau$-transitions $\langle X|E\rangle\xrightarrow{\tau}\langle X'|E\rangle\xrightarrow{\tau}\langle X''|E\rangle\xrightarrow{\tau}\cdots$.
\end{definition}

\begin{theorem}[Conservativity of APTC with silent step and guarded linear recursion]
APTC with silent step and guarded linear recursion is a conservative extension of APTC with linear recursion.
\end{theorem}

\begin{proof}
Since the transition rules of APTC with linear recursion are source-dependent, and the transition rules for silent step in Table \ref{TRForTau} contain only a fresh constant $\tau$ in their source, so the transition rules of APTC with silent step and guarded linear recursion is a conservative extension of those of APTC with linear recursion.
\end{proof}

\begin{theorem}[Congruence theorem of APTC with silent step and guarded linear recursion]
Rooted branching truly concurrent bisimulation equivalences $\approx_{rbp}$, $\approx_{rbs}$, $\approx_{rbhp}$ and $\approx_{rbhhp}$ are all congruences with respect to APTC with silent step and guarded linear recursion.
\end{theorem}

\begin{proof}
Since the TSS of APTC with silent step and guarded linear recursion in Table \ref{TRForTau} is positive after reduction and in RBB cool format, according to Theorem \ref{rbtcbac}, rooted branching truly concurrent bisimulation equivalences $\approx_{rbp}$, $\approx_{rbs}$, $\approx_{rbhp}$ and $\approx_{rbhhp}$ are all congruences with respect to APTC with with silent step and guarded linear recursion.
\end{proof}

We design the axioms for the silent step $\tau$ in Table \ref{AxiomsForTau}.

\begin{center}
\begin{table}
  \begin{tabular}{@{}ll@{}}
\hline No. &Axiom\\
  $B1$ & $x\cdot\tau=x$\\
  $B2$ & $a\cdot(\tau\cdot(x+y)+x)=a\cdot(x+y)$\\
  $B3$ & $x\leftmerge\tau=x$\\
\end{tabular}
\caption{Axioms of silent step}
\label{AxiomsForTau}
\end{table}
\end{center}

\begin{theorem}[Elimination theorem of APTC with silent step and guarded linear recursion]\label{ETTau}
Each process term in APTC with silent step and guarded linear recursion is equal to a process term $\langle X_1|E\rangle$ with $E$ a guarded linear recursive specification.
\end{theorem}

\begin{proof}
By applying structural induction with respect to term size, each process term $t_1$ in APTC with silent step and guarded linear recursion generates a process can be expressed in the form of equations

$$t_i=(a_{i11}\leftmerge\cdots\leftmerge a_{i1i_1})t_{i1}+\cdots+(a_{ik_i1}\leftmerge\cdots\leftmerge a_{ik_ii_k})t_{ik_i}+(b_{i11}\leftmerge\cdots\leftmerge b_{i1i_1})+\cdots+(b_{il_i1}\leftmerge\cdots\leftmerge b_{il_ii_l})$$

for $i\in\{1,\cdots,n\}$. Let the linear recursive specification $E$ consist of the recursive equations

$$X_i=(a_{i11}\leftmerge\cdots\leftmerge a_{i1i_1})X_{i1}+\cdots+(a_{ik_i1}\leftmerge\cdots\leftmerge a_{ik_ii_k})X_{ik_i}+(b_{i11}\leftmerge\cdots\leftmerge b_{i1i_1})+\cdots+(b_{il_i1}\leftmerge\cdots\leftmerge b_{il_ii_l})$$

for $i\in\{1,\cdots,n\}$. Replacing $X_i$ by $t_i$ for $i\in\{1,\cdots,n\}$ is a solution for $E$, $RSP$ yields $t_1=\langle X_1|E\rangle$.
\end{proof}

\begin{theorem}[Soundness of APTC with silent step and guarded linear recursion]\label{SAPTCTAU}
Let $x$ and $y$ be APTC with silent step and guarded linear recursion terms. If APTC with silent step and guarded linear recursion $\vdash x=y$, then

\begin{enumerate}
  \item $x\approx_{rbs} y$;
  \item $x\approx_{rbp} y$;
  \item $x\approx_{rbhp} y$;
  \item $x\approx_{rbhhp} y$
\end{enumerate}
\end{theorem}

\begin{proof}
Since rooted branching truly concurrent bisimulation $\approx_{rbp}$, $\approx_{rbs}$, $\approx_{rbhp}$ and $\approx_{rbhhp}$ are all both equivalent and congruent relations with respect to APTC with silent step and guarded linear recursion, we only need to check if each axiom in Table \ref{AxiomsForTau} is sound modulo rooted branching rooted branching truly concurrent bisimulation equivalences. The proof is trivial and we omit it.
\end{proof}

\begin{theorem}[Completeness of APTC with silent step and guarded linear recursion]\label{CAPTCTAU}
Let $p$ and $q$ be closed APTC with silent step and guarded linear recursion terms, then,

\begin{enumerate}
  \item if $p\approx_{rbs} q$ then $p=q$;
  \item if $p\approx_{rbp} q$ then $p=q$;
  \item if $p\approx_{rbhp} q$ then $p=q$;
  \item if $p\approx_{rbhhp} q$ then $p=q$
\end{enumerate}
\end{theorem}

\begin{proof}
Firstly, by the elimination theorem of APTC with silent step and guarded linear recursion (see Theorem \ref{ETTau}), we know that each process term in APTC with silent step and guarded linear recursion is equal to a process term $\langle X_1|E\rangle$ with $E$ a guarded linear recursive specification.

It remains to prove the following cases.

(1) If $\langle X_1|E_1\rangle \approx_{rbs} \langle Y_1|E_2\rangle$ for guarded linear recursive specification $E_1$ and $E_2$, then $\langle X_1|E_1\rangle = \langle Y_1|E_2\rangle$.

Firstly, the recursive equation $W=\tau+\cdots+\tau$ with $W\nequiv X_1$ in $E_1$ and $E_2$, can be removed, and the corresponding summands $aW$ are replaced by $a$, to get $E_1'$ and $E_2'$, by use of the axioms $RDP$, $A3$ and $B1$, and $\langle X|E_1\rangle = \langle X|E_1'\rangle$, $\langle Y|E_2\rangle = \langle Y|E_2'\rangle$.

Let $E_1$ consists of recursive equations $X=t_X$ for $X\in \mathcal{X}$ and $E_2$
consists of recursion equations $Y=t_Y$ for $Y\in\mathcal{Y}$, and are not the form $\tau+\cdots+\tau$. Let the guarded linear recursive specification $E$ consists of recursion equations $Z_{XY}=t_{XY}$, and $\langle X|E_1\rangle\approx_{rbs}\langle Y|E_2\rangle$, and $t_{XY}$ consists of the following summands:

\begin{enumerate}
  \item $t_{XY}$ contains a summand $(a_1\leftmerge\cdots\leftmerge a_m)Z_{X'Y'}$ iff $t_X$ contains the summand $(a_1\leftmerge\cdots\leftmerge a_m)X'$ and $t_Y$ contains the summand $(a_1\leftmerge\cdots\leftmerge a_m)Y'$ such that $\langle X'|E_1\rangle\approx_{rbs}\langle Y'|E_2\rangle$;
  \item $t_{XY}$ contains a summand $b_1\leftmerge\cdots\leftmerge b_n$ iff $t_X$ contains the summand $b_1\leftmerge\cdots\leftmerge b_n$ and $t_Y$ contains the summand $b_1\leftmerge\cdots\leftmerge b_n$;
  \item $t_{XY}$ contains a summand $\tau Z_{X'Y}$ iff $XY\nequiv X_1Y_1$, $t_X$ contains the summand $\tau X'$, and $\langle X'|E_1\rangle\approx_{rbs}\langle Y|E_2\rangle$;
  \item $t_{XY}$ contains a summand $\tau Z_{XY'}$ iff $XY\nequiv X_1Y_1$, $t_Y$ contains the summand $\tau Y'$, and $\langle X|E_1\rangle\approx_{rbs}\langle Y'|E_2\rangle$.
\end{enumerate}

Since $E_1$ and $E_2$ are guarded, $E$ is guarded. Constructing the process term $u_{XY}$ consist of the following summands:

\begin{enumerate}
  \item $u_{XY}$ contains a summand $(a_1\leftmerge\cdots\leftmerge a_m)\langle X'|E_1\rangle$ iff $t_X$ contains the summand $(a_1\leftmerge\cdots\leftmerge a_m)X'$ and $t_Y$ contains the summand $(a_1\leftmerge\cdots\leftmerge a_m)Y'$ such that $\langle X'|E_1\rangle\approx_{rbs}\langle Y'|E_2\rangle$;
  \item $u_{XY}$ contains a summand $b_1\leftmerge\cdots\leftmerge b_n$ iff $t_X$ contains the summand $b_1\leftmerge\cdots\leftmerge b_n$ and $t_Y$ contains the summand $b_1\leftmerge\cdots\leftmerge b_n$;
  \item $u_{XY}$ contains a summand $\tau \langle X'|E_1\rangle$ iff $XY\nequiv X_1Y_1$, $t_X$ contains the summand $\tau X'$, and $\langle X'|E_1\rangle\approx_{rbs}\langle Y|E_2\rangle$.
\end{enumerate}

Let the process term $s_{XY}$ be defined as follows:

\begin{enumerate}
  \item $s_{XY}\triangleq\tau\langle X|E_1\rangle + u_{XY}$ iff $XY\nequiv X_1Y_1$, $t_Y$ contains the summand $\tau Y'$, and $\langle X|E_1\rangle\approx_{rbs}\langle Y'|E_2\rangle$;
  \item $s_{XY}\triangleq\langle X|E_1\rangle$, otherwise.
\end{enumerate}

So, $\langle X|E_1\rangle=\langle X|E_1\rangle+u_{XY}$, and $(a_1\leftmerge\cdots\leftmerge a_m)(\tau\langle X|E_1\rangle+u_{XY})=(a_1\leftmerge\cdots\leftmerge a_m)((\tau\langle X|E_1\rangle+u_{XY})+u_{XY})=(a_1\leftmerge\cdots\leftmerge a_m)(\langle X|E_1\rangle+u_{XY})=(a_1\leftmerge\cdots\leftmerge a_m)\langle X|E_1\rangle$, hence, $(a_1\leftmerge\cdots\leftmerge a_m)s_{XY}=(a_1\leftmerge\cdots\leftmerge a_m)\langle X|E_1\rangle$.

Let $\sigma$ map recursion variable $X$ in $E_1$ to $\langle X|E_1\rangle$, and let $\psi$ map recursion variable $Z_{XY}$ in $E$ to $s_{XY}$. It is sufficient to prove $s_{XY}=\psi(t_{XY})$ for recursion variables $Z_{XY}$ in $E$. Either $XY\equiv X_1Y_1$ or $XY\nequiv X_1Y_1$, we all can get $s_{XY}=\psi(t_{XY})$. So, $s_{XY}=\langle Z_{XY}|E\rangle$ for recursive variables $Z_{XY}$ in $E$ is a solution for $E$. Then by $RSP$, particularly, $\langle X_1|E_1\rangle=\langle Z_{X_1Y_1}|E\rangle$. Similarly, we can obtain $\langle Y_1|E_2\rangle=\langle Z_{X_1Y_1}|E\rangle$. Finally, $\langle X_1|E_1\rangle=\langle Z_{X_1Y_1}|E\rangle=\langle Y_1|E_2\rangle$, as desired.

(2) If $\langle X_1|E_1\rangle \approx_{rbp} \langle Y_1|E_2\rangle$ for guarded linear recursive specification $E_1$ and $E_2$, then $\langle X_1|E_1\rangle = \langle Y_1|E_2\rangle$.

It can be proven similarly to (1), we omit it.

(3) If $\langle X_1|E_1\rangle \approx_{rbhb} \langle Y_1|E_2\rangle$ for guarded linear recursive specification $E_1$ and $E_2$, then $\langle X_1|E_1\rangle = \langle Y_1|E_2\rangle$.

It can be proven similarly to (1), we omit it.

(4) If $\langle X_1|E_1\rangle \approx_{rbhhb} \langle Y_1|E_2\rangle$ for guarded linear recursive specification $E_1$ and $E_2$, then $\langle X_1|E_1\rangle = \langle Y_1|E_2\rangle$.

It can be proven similarly to (1), we omit it.
\end{proof}

The unary abstraction operator $\tau_I$ ($I\subseteq \mathbb{E}$) renames all atomic events in $I$ into $\tau$. APTC with silent step and abstraction operator is called $APTC_{\tau}$. The transition rules of operator $\tau_I$ are shown in Table \ref{TRForAbstraction}.

\begin{center}
    \begin{table}
        $$\frac{x\xrightarrow{a}\surd}{\tau_I(x)\xrightarrow{a}\surd}\quad a\notin I
        \quad\quad\frac{x\xrightarrow{a}x'}{\tau_I(x)\xrightarrow{a}\tau_I(x')}\quad a\notin I$$

        $$\frac{x\xrightarrow{a}\surd}{\tau_I(x)\xrightarrow{\tau}\surd}\quad a\in I
        \quad\quad\frac{x\xrightarrow{a}x'}{\tau_I(x)\xrightarrow{\tau}\tau_I(x')}\quad a\in I$$
        \caption{Transition rule of the abstraction operator}
        \label{TRForAbstraction}
    \end{table}
\end{center}

\begin{theorem}[Conservativity of $APTC_{\tau}$ with guarded linear recursion]
$APTC_{\tau}$ with guarded linear recursion is a conservative extension of APTC with silent step and guarded linear recursion.
\end{theorem}

\begin{proof}
Since the transition rules of APTC with silent step and guarded linear recursion are source-dependent, and the transition rules for abstraction operator in Table \ref{TRForAbstraction} contain only a fresh operator $\tau_I$ in their source, so the transition rules of $APTC_{\tau}$ with guarded linear recursion is a conservative extension of those of APTC with silent step and guarded linear recursion.
\end{proof}

\begin{theorem}[Congruence theorem of $APTC_{\tau}$ with guarded linear recursion]
Rooted branching truly concurrent bisimulation equivalences $\approx_{rbp}$, $\approx_{rbs}$, $\approx_{rbhp}$ and $\approx_{rbhhp}$ are all congruences with respect to $APTC_{\tau}$ with guarded linear recursion.
\end{theorem}

\begin{proof}
Since the TSS of $APTC_{\tau}$ with guarded linear recursion in Table \ref{TRForAbstraction} is positive after reduction and in RBB cool format, according to Theorem \ref{rbtcbac}, rooted branching truly concurrent bisimulation equivalences $\approx_{rbp}$, $\approx_{rbs}$, $\approx_{rbhp}$ and $\approx_{rbhhp}$ are all congruences with respect to $APTC_{\tau}$ with guarded linear recursion.
\end{proof}

We design the axioms for the abstraction operator $\tau_I$ in Table \ref{AxiomsForAbstraction}.

\begin{center}
\begin{table}
  \begin{tabular}{@{}ll@{}}
\hline No. &Axiom\\
  $TI1$ & $a\notin I\quad \tau_I(a)=a$\\
  $TI2$ & $a\in I\quad \tau_I(a)=\tau$\\
  $TI3$ & $\tau_I(0)=0$\\
  $TI4$ & $\tau_I(1)=1$\\
  $TI5$ & $\tau_I(x+y)=\tau_I(x)+\tau_I(y)$\\
  $TI6$ & $\tau_I(x\cdot y)=\tau_I(x)\cdot\tau_I(y)$\\
  $TI7$ & $\tau_I(x\leftmerge y)=\tau_I(x)\leftmerge\tau_I(y)$\\
\end{tabular}
\caption{Axioms of abstraction operator}
\label{AxiomsForAbstraction}
\end{table}
\end{center}

\begin{theorem}[Soundness of $APTC_{\tau}$ with guarded linear recursion]\label{SAPTCABS}
Let $x$ and $y$ be $APTC_{\tau}$ with guarded linear recursion terms. If $APTC_{\tau}$ with guarded linear recursion $\vdash x=y$, then

\begin{enumerate}
  \item $x\approx_{rbs} y$;
  \item $x\approx_{rbp} y$;
  \item $x\approx_{rbhp} y$;
  \item $x\approx_{rbhhp} y$.
\end{enumerate}
\end{theorem}

\begin{proof}
Since rooted branching truly concurrent bisimulations $\approx_{rbp}$, $\approx_{rbs}$, $\approx_{rbhp}$ and $\approx_{rbhhp}$ are all both equivalent and congruent relations with respect to $APTC_{\tau}$ with guarded linear recursion, we only need to check if each axiom in Table \ref{AxiomsForAbstraction} is sound modulo rooted branching truly concurrent bisimulation equivalences. The proof is trivial and we omit it.
\end{proof}

Though $\tau$-loops are prohibited in guarded linear recursive specifications (see Definition \ref{GLRS}) in a specifiable way, they can be constructed using the abstraction operator, for example, there exist $\tau$-loops in the process term $\tau_{\{a\}}(\langle X|X=aX\rangle)$. To avoid $\tau$-loops caused by $\tau_I$ and ensure fairness, the concept of cluster and $CFAR$ (Cluster Fair Abstraction Rule) \cite{CFAR} are still valid in true concurrency, we introduce them below.

\begin{definition}[Cluster]\label{CLUSTER}
Let $E$ be a guarded linear recursive specification, and $I\subseteq \mathbb{E}$. Two recursion variable $X$ and $Y$ in $E$ are in the same cluster for $I$ iff there exist sequences of transitions $\langle X|E\rangle\xrightarrow{\{b_{11},\cdots, b_{1i}\}}\cdots\xrightarrow{\{b_{m1},\cdots, b_{mi}\}}\langle Y|E\rangle$ and $\langle Y|E\rangle\xrightarrow{\{c_{11},\cdots, c_{1j}\}}\cdots\xrightarrow{\{c_{n1},\cdots, c_{nj}\}}\langle X|E\rangle$, where $b_{11},\cdots,b_{mi},c_{11},\cdots,c_{nj}\in I\cup\{\tau\}$.

$a_1\leftmerge\cdots\leftmerge a_k$ or $(a_1\leftmerge\cdots\leftmerge a_k) X$ is an exit for the cluster $C$ iff: (1) $a_1\leftmerge\cdots\leftmerge a_k$ or $(a_1\leftmerge\cdots\leftmerge a_k) X$ is a summand at the right-hand side of the recursive equation for a recursion variable in $C$, and (2) in the case of $(a_1\leftmerge\cdots\leftmerge a_k) X$, either $a_l\notin I\cup\{\tau\}(l\in\{1,2,\cdots,k\})$ or $X\notin C$.
\end{definition}

\begin{center}
\begin{table}
  \begin{tabular}{@{}ll@{}}
\hline No. &Axiom\\
  $CFAR$ & If $X$ is in a cluster for $I$ with exits \\
           & $\{(a_{11}\leftmerge\cdots\leftmerge a_{1i})Y_1,\cdots,(a_{m1}\leftmerge\cdots\leftmerge a_{mi})Y_m, b_{11}\leftmerge\cdots\leftmerge b_{1j},\cdots,b_{n1}\leftmerge\cdots\leftmerge b_{nj}\}$, \\
           & then $\tau\cdot\tau_I(\langle X|E\rangle)=$\\
           & $\tau\cdot\tau_I((a_{11}\leftmerge\cdots\leftmerge a_{1i})\langle Y_1|E\rangle+\cdots+(a_{m1}\leftmerge\cdots\leftmerge a_{mi})\langle Y_m|E\rangle+b_{11}\leftmerge\cdots\leftmerge b_{1j}+\cdots+b_{n1}\leftmerge\cdots\leftmerge b_{nj})$\\
\end{tabular}
\caption{Cluster fair abstraction rule}
\label{CFAR}
\end{table}
\end{center}

\begin{theorem}[Soundness of $CFAR$]\label{SCFAR}
$CFAR$ is sound modulo rooted branching truly concurrent bisimulation equivalences $\approx_{rbs}$, $\approx_{rbp}$, $\approx_{rbhp}$ and $\approx_{rbhhp}$.
\end{theorem}

\begin{proof}
(1) Soundness of $CFAR$ with respect to rooted branching step bisimulation $\approx_{rbs}$.

Let $X$ be in a cluster for $I$ with exits $\{(a_{11}\leftmerge\cdots\leftmerge a_{1i})Y_1,\cdots,(a_{m1}\leftmerge\cdots\leftmerge a_{mi})Y_m,b_{11}\leftmerge\cdots\leftmerge b_{1j},\cdots,b_{n1}\leftmerge\cdots\leftmerge b_{nj}\}$. Then $\langle X|E\rangle$ can execute a string of atomic events from $I\cup\{\tau\}$ inside the cluster of $X$, followed by an exit $(a_{i'1}\leftmerge\cdots\leftmerge a_{i'i})Y_{i'}$ for $i'\in\{1,\cdots,m\}$ or $b_{j'1}\leftmerge\cdots\leftmerge b_{j'j}$ for $j'\in\{1,\cdots,n\}$. Hence, $\tau_I(\langle X|E\rangle)$ can execute a string of $\tau^*$ inside the cluster of $X$, followed by an exit $\tau_I((a_{i'1}\leftmerge\cdots\leftmerge a_{i'i})\langle Y_{i'}|E\rangle)$ for $i'\in\{1,\cdots,m\}$ or $\tau_I(b_{j'1}\leftmerge\cdots\leftmerge b_{j'j})$ for $j'\in\{1,\cdots,n\}$. And these $\tau^*$ are non-initial in $\tau\tau_I(\langle X|E\rangle)$, so they are truly silent by the axiom $B1$, we obtain $\tau\tau_I(\langle X|E\rangle)\approx_{rbs}\tau\cdot\tau_I((a_{11}\leftmerge\cdots\leftmerge a_{1i})\langle Y_1|E\rangle+\cdots+(a_{m1}\leftmerge\cdots\leftmerge a_{mi})\langle Y_m|E\rangle+b_{11}\leftmerge\cdots\leftmerge b_{1j}+\cdots+b_{n1}\leftmerge\cdots\leftmerge b_{nj})$, as desired.

(2) Soundness of $CFAR$ with respect to rooted branching pomset bisimulation $\approx_{rbp}$.

Similarly to the proof of soundness of $CFAR$ modulo rooted branching step bisimulation $\approx_{rbs}$ (1), we can prove that $CFAR$ in Table \ref{CFAR} is sound modulo rooted branching pomset bisimulation $\approx_{rbp}$, we omit it.

(3) Soundness of $CFAR$ with respect to rooted branching hp-bisimulation $\approx_{rbhp}$.

Similarly to the proof of soundness of $CFAR$ modulo rooted branching pomset bisimulation equivalence (2), we can prove that $CFAR$ in Table \ref{CFAR} is sound modulo rooted branching hp-bisimulation equivalence, we omit it.

(4) Soundness of $CFAR$ with respect to rooted branching hhp-bisimulation $\approx_{rbhhp}$.

Similarly to the proof of soundness of $CFAR$ modulo rooted branching hp-bisimulation equivalence (3), we can prove that $CFAR$ in Table \ref{CFAR} is sound modulo rooted branching hhp-bisimulation equivalence, we omit it.
\end{proof}

\begin{theorem}[Completeness of $APTC_{\tau}$ with guarded linear recursion and $CFAR$]\label{CCFAR}
Let $p$ and $q$ be closed $APTC_{\tau}$ with guarded linear recursion and $CFAR$ terms, then,

\begin{enumerate}
  \item if $p\approx_{rbs} q$ then $p=q$;
  \item if $p\approx_{rbp} q$ then $p=q$;
  \item if $p\approx_{rbhp} q$ then $p=q$;
  \item if $p\approx_{rbhhp} q$ then $p=q$.
\end{enumerate}
\end{theorem}

\begin{proof}
(1) For the case of rooted branching step bisimulation, the proof is following.

Firstly, in the proof the Theorem \ref{CAPTCTAU}, we know that each process term $p$ in APTC with silent step and guarded linear recursion is equal to a process term $\langle X_1|E\rangle$ with $E$ a guarded linear recursive specification. And we prove if $\langle X_1|E_1\rangle\approx_{rbs}\langle Y_1|E_2\rangle$, then $\langle X_1|E_1\rangle=\langle Y_1|E_2\rangle$

The only new case is $p\equiv\tau_I(q)$. Let $q=\langle X|E\rangle$ with $E$ a guarded linear recursive specification, so $p=\tau_I(\langle X|E\rangle)$. Then the collection of recursive variables in $E$ can be divided into its clusters $C_1,\cdots,C_N$ for $I$. Let

$$(a_{1i1}\leftmerge\cdots\leftmerge a_{k_{i1}i1}) Y_{i1}+\cdots+(a_{1im_i}\leftmerge\cdots\leftmerge a_{k_{im_i}im_i}) Y_{im_i}+b_{1i1}\leftmerge\cdots\leftmerge b_{l_{i1}i1}+\cdots+b_{1im_i}\leftmerge\cdots\leftmerge b_{l_{im_i}im_i}$$

be the conflict composition of exits for the cluster $C_i$, with $i\in\{1,\cdots,N\}$.

For $Z\in C_i$ with $i\in\{1,\cdots,N\}$, we define

$s_Z\triangleq (\hat{a_{1i1}}\leftmerge\cdots\leftmerge \hat{a_{k_{i1}i1}}) \tau_I(\langle Y_{i1}|E\rangle)+\cdots+(\hat{a_{1im_i}}\leftmerge\cdots\leftmerge \hat{a_{k_{im_i}im_i}}) \tau_I(\langle Y_{im_i}|E\rangle)+\hat{b_{1i1}}\leftmerge\cdots\leftmerge \hat{b_{l_{i1}i1}}+\cdots+\hat{b_{1im_i}}\leftmerge\cdots\leftmerge \hat{b_{l_{im_i}im_i}}$

For $Z\in C_i$ and $a_1,\cdots,a_j\in \mathbb{E}\cup\{\tau\}$ with $j\in\mathbb{N}$, we have

$(a_1\leftmerge\cdots\leftmerge a_j)\tau_I(\langle Z|E\rangle)$

$=(a_1\leftmerge\cdots\leftmerge a_j)\tau_I((a_{1i1}\leftmerge\cdots\leftmerge a_{k_{i1}i1}) \langle Y_{i1}|E\rangle+\cdots+(a_{1im_i}\leftmerge\cdots\leftmerge a_{k_{im_i}im_i}) \langle Y_{im_i}|E\rangle+b_{1i1}\leftmerge\cdots\leftmerge b_{l_{i1}i1}+\cdots+b_{1im_i}\leftmerge\cdots\leftmerge b_{l_{im_i}im_i})$

$=(a_1\leftmerge\cdots\leftmerge a_j)s_Z$

Let the linear recursive specification $F$ contain the same recursive variables as $E$, for $Z\in C_i$, $F$ contains the following recursive equation

$Z=(\hat{a_{1i1}}\leftmerge\cdots\leftmerge \hat{a_{k_{i1}i1}}) Y_{i1}+\cdots+(\hat{a_{1im_i}}\leftmerge\cdots\leftmerge \hat{a_{k_{im_i}im_i}})  Y_{im_i}+\hat{b_{1i1}}\leftmerge\cdots\leftmerge \hat{b_{l_{i1}i1}}+\cdots+\hat{b_{1im_i}}\leftmerge\cdots\leftmerge \hat{b_{l_{im_i}im_i}}$

It is easy to see that there is no sequence of one or more $\tau$-transitions from $\langle Z|F\rangle$ to itself, so $F$ is guarded.

For

$s_Z=(\hat{a_{1i1}}\leftmerge\cdots\leftmerge \hat{a_{k_{i1}i1}}) Y_{i1}+\cdots+(\hat{a_{1im_i}}\leftmerge\cdots\leftmerge \hat{a_{k_{im_i}im_i}}) Y_{im_i}+\hat{b_{1i1}}\leftmerge\cdots\leftmerge \hat{b_{l_{i1}i1}}+\cdots+\hat{b_{1im_i}}\leftmerge\cdots\leftmerge \hat{b_{l_{im_i}im_i}}$

is a solution for $F$. So, $(a_1\leftmerge\cdots\leftmerge a_j)\tau_I(\langle Z|E\rangle)=(a_1\leftmerge\cdots\leftmerge a_j)s_Z=(a_1\leftmerge\cdots\leftmerge a_j)\langle Z|F\rangle$.

So,

$\langle Z|F\rangle=(\hat{a_{1i1}}\leftmerge\cdots\leftmerge \hat{a_{k_{i1}i1}}) \langle Y_{i1}|F\rangle+\cdots+(\hat{a_{1im_i}}\leftmerge\cdots\leftmerge \hat{a_{k_{im_i}im_i}}) \langle Y_{im_i}|F\rangle+\hat{b_{1i1}}\leftmerge\cdots\leftmerge \hat{b_{l_{i1}i1}}+\cdots+\hat{b_{1im_i}}\leftmerge\cdots\leftmerge \hat{b_{l_{im_i}im_i}}$

Hence, $\tau_I(\langle X|E\rangle=\langle Z|F\rangle)$, as desired.

(2) For the case of rooted branching pomset bisimulation, it can be proven similarly to (1), we omit it.

(3) For the case of rooted branching hp-bisimulation, it can be proven similarly to (1), we omit it.

(4) For the case of rooted branching hhp-bisimulation, it can be proven similarly to (1), we omit it.
\end{proof}

%% file: section4.tex
\section{Process Algebra vs. Petri Net}\label{pp}

In this chapter, we discuss the relationship between process algebra and Petri net. Firstly, we establish the relationship between Petri nets and processes in section \ref{pnap}, based on the structurization of Petri nets. Then, we establish structural operational semantics of Petri net in section \ref{sosopn}. Finally, we reproduce the guarded truly concurrent process algebra $APTC_G$ \cite{APTC} based on the structural operational semantics of Petri net, in section \ref{paopn}.

\subsection{Petri Nets as Processes}\label{pnap}

Informally, a Petri net can be viewed as a bipartite directed graph with two kinds of nodes, called places (denoted by circles $\bigcirc$) and transitions (denoted by squares $\square$) respectively, and also causalities among them. A place is usually a condition and a transition is usually an action, and causalities among the places and transitions: a place may be an input of a transition, or an output of a transition, or the input and output of the same transition. 

The set of places, transitions and causalities among them define the static structure of a Petri net. The behaviour of a Petri net is defined with respect to a starting marking of the graph, called initial marking, in which a marking is a function from the set of places to $\mathbb{N}$. So, a set of execution sequences of a marked Petri net transform markings into other markings through occurrences of transitions. Usually, a token is removed from each input place of a transition, after the execution of this transition, then a token is added to each output place of this transition.

As mentioned in chapter \ref{pe}, a process graph describes the execution of a (closed) term and the definition of the initial term is following the structural way, but, in the definition of a Petri net, there may exist unstructured causalities. So, before treating a Petri net as a process, we must structurize the Petri net.

Formally, we give the definition of Petri net as follows.

\begin{definition}[Petri net]
A Petri net is a triple $PN=\langle\mathbb{B},\mathbb{E},F\rangle$, where $\mathbb{B}$ is a set of places or conditions, $\mathbb{E}$ is a set of transitions or events (actions), and $F\subseteq(\mathbb{B}\times\mathbb{E})\cup(\mathbb{E}\times\mathbb{B})$, with $\mathbb{B}\cap\mathbb{E}=\emptyset$ and $F\neq\emptyset$. And let $\mathbb{B}_i\subseteq\mathbb{B}$ be the initial places and $\mathbb{E}_o\subseteq\mathbb{E}$ be the set of ending transitions.
\end{definition}

Petri nets can be composed together into a bigger Petri net, and a Petri net can or cannot be decomposed into several smaller Petri nets. To compose Petri nets or decompose a Petri net, the first problem is to define the relations between Petri nets. We can see that there are only one kind of relation called causality among the places (conditions) and transitions (actions) in the definition of a Petri net, and concurrency is implicit. We adopt the four relations of causality, choice, concurrency and recursion as the basic relations to compose Petri nets or decompose a Petri net. Since recursions are expressed by recursive equations, while recursive equations are mixtures of recursive variables and places and transitions by sequence $\cdot$, choice $+$ and parallelism $\parallel$ to form terms, so, causality, choice and concurrency are three fundamental relations to compose or decompose Petri nets.

In concurrency theory, causality can be classified finely into sequence and communication, the situation of Petri net is similar to that of event structure in section \ref{pesap}, but without the unstructured conflictions.

Some Petri nets can be composed into a bigger Petri net in sequence, in choice, in parallel, and in concurrency.

\begin{definition}[Petri net composition in sequence]
Let Petri net $PN_1=\langle \mathbb{B}_1, \mathbb{E}_1, F_1\rangle$, $PN_2=\langle \mathbb{B}_2,\mathbb{E}_2, F_2\rangle$ with $\mathbb{B}_i$, $\mathbb{E}_i$ and $F_i$ for $i\in\{1,2\}$ being the corresponding set of places, set of transitions and set of causality relations respectively of Petri net $PN_i$ for $i\in\{1,2\}$ (with a little abuse of symbols), we write $PN_1\cdot PN_2=\langle \mathbb{B}_{PN1\cdot PN2},\mathbb{E}_{PN_1\cdot PN_2}, F_{PN_1\cdot PN_2}\rangle$ for the sequential composition of $PN_1$ and $PN_2$, where
$$\mathbb{B}_{PN_1\cdot PN_2}=\mathbb{B}_1\cup\mathbb{B}_2\quad\mathbb{E}_{PN_1\cdot PN_2}=\mathbb{E}_1\cup\mathbb{E}_2\quad F_{PN_1\cdot PN_2}=F_1\cup F_2\cup(\mathbb{E}_{o_1}\times\mathbb{B}_{i_2}) $$

with $\mathbb{E}_{o_1}\times\mathbb{B}_{i_2}$ is the Cartesian product of $\mathbb{E}_{o_1}$ and $\mathbb{B}_{i_2}$.
\end{definition}

\begin{definition}[Petri net composition in choice]
Let Petri nets $PN_1=\langle \mathbb{B}_1, \mathbb{E}_1, F_1\rangle$, $PN_2=\langle \mathbb{B}_2, \mathbb{E}_2, F_2\rangle$ with $\mathbb{B}_i$, $\mathbb{E}_i$ and $F_i$ for $i\in\{1,2\}$ being the corresponding set of places, set of transitions and set of causality relations respectively of Petri net $PN_i$ for $i\in\{1,2\}$ (with a little abuse of symbols), we write $PN_1+PN_2=\langle \mathbb{B}_{PN_1+PN_2},\mathbb{E}_{PN_1+PN_2}, F_{PN_1+PN_2}\rangle$ for the alternative composition of $PN_1$ and $PN_2$, where
$$\mathbb{B}_{PN_1+PN_2}=\mathbb{B}_1\cup\mathbb{B}_2\quad\mathbb{E}_{PN_1+PN_2}=\mathbb{E}_1\cup\mathbb{E}_2\quad F_{PN_1+PN_2}=F_1\cup F_2$$
\end{definition}

\begin{definition}[Petri net composition in parallel]
Let Petri nets $PN_1=\langle \mathbb{B}_1, \mathbb{E}_1, F_1\rangle$, $PN_2=\langle \mathbb{B}_2,\mathbb{E}_2, F_2\rangle$ with $\mathbb{B}_i$, $\mathbb{E}_i$ and $F_i$ for $i\in\{1,2\}$ being the corresponding set of places, set of transitions and set of causality relations respectively of Petri net $PN_i$ for $i\in\{1,2\}$ (with a little abuse of symbols), we write $PN_1\parallel PN_2=\langle \mathbb{B}_{PN_1\parallel PN_2}, \mathbb{E}_{PN_1\parallel PN_2}, F_{PN_1\parallel PN_2}\rangle$ for the parallel composition of $PN_1$ and $PN_2$, where
$$\mathbb{B}_{PN_1\parallel PN_2}=\mathbb{B}_1\cup\mathbb{B}_2\quad\mathbb{E}_{PN_1\parallel PN_2}=\mathbb{E}_1\cup\mathbb{E}_2\quad F_{PN_1\parallel PN_2}=F_1\cup F_2$$
\end{definition}

\begin{definition}[Petri net composition in concurrency]
Let Petri nets $PN_1=\langle \mathbb{B}_1, \mathbb{E}_1, F_1\rangle$, $PN_2=\langle \mathbb{B}_2, \mathbb{E}_2, F_2\rangle$ with $\mathbb{B}_i$, $\mathbb{E}_i$ and $F_i$ for $i\in\{1,2\}$ being the corresponding set of places, set of transitions and set of causality relations respectively of Petri net $PN_i$ for $i\in\{1,2\}$ (with a little abuse of symbols), we write $PN_1\between PN_2=\langle \mathbb{B}_{PN_1\between PN_2}, \mathbb{E}_{PN_1\between PN_2}, F_{PN_1\between PN_2}\rangle$ for the concurrent composition of $PN_1$ and $PN_2$, where
$$\mathbb{B}_{PN_1\between PN_2}=\mathbb{B}_1\cup\mathbb{B}_2\quad \mathbb{E}_{PN_1\between PN_2}=\mathbb{E}_1\cup\mathbb{E}_2\quad F_{PN_1\between PN_2}=F_1\cup F_2\cup F_{1,2}$$

where $F_{1,2}$ are the set of the newly added causality relations among places and transitions in $\mathbb{B}_1$ and $\mathbb{E}_2$, and $\mathbb{E}_1$ and $\mathbb{B}_2$, and $\mathbb{B}_2$ and $\mathbb{E}_1$, and $\mathbb{E}_2$ and $\mathbb{B}_1$ which are unstructured.
\end{definition}

Note that concurrent composition of Petri net is the common sense composition pattern, and other compositions are all special cases of concurrent composition, such that sequential composition is a concurrent composition with newly added causalities from the ending actions (actions without outgoing causalities) of the first Petri net to the beginning places (places without incoming causalities) of the second Petri net, alternative composition is a concurrent composition with newly added conflictions between the beginning actions of the two Petri nets, parallel composition is a concurrent composition without newly added causalities and conflictions.

Then we discuss the decomposition of a Petri net. We say that a Petri net is structured, we mean that a Petri net can be decomposed into several sub-Petri nets (the sub-Petri nets can composed into the original Petri net by the above composition patterns) without unstructured causalities and conflictions among them. Structured Petri net can capture the above meanings inductively.

\begin{definition}[Structured Petri net]
A Structured Petri net $\mathcal{SPN}$ which is a Petri net $PN=\langle \mathbb{B}, \mathbb{E}, F\rangle$, is inductively defined as follows:

\begin{enumerate}
  \item $\mathbb{B}\cup\mathbb{E}\subset \mathcal{SPN}$;
  \item If $PN_1$ is an $\mathcal{SPN}$ and $PN_2$ is an $\mathcal{SPN}$, then $PN_1\cdot PN_2$ is an $\mathcal{SPN}$;
  \item If $PN_1$ is an $\mathcal{SPN}$ and $PN_2$ is an $\mathcal{SPN}$, then $PN_1+PN_2$ is an $\mathcal{SPN}$;
  \item If $PN_1$ is an $\mathcal{SPN}$ and $PN_2$ is an $\mathcal{SPN}$, then $PN_1\parallel PN_2$ is an $\mathcal{SPN}$.
\end{enumerate}
\end{definition}

Actually, a Petri net defines an unstructured graph with a truly concurrent flavor and can not be structured usually.

\begin{figure}[h]
  \centering
  \includegraphics{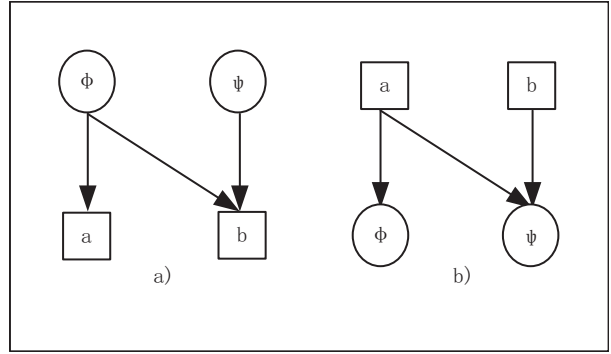}
  \caption{N-shape of Petri net}
  \label{nshapepn}
\end{figure}

\begin{definition}[N-shape]
A Petri net is an N-shape, if it has at least four places and transitions in $\mathbb{B}\cup\mathbb{E}$ with labels such as $\phi$, $\psi$, $a$ and $b$, and causality relations $\phi\leq a$, $\psi\leq b$ and $\phi\leq b$, or causality relations $a\leq \phi$, $b\leq \psi$ and $a\leq \psi$ as illustrated in Fig. \ref{nshapepn}.
\end{definition}

\begin{proposition}[Structurization of N-shape]
An N-shape can not be structurized.
\end{proposition}

\begin{proof}
For Figure \ref{nshapepn}-a) $\phi$ and $\psi$ are in parallel, $a$ after $\phi$, so $\phi$ and $a$ are in the same parallel branch; $b$ after $\psi$, so $\psi$ and $b$ are in the same parallel branch; so $\phi$ and $b$ are in different parallel branches. But, $b$ after $\phi$ means that $\phi$ and $b$ are in the same parallel branch. This causes contradictions.

For the case of Figure \ref{nshapepn}-b, it can be proven similarly that it can not be structurized too.
\end{proof}

Through the above analyses on the composition of Petri nets, it is also reasonable to assume that a Petri net is composed by parallel branches and the unstructured causalities always exist between places and transitions in different parallel branches, usually the unstructured causalities are caused by communications. Now, let us discuss the structurization of Petri net. Firstly, we only consider the synchronous communications. In a synchronous communication, two atomic event pair $a,b$ shakes hands denoted $a\mid b$ and merges as a communication action $\gamma(a,b)$ if the communication exists, otherwise, will cause deadlock (the special constant 0). 

$$a\mid b=\begin{cases}\gamma(a,b),& \textrm{if }\gamma(a,b)\textrm{ is defined;} \\ 0, & \textrm{otherwise.}\end{cases}$$

For the case of N-shape as Figure \ref{nshapepn}-a) illustrated, since evaluation of a place is either 0 or 1, and the evaluation process depends on the manipulation of some data, before the evaluation of the place $\phi$, some data should be communicated between the two branches. Here, we do not process the communication explicitly, and just make a copy of $\phi$ in parallel with the place $\psi$, i.e., $\phi\parallel\psi$ as the new input of the transition $b$ (the algebraic laws in section \ref{paopn} ensure the equivalence), as illustrated in Figure \ref{snshapepn}-a).

For the case of N-shape as Figure \ref{nshapepn}-b) illustrated, the two actions $a$ and $b$ are the inputs of the place $\psi$, they should merge into a single communication action $\gamma(a,b)$ (the algebraic laws in section \ref{paopn} ensure the equivalence), as illustrated in Figure \ref{snshapepn}-b). 

\begin{figure}[h]
  \centering
  \includegraphics{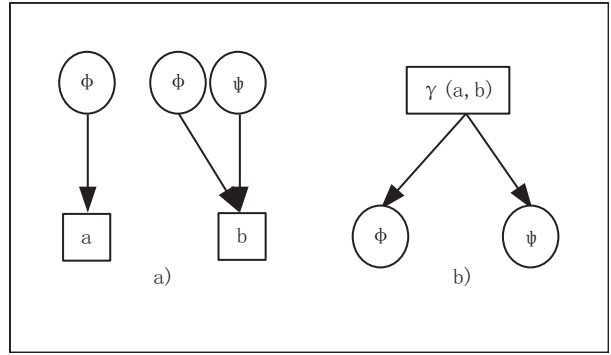}
  \caption{Structurization of N-shape of Petri net}
  \label{snshapepn}
\end{figure}

Thus, with the assumption of parallel branches, and unstructured causalities among them, by use of the elimination methods in Fig. \ref{snshapepn}, a Petri net $PN$ can be structurizated to a structured one $\mathcal{SPN}$. Actually, the unstructured Petri nets and their corresponding structured ones are equivalent modulo truly concurrent bisimulation equivalences, such as pomset bisimulation $\sim_p$, step bisimulation $\sim_s$, hp-bisimulation $\sim_{hp}$ and hhp-bisimulation $\sim_{hhp}$ in section \ref{sosopn}. A Petri net can be expressed by a term of atomic actions (including 0 and 1), binary operators $\cdot$, $+$, $\between$, $\parallel$, $\mid$. This means that a Petri net can be treated as a process, so, in the following, we will not distinguish Petri nets and processes, transitions and events (actions), guards and places.

\subsection{Structural Operational Semantics of Petri Net}\label{sosopn}

In this section, we use some concepts of structural operational semantics, and we do not refer to the original reference, please refer to \cite{SOS} for details.

\subsubsection{Truly Concurrent Bisimulations as Congruences}

\begin{definition}[Configuration]
Let $PN$ be a Petri net. A (finite) configuration in $PN$ is a (finite) consistent subset of events $C\subseteq PN$, closed with respect to causality (i.e. $\lceil C\rceil=C$), and a data state $s\in S$ with $S$ the set of all data states, denoted $\langle C, s\rangle$. The set of finite configurations of $PN$ is denoted by $\langle\mathcal{C}(PN), S\rangle$. We let $\hat{C}=C\backslash\{\tau\}\cup\{\epsilon\}$, where $\tau$ is the silent step and $\epsilon$ is the empty place or transition.
\end{definition}

A consistent subset of $X\subseteq \mathbb{E}$ of events can be seen as a pomset. Given $X, Y\subseteq \mathbb{E}$, $\hat{X}\sim \hat{Y}$ if $\hat{X}$ and $\hat{Y}$ are isomorphic as pomsets. In the following of the paper, we say $C_1\sim C_2$, we mean $\hat{C_1}\sim\hat{C_2}$.

\begin{definition}[Pomset transitions and step]
Let $PN$ be a Petri net and let $C\in\mathcal{C}(PN)$, and $\emptyset\neq X\subseteq \mathbb{E}$, if $C\cap X=\emptyset$ and $C'=C\cup X\in\mathcal{C}(PN)$, then $\langle C,s\rangle\xrightarrow{X} \langle C',s'\rangle$ is called a pomset transition from $\langle C,s\rangle$ to $\langle C',s'\rangle$. When the events in $X$ are pairwise concurrent, we say that $\langle C,s\rangle\xrightarrow{X}\langle C',s'\rangle$ is a step. It is obvious that $\rightarrow^*\xrightarrow{X}\rightarrow^*=\xrightarrow{X}$ and $\rightarrow^*\xrightarrow{e}\rightarrow^*=\xrightarrow{e}$ for any $e\in\mathbb{E}$ and $X\subseteq\mathbb{E}$.
\end{definition}

\begin{definition}[Weak pomset transitions and weak step]
Let $PN$ be a Petri net and let $C\in\mathcal{C}(PN)$, and $\emptyset\neq X\subseteq \hat{\mathbb{E}}$, if $C\cap X=\emptyset$ and $\hat{C'}=\hat{C}\cup X\in\mathcal{C}(PN)$, then $\langle C,s\rangle\xRightarrow{X} \langle C',s'\rangle$ is called a weak pomset transition from $\langle C,s\rangle$ to $\langle C',s'\rangle$, where we define $\xRightarrow{e}\triangleq\xrightarrow{\tau^*}\xrightarrow{e}\xrightarrow{\tau^*}$. And $\xRightarrow{X}\triangleq\xrightarrow{\tau^*}\xrightarrow{e}\xrightarrow{\tau^*}$, for every $e\in X$. When the events in $X$ are pairwise concurrent, we say that $\langle C,s\rangle\xRightarrow{X}\langle C',s'\rangle$ is a weak step.
\end{definition}

We will also suppose that all the Petri nets in this paper are image finite, that is, for any Petri net $PN$ and $C\in \mathcal{C}(PN)$, $\{e\in \mathbb{E}|\langle C,s\rangle\xrightarrow{e} \langle C',s'\rangle\}$ and $\{e\in\hat{\mathbb{E}}|\langle C,s\rangle\xRightarrow{e} \langle C',s'\rangle\}$ is finite.

\begin{definition}[Pomset, step bisimulation]\label{PSBG}
Let $PN_1$, $PN_2$ be Petri nets. A pomset bisimulation is a relation $R\subseteq\langle\mathcal{C}(PN_1),S\rangle\times\langle\mathcal{C}(PN_2),S\rangle$, such that if $(\langle C_1,s\rangle,\langle C_2,s\rangle)\in R$, and $\langle C_1,s\rangle\xrightarrow{X_1}\langle C_1',s'\rangle$ then $\langle C_2,s\rangle\xrightarrow{X_2}\langle C_2',s'\rangle$, with $X_1\subseteq \mathbb{E}_1$, $X_2\subseteq \mathbb{E}_2$, $X_1\sim X_2$ and $(\langle C_1',s'\rangle,\langle C_2',s'\rangle)\in R$ for all $s,s'\in S$, and vice-versa. We say that $PN_1$, $PN_2$ are pomset bisimilar, written $PN_1\sim_pPN_2$, if there exists a pomset bisimulation $R$, such that $(\langle\emptyset,\emptyset\rangle,\langle\emptyset,\emptyset\rangle)\in R$. By replacing pomset transitions with steps, we can get the definition of step bisimulation. When Petri nets $PN_1$ and $PN_2$ are step bisimilar, we write $PN_1\sim_sPN_2$.
\end{definition}

\begin{definition}[Weak pomset, step bisimulation]\label{WPSBG}
Let $PN_1$, $PN_2$ be Petri nets. A weak pomset bisimulation is a relation $R\subseteq\langle\mathcal{C}(PN_1),S\rangle\times\langle\mathcal{C}(PN_2),S\rangle$, such that if $(\langle C_1,s\rangle,\langle C_2,s\rangle)\in R$, and $\langle C_1,s\rangle\xRightarrow{X_1}\langle C_1',s'\rangle$ then $\langle C_2,s\rangle\xRightarrow{X_2}\langle C_2',s'\rangle$, with $X_1\subseteq \hat{\mathbb{E}_1}$, $X_2\subseteq \hat{\mathbb{E}_2}$, $X_1\sim X_2$ and $(\langle C_1',s'\rangle,\langle C_2',s'\rangle)\in R$ for all $s,s'\in S$, and vice-versa. We say that $PN_1$, $PN_2$ are weak pomset bisimilar, written $PN_1\approx_pPN_2$, if there exists a weak pomset bisimulation $R$, such that $(\langle\emptyset,\emptyset\rangle,\langle\emptyset,\emptyset\rangle)\in R$. By replacing weak pomset transitions with weak steps, we can get the definition of weak step bisimulation. When Petri nets $PN_1$ and $PN_2$ are weak step bisimilar, we write $PN_1\approx_sPN_2$.
\end{definition}

\begin{definition}[Posetal product]
Given two Petri nets $PN_1$, $PN_2$, the posetal product of their configurations, denoted $\langle\mathcal{C}(PN_1),S\rangle\overline{\times}\langle\mathcal{C}(PN_2),S\rangle$, is defined as

$$\{(\langle C_1,s\rangle,f,\langle C_2,s\rangle)|C_1\in\mathcal{C}(PN_1),C_2\in\mathcal{C}(PN_2),f:C_1\rightarrow C_2 \textrm{ isomorphism}\}.$$

A subset $R\subseteq\langle\mathcal{C}(PN_1),S\rangle\overline{\times}\langle\mathcal{C}(PN_2),S\rangle$ is called a posetal relation. We say that $R$ is downward closed when for any $(\langle C_1,s\rangle,f,\langle C_2,s\rangle),(\langle C_1',s'\rangle,f',\langle C_2',s'\rangle)\in \langle\mathcal{C}(PN_1),S\rangle\overline{\times}\langle\mathcal{C}(PN_2),S\rangle$, if $(\langle C_1,s\rangle,f,\langle C_2,s\rangle)\subseteq (\langle C_1',s'\rangle,f',\langle C_2',s'\rangle)$ pointwise and $(\langle C_1',s'\rangle,f',\langle C_2',s'\rangle)\in R$, then $(\langle C_1,s\rangle,f,\langle C_2,s\rangle)\in R$.

For $f:X_1\rightarrow X_2$, we define $f[x_1\mapsto x_2]:X_1\cup\{x_1\}\rightarrow X_2\cup\{x_2\}$, $z\in X_1\cup\{x_1\}$,(1)$f[x_1\mapsto x_2](z)=
x_2$,if $z=x_1$;(2)$f[x_1\mapsto x_2](z)=f(z)$, otherwise. Where $X_1\subseteq \mathbb{E}_1$, $X_2\subseteq \mathbb{E}_2$, $x_1\in \mathbb{E}_1$, $x_2\in \mathbb{E}_2$.
\end{definition}

\begin{definition}[Weakly posetal product]
Given two Petri nets $PN_1$, $PN_2$, the weakly posetal product of their configurations, denoted $\langle\mathcal{C}(PN_1),S\rangle\overline{\times}\langle\mathcal{C}(PN_2),S\rangle$, is defined as

$$\{(\langle C_1,s\rangle,f,\langle C_2,s\rangle)|C_1\in\mathcal{C}(PN_1),C_2\in\mathcal{C}(PN_2),f:\hat{C_1}\rightarrow \hat{C_2} \textrm{ isomorphism}\}.$$

A subset $R\subseteq\langle\mathcal{C}(PN_1),S\rangle\overline{\times}\langle\mathcal{C}(PN_2),S\rangle$ is called a weakly posetal relation. We say that $R$ is downward closed when for any $(\langle C_1,s\rangle,f,\langle C_2,s\rangle),(\langle C_1',s'\rangle,f,\langle C_2',s'\rangle)\in \langle\mathcal{C}(PN_1),S\rangle\overline{\times}\langle\mathcal{C}(PN_2),S\rangle$, if $(\langle C_1,s\rangle,f,\langle C_2,s\rangle)\subseteq (\langle C_1',s'\rangle,f',\langle C_2',s'\rangle)$ pointwise and $(\langle C_1',s'\rangle,f',\langle C_2',s'\rangle)\in R$, then $(\langle C_1,s\rangle,f,\langle C_2,s\rangle)\in R$.

For $f:X_1\rightarrow X_2$, we define $f[x_1\mapsto x_2]:X_1\cup\{x_1\}\rightarrow X_2\cup\{x_2\}$, $z\in X_1\cup\{x_1\}$,(1)$f[x_1\mapsto x_2](z)=
x_2$,if $z=x_1$;(2)$f[x_1\mapsto x_2](z)=f(z)$, otherwise. Where $X_1\subseteq \hat{\mathbb{E}_1}$, $X_2\subseteq \hat{\mathbb{E}_2}$, $x_1\in \hat{\mathbb{E}}_1$, $x_2\in \hat{\mathbb{E}}_2$. Also, we define $f(\tau^*)=f(\tau^*)$.
\end{definition}

\begin{definition}[(Hereditary) history-preserving bisimulation]\label{HHPBG}
A history-preserving (hp-) bisimulation is a posetal relation $R\subseteq\langle\mathcal{C}(PN_1),S\rangle\overline{\times}\langle\mathcal{C}(PN_2),S\rangle$ such that if $(\langle C_1,s\rangle,f,\langle C_2,s\rangle)\in R$, and $\langle C_1,s\rangle\xrightarrow{e_1} \langle C_1',s'\rangle$, then $\langle C_2,s\rangle\xrightarrow{e_2} \langle C_2',s'\rangle$, with $(\langle C_1',s'\rangle,f[e_1\mapsto e_2],\langle C_2',s'\rangle)\in R$ for all $s,s'\in S$, and vice-versa. $PN_1,PN_2$ are history-preserving (hp-)bisimilar and are written $PN_1\sim_{hp}PN_2$ if there exists a hp-bisimulation $R$ such that $(\langle\emptyset,\emptyset\rangle,\emptyset,\langle\emptyset,\emptyset\rangle)\in R$.

A hereditary history-preserving (hhp-)bisimulation is a downward closed hp-bisimulation. $PN_1,PN_2$ are hereditary history-preserving (hhp-)bisimilar and are written $PN_1\sim_{hhp}PN_2$.
\end{definition}

\begin{definition}[Weak (hereditary) history-preserving bisimulation]\label{WHHPBG}
A weak history-preserving (hp-) bisimulation is a weakly posetal relation $R\subseteq\langle\mathcal{C}(PN_1),S\rangle\overline{\times}\langle\mathcal{C}(PN_2),S\rangle$ such that if $(\langle C_1,s\rangle,f,\langle C_2,s\rangle)\in R$, and $\langle C_1,s\rangle\xRightarrow{e_1} \langle C_1',s'\rangle$, then $\langle C_2,s\rangle\xRightarrow{e_2} \langle C_2',s'\rangle$, with $(\langle C_1',s'\rangle,f[e_1\mapsto e_2],\langle C_2',s'\rangle)\in R$ for all $s,s'\in S$, and vice-versa. $PN_1,PN_2$ are weak history-preserving (hp-)bisimilar and are written $PN_1\approx_{hp}PN_2$ if there exists a weak hp-bisimulation $R$ such that $(\langle\emptyset,\emptyset\rangle,\emptyset,\langle\emptyset,\emptyset\rangle)\in R$.

A weakly hereditary history-preserving (hhp-)bisimulation is a downward closed weak hp-bisimulation. $PN_1,PN_2$ are weakly hereditary history-preserving (hhp-)bisimilar and are written $PN_1\approx_{hhp}PN_2$.
\end{definition}

\begin{theorem}[Truly concurrent bisimulations as congruences]\label{tcbacpn}
If a TSS is positive after reduction and in panth format, then the truly concurrent bisimulation equivalences, including pomset bisimulation equivalence $\sim_p$, step bisimulation equivalence $\sim_s$, hp-bisimulation equivalence $\sim_{hp}$ and hhp-bisimulation equivalence $\sim_{hhp}$, that it induces are all congruences.
\end{theorem}

\subsubsection{Branching Truly Concurrent Bisimulations as Congruences}

In this section, we give the concepts of rooted branching truly concurrent bisimulation equivalences, based on these concepts, we can design the axiom system of the silent step $\tau$.

\begin{definition}[Branching pomset, step bisimulation]\label{BPSBG}
Assume a special termination predicate $\downarrow$, and let $\surd$ represent a state with $\surd\downarrow$. Let $PN_1$, $PN_2$ be Petri nets. A branching pomset bisimulation is a relation $R\subseteq\langle\mathcal{C}(PN_1),S\rangle\times\langle\mathcal{C}(PN_2),S\rangle$, such that:
 \begin{enumerate}
   \item if $(\langle C_1,s\rangle,\langle C_2,s\rangle)\in R$, and $\langle C_1,s\rangle\xrightarrow{X}\langle C_1',s'\rangle$ then
   \begin{itemize}
     \item either $X\equiv \tau^*$, and $(\langle C_1',s'\rangle,\langle C_2,s\rangle)\in R$ with $s'\in \tau(s)$;
     \item or there is a sequence of (zero or more) $\tau$-transitions $\langle C_2,s\rangle\xrightarrow{\tau^*} \langle C_2^0,s^0\rangle$, such that $(\langle C_1,s\rangle,\langle C_2^0,s^0\rangle)\in R$ and $\langle C_2^0,s^0\rangle\xRightarrow{X}\langle C_2',s'\rangle$ with $(\langle C_1',s'\rangle,\langle C_2',s'\rangle)\in R$;
   \end{itemize}
   \item if $(\langle C_1,s\rangle,\langle C_2,s\rangle)\in R$, and $\langle C_2,s\rangle\xrightarrow{X}\langle C_2',s'\rangle$ then
   \begin{itemize}
     \item either $X\equiv \tau^*$, and $(\langle C_1,s\rangle,\langle C_2',s'\rangle)\in R$;
     \item or there is a sequence of (zero or more) $\tau$-transitions $\langle C_1,s\rangle\xrightarrow{\tau^*} \langle C_1^0,s^0\rangle$, such that $(\langle C_1^0,s^0\rangle,\langle C_2,s\rangle)\in R$ and $\langle C_1^0,s^0\rangle\xRightarrow{X}\langle C_1',s'\rangle$ with $(\langle C_1',s'\rangle,\langle C_2',s'\rangle)\in R$;
   \end{itemize}
   \item if $(\langle C_1,s\rangle,\langle C_2,s\rangle)\in R$ and $\langle C_1,s\rangle\downarrow$, then there is a sequence of (zero or more) $\tau$-transitions $\langle C_2,s\rangle\xrightarrow{\tau^*}\langle C_2^0,s^0\rangle$ such that $(\langle C_1,s\rangle,\langle C_2^0,s^0\rangle)\in R$ and $\langle C_2^0,s^0\rangle\downarrow$;
   \item if $(\langle C_1,s\rangle,\langle C_2,s\rangle)\in R$ and $\langle C_2,s\rangle\downarrow$, then there is a sequence of (zero or more) $\tau$-transitions $\langle C_1,s\rangle\xrightarrow{\tau^*}\langle C_1^0,s^0\rangle$ such that $(\langle C_1^0,s^0\rangle,\langle C_2,s\rangle)\in R$ and $\langle C_1^0,s^0\rangle\downarrow$.
 \end{enumerate}

We say that $PN_1$, $PN_2$ are branching pomset bisimilar, written $PN_1\approx_{bp}PN_2$, if there exists a branching pomset bisimulation $R$, such that $(\langle\emptyset,\emptyset\rangle,\langle\emptyset,\emptyset\rangle)\in R$.

By replacing pomset transitions with steps, we can get the definition of branching step bisimulation. When Petri nets $PN_1$ and $PN_2$ are branching step bisimilar, we write $PN_1\approx_{bs}PN_2$.
\end{definition}

\begin{definition}[Rooted branching pomset, step bisimulation]\label{RBPSBG}
Assume a special termination predicate $\downarrow$, and let $\surd$ represent a state with $\surd\downarrow$. Let $PN_1$, $PN_2$ be Petri nets. A rooted branching pomset bisimulation is a relation $R\subseteq\langle\mathcal{C}(PN_1),S\rangle\times\langle\mathcal{C}(PN_2),S\rangle$, such that:
 \begin{enumerate}
   \item if $(\langle C_1,s\rangle,\langle C_2,s\rangle)\in R$, and $\langle C_1,s\rangle\xrightarrow{X}\langle C_1',s'\rangle$ then $\langle C_2,s\rangle\xrightarrow{X}\langle C_2',s'\rangle$ with $\langle C_1',s'\rangle\approx_{bp}\langle C_2',s'\rangle$;
   \item if $(\langle C_1,s\rangle,\langle C_2,s\rangle)\in R$, and $\langle C_2,s\rangle\xrightarrow{X}\langle C_2',s'\rangle$ then $\langle C_1,s\rangle\xrightarrow{X}\langle C_1',s'\rangle$ with $\langle C_1',s'\rangle\approx_{bp}\langle C_2',s'\rangle$;
   \item if $(\langle C_1,s\rangle,\langle C_2,s\rangle)\in R$ and $\langle C_1,s\rangle\downarrow$, then $\langle C_2,s\rangle\downarrow$;
   \item if $(\langle C_1,s\rangle,\langle C_2,s\rangle)\in R$ and $\langle C_2,s\rangle\downarrow$, then $\langle C_1,s\rangle\downarrow$.
 \end{enumerate}

We say that $PN_1$, $PN_2$ are rooted branching pomset bisimilar, written $PN_1\approx_{rbp}PN_2$, if there exists a rooted branching pomset bisimulation $R$, such that $(\langle\emptyset,\emptyset\rangle,\langle\emptyset,\emptyset\rangle)\in R$.

By replacing pomset transitions with steps, we can get the definition of rooted branching step bisimulation. When Petri nets $PN_1$ and $PN_2$ are rooted branching step bisimilar, we write $PN_1\approx_{rbs}PN_2$.
\end{definition}

\begin{definition}[Branching (hereditary) history-preserving bisimulation]\label{BHHPBG}
Assume a special termination predicate $\downarrow$, and let $\surd$ represent a state with $\surd\downarrow$. A branching history-preserving (hp-) bisimulation is a weakly posetal relation $R\subseteq\langle\mathcal{C}(PN_1),S\rangle\overline{\times}\langle\mathcal{C}(PN_2),S\rangle$ such that:

 \begin{enumerate}
   \item if $(\langle C_1,s\rangle,f,\langle C_2,s\rangle)\in R$, and $\langle C_1,s\rangle\xrightarrow{e_1}\langle C_1',s'\rangle$ then
   \begin{itemize}
     \item either $e_1\equiv \tau$, and $(\langle C_1',s'\rangle,f[e_1\mapsto \tau^{e_1}],\langle C_2,s\rangle)\in R$;
     \item or there is a sequence of (zero or more) $\tau$-transitions $\langle C_2,s\rangle\xrightarrow{\tau^*} \langle C_2^0,s^0\rangle$, such that $(\langle C_1,s\rangle,f,\langle C_2^0,s^0\rangle)\in R$ and $\langle C_2^0,s^0\rangle\xrightarrow{e_2}\langle C_2',s'\rangle$ with $(\langle C_1',s'\rangle,f[e_1\mapsto e_2],\langle C_2',s'\rangle)\in R$;
   \end{itemize}
   \item if $(\langle C_1,s\rangle,f,\langle C_2,s\rangle)\in R$, and $\langle C_2,s\rangle\xrightarrow{e_2}\langle C_2',s'\rangle$ then
   \begin{itemize}
     \item either $e_2\equiv \tau$, and $(\langle C_1,s\rangle,f[e_2\mapsto \tau^{e_2}],\langle C_2',s'\rangle)\in R$;
     \item or there is a sequence of (zero or more) $\tau$-transitions $\langle C_1,s\rangle\xrightarrow{\tau^*} \langle C_1^0,s^0\rangle$, such that $(\langle C_1^0,s^0\rangle,f,\langle C_2,s\rangle)\in R$ and $\langle C_1^0,s^0\rangle\xrightarrow{e_1}\langle C_1',s'\rangle$ with $(\langle C_1',s'\rangle,f[e_2\mapsto e_1],\langle C_2',s'\rangle)\in R$;
   \end{itemize}
   \item if $(\langle C_1,s\rangle,f,\langle C_2,s\rangle)\in R$ and $\langle C_1,s\rangle\downarrow$, then there is a sequence of (zero or more) $\tau$-transitions $\langle C_2,s\rangle\xrightarrow{\tau^*}\langle C_2^0,s^0\rangle$ such that $(\langle C_1,s\rangle,f,\langle C_2^0,s^0\rangle)\in R$ and $\langle C_2^0,s^0\rangle\downarrow$;
   \item if $(\langle C_1,s\rangle,f,\langle C_2,s\rangle)\in R$ and $\langle C_2,s\rangle\downarrow$, then there is a sequence of (zero or more) $\tau$-transitions $\langle C_1,s\rangle\xrightarrow{\tau^*}\langle C_1^0,s^0\rangle$ such that $(\langle C_1^0,s^0\rangle,f,\langle C_2,s\rangle)\in R$ and $\langle C_1^0,s^0\rangle\downarrow$.
 \end{enumerate}

$PN_1,PN_2$ are branching history-preserving (hp-)bisimilar and are written $PN_1\approx_{bhp}PN_2$ if there exists a branching hp-bisimulation $R$ such that $(\langle\emptyset,\emptyset\rangle,\emptyset,\langle\emptyset,\emptyset\rangle)\in R$.

A branching hereditary history-preserving (hhp-)bisimulation is a downward closed branching hp-bisimulation. $PN_1,PN_2$ are branching hereditary history-preserving (hhp-)bisimilar and are written $PN_1\approx_{bhhp}PN_2$.
\end{definition}

\begin{definition}[Rooted branching (hereditary) history-preserving bisimulation]\label{RBHHPBG}
Assume a special termination predicate $\downarrow$, and let $\surd$ represent a state with $\surd\downarrow$. A rooted branching history-preserving (hp-) bisimulation is a weakly posetal relation $R\subseteq\langle\mathcal{C}(PN_1),S\rangle\overline{\times}\langle\mathcal{C}(PN_2),S\rangle$ such that:

 \begin{enumerate}
   \item if $(\langle C_1,s\rangle,f,\langle C_2,s\rangle)\in R$, and $\langle C_1,s\rangle\xrightarrow{e_1}\langle C_1',s'\rangle$, then $\langle C_2,s\rangle\xrightarrow{e_2}\langle C_2',s'\rangle$ with $\langle C_1',s'\rangle\approx_{bhp}\langle C_2',s'\rangle$;
   \item if $(\langle C_1,s\rangle,f,\langle C_2,s\rangle)\in R$, and $\langle C_2,s\rangle\xrightarrow{e_2}\langle C_2',s'\rangle$, then $\langle C_1,s\rangle\xrightarrow{e_1}\langle C_1',s'\rangle$ with $\langle C_1',s'\rangle\approx_{bhp}\langle C_2',s'\rangle$;
   \item if $(\langle C_1,s\rangle,f,\langle C_2,s\rangle)\in R$ and $\langle C_1,s\rangle\downarrow$, then $\langle C_2,s\rangle\downarrow$;
   \item if $(\langle C_1,s\rangle,f,\langle C_2,s\rangle)\in R$ and $\langle C_2,s\rangle\downarrow$, then $\langle C_1,s\rangle\downarrow$.
 \end{enumerate}

$PN_1,PN_2$ are rooted branching history-preserving (hp-)bisimilar and are written $PN_1\approx_{rbhp}PN_2$ if there exists a rooted branching hp-bisimulation $R$ such that $(\langle\emptyset,\emptyset\rangle,\emptyset,\langle\emptyset,\emptyset\rangle)\in R$.

A rooted branching hereditary history-preserving (hhp-)bisimulation is a downward closed rooted branching hp-bisimulation. $PN_1,PN_2$ are rooted branching hereditary history-preserving (hhp-)bisimilar and are written $PN_1\approx_{rbhhp}PN_2$.
\end{definition}

\begin{theorem}[Rooted branching truly concurrent bisimulations as congruences]\label{rbtcbacpn}
If a TSS is positive after reduction and in RBB cool format, then the rooted branching truly concurrent bisimulation equivalences, including rooted branching pomset bisimulation equivalence $\approx_{rbp}$, rooted branching step bisimulation equivalence $\approx_{rbs}$, rooted branching hp-bisimulation equivalence $\approx_{rbhp}$ and rooted branching hhp-bisimulation equivalence $\approx_{rbhhp}$, that it induces are all congruences.
\end{theorem}

\subsection{Process Algebra of Petri Net}\label{paopn} 

In this section, we use $APTC_G$ \cite{APTC} as the process algebra of Petri net, and use guards to model the places in Petri net.

\subsubsection{$BATC$ with Guards}

Let $\mathbb{E}$ be the set of atomic events (actions), $\mathbb{B}_{at}$ be the set of atomic guards to model places in Petri net, $0$ be the deadlock constant, and $1$ be the empty event. We extend $\mathbb{B}_{at}$ to the set of basic guards $\mathbb{B}$ (places) with element $\phi,\psi,\cdots$, which is generated by the following formation rules:

$$\phi::=0|1|\neg\phi|\psi\in \mathbb{B}_{at}|\phi+\psi|\phi\cdot\psi$$

In the following, let $e_1, e_2, e_1', e_2'\in \mathbb{E}$, $\phi,\psi\in \mathbb{B}$ and let variables $x,y,z$ range over the set of process terms, $p,q,s$ range over the set of closed terms. The predicate $test(\phi,s)$ represents that $\phi$ holds in the state $s$, and $test(1,s)$ holds and $test(0,s)$ does not hold. $effect(e,s)\in S$ denotes $s'$ in $s\xrightarrow{e}s'$. The predicate weakest precondition $wp(e,\phi)$ denotes that $\forall s,s'\in S, test(\phi,effect(e,s))$ holds.

The set of axioms of $BATC_G$ consists of the laws given in Table \ref{AxiomsForBATCG}.

\begin{center}
    \begin{table}
        \begin{tabular}{@{}ll@{}}
            \hline No. &Axiom\\
            $A1$ & $x+ y = y+ x$\\
            $A2$ & $(x+ y)+ z = x+ (y+ z)$\\
            $A3$ & $x+ x = x$\\
            $A4$ & $(x+ y)\cdot z = x\cdot z + y\cdot z$\\
            $A5$ & $(x\cdot y)\cdot z = x\cdot(y\cdot z)$\\
            $A6$ & $x+0 = x$\\
            $A7$ & $0\cdot x = 0$\\
            $A8$ & $1\cdot x = x$\\
            $A9$ & $x\cdot1 = x$\\
            $G1$ & $\phi\cdot\neg\phi = 0$\\
            $G2$ & $\phi+\neg\phi = 1$\\
            $G3$ & $\phi+0 = 0$\\
            $G4$ & $\phi(x+y)=\phi x+\phi y$\\
            $G5$ & $\phi(x\cdot y)= \phi x\cdot y$\\
            $G6$ & $(\phi+\psi)x = \phi x + \psi x$\\
            $G7$ & $(\phi\cdot \psi)\cdot x = \phi\cdot(\psi\cdot x)$\\
            $G8$ & $\phi=1$ if $\forall s\in S.test(\phi,s)$\\
            $G9$ & $\phi_0\cdot\cdots\cdot\phi_n = 0$ if $\forall s\in S,\exists i\leq n.test(\neg\phi_i,s)$\\
            $G10$ & $wp(e,\phi)e\phi=wp(e,\phi)e$\\
            $G11$ & $\neg wp(e,\phi)e\neg\phi=\neg wp(e,\phi)e$\\
        \end{tabular}
        \caption{Axioms of $BATC_G$}
        \label{AxiomsForBATCG}
    \end{table}
\end{center}

Note that, by eliminating atomic event from the process terms, the axioms in Table \ref{AxiomsForBATCG} will lead to a Boolean Algebra. And $G9$ is a precondition of $e$ and $\phi$, $G10$ is the weakest precondition of $e$ and $\phi$. A data environment with $effect$ function is sufficiently deterministic, and it is obvious that if the weakest precondition is expressible and $G9$, $G10$ are sound, then the related data environment is sufficiently deterministic.

\begin{definition}[Basic terms of $BATC_G$]\label{BTBATCG}
The set of basic terms of $BATC_G$, $\mathcal{B}(BATC_G)$, is inductively defined as follows:
\begin{enumerate}
  \item $\mathbb{E}\subset\mathcal{B}(BATC_G)$;
  \item $\mathbb{B}\subset\mathcal{B}(BATC_G)$;
  \item if $e\in \mathbb{E}, t\in\mathcal{B}(BATC_G)$ then $e\cdot t\in\mathcal{B}(BATC_G)$;
  \item if $\phi\in \mathbb{B}, t\in\mathcal{B}(BATC_G)$ then $\phi\cdot t\in\mathcal{B}(BATC_G)$;
  \item if $t,s\in\mathcal{B}(BATC_G)$ then $t+ s\in\mathcal{B}(BATC_G)$.
\end{enumerate}
\end{definition}

\begin{theorem}[Elimination theorem of $BATC_G$]\label{ETBATCG}
Let $p$ be a closed $BATC_G$ term. Then there is a basic $BATC_G$ term $q$ such that $BATC_G\vdash p=q$.
\end{theorem}

\begin{proof}
(1) Firstly, suppose that the following ordering on the signature of $BATC_G$ is defined: $\cdot > +$ and the symbol $\cdot$ is given the lexicographical status for the first argument, then for each rewrite rule $p\rightarrow q$ in Table \ref{TRSForBATCG} relation $p>_{lpo} q$ can easily be proved. We obtain that the term rewrite system shown in Table \ref{TRSForBATCG} is strongly normalizing, for it has finitely many rewriting rules, and $>$ is a well-founded ordering on the signature of $BATC_G$, and if $s>_{lpo} t$, for each rewriting rule $s\rightarrow t$ is in Table \ref{TRSForBATCG}.

\begin{center}
    \begin{table}
        \begin{tabular}{@{}ll@{}}
            \hline No. &Rewriting Rule\\
            $RA3$ & $x+ x \rightarrow x$\\
            $RA4$ & $(x+ y)\cdot z \rightarrow x\cdot z + y\cdot z$\\
            $RA5$ & $(x\cdot y)\cdot z \rightarrow x\cdot(y\cdot z)$\\
            $RA6$ & $x+0 \rightarrow x$\\
            $RA7$ & $0\cdot x \rightarrow 0$\\
            $RA8$ & $1\cdot x \rightarrow x$\\
            $RA9$ & $x\cdot1 \rightarrow x$\\
            $RG1$ & $\phi\cdot\neg\phi \rightarrow 0$\\
            $RG2$ & $\phi+\neg\phi \rightarrow 1$\\
            $RG3$ & $\phi0 \rightarrow 0$\\
            $RG4$ & $\phi(x+y)\rightarrow\phi x+\phi y$\\
            $RG5$ & $\phi(x\cdot y)\rightarrow \phi x\cdot y$\\
            $RG6$ & $(\phi+\psi)x \rightarrow \phi x + \psi x$\\
            $RG7$ & $(\phi\cdot \psi)\cdot x \rightarrow \phi\cdot(\psi\cdot x)$\\
            $RG8$ & $\phi\rightarrow 1$ if $\forall s\in S.test(\phi,s)$\\
            $RG9$ & $\phi_0\cdot\cdots\cdot\phi_n \rightarrow 0$ if $\forall s\in S,\exists i\leq n.test(\neg\phi_i,s)$\\
            $RG10$ & $wp(e,\phi)e\phi\rightarrow wp(e,\phi)e$\\
            $RG11$ & $\neg wp(e,\phi)e\neg\phi\rightarrow\neg wp(e,\phi)e$\\
        \end{tabular}
        \caption{Term rewrite system of $BATC_G$}
        \label{TRSForBATCG}
    \end{table}
\end{center}

(2) Then we prove that the normal forms of closed $BATC_G$ terms are basic $BATC_G$ terms.

Suppose that $p$ is a normal form of some closed $BATC_G$ term and suppose that $p$ is not a basic term. Let $p'$ denote the smallest sub-term of $p$ which is not a basic term. It implies that each sub-term of $p'$ is a basic term. Then we prove that $p$ is not a term in normal form. It is sufficient to induct on the structure of $p'$:

\begin{itemize}
  \item Case $p'\equiv e, e\in \mathbb{E}$. $p'$ is a basic term, which contradicts the assumption that $p'$ is not a basic term, so this case should not occur.
  \item Case $p'\equiv \phi, \phi\in \mathbb{B}$. $p'$ is a basic term, which contradicts the assumption that $p'$ is not a basic term, so this case should not occur.
  \item Case $p'\equiv p_1\cdot p_2$. By induction on the structure of the basic term $p_1$:
      \begin{itemize}
        \item Subcase $p_1\in \mathbb{E}$. $p'$ would be a basic term, which contradicts the assumption that $p'$ is not a basic term;
        \item Subcase $p_1\in \mathbb{B}$. $p'$ would be a basic term, which contradicts the assumption that $p'$ is not a basic term;
        \item Subcase $p_1\equiv e\cdot p_1'$. $RA5$ or $RA9$ rewriting rule can be applied. So $p$ is not a normal form;
        \item Subcase $p_1\equiv \phi\cdot p_1'$. $RG1$, $RG3$, $RG4$, $RG5$, $RG7$, or $RG8-9$ rewriting rules can be applied. So $p$ is not a normal form;
        \item Subcase $p_1\equiv p_1'+ p_1''$. $RA4$, $RA6$, $RG2$, or $RG6$ rewriting rules can be applied. So $p$ is not a normal form.
      \end{itemize}
  \item Case $p'\equiv p_1+ p_2$. By induction on the structure of the basic terms both $p_1$ and $p_2$, all subcases will lead to that $p'$ would be a basic term, which contradicts the assumption that $p'$ is not a basic term.
\end{itemize}
\end{proof}

We will define a term-deduction system which gives the operational semantics of $BATC_G$. We give the operational transition rules for $1$, atomic guard $\phi\in \mathbb{B}_{at}$, atomic event $e\in\mathbb{E}$, operators $\cdot$ and $+$ as Table \ref{SETRForBATCG} shows. And the predicate $e\xrightarrow{e}\underline{e}$ represents successful termination after execution of the event $e$, where $\underline{e}$ is the past of event $e$.

\begin{center}
    \begin{table}
        $$\frac{}{\langle 1,s\rangle\rightarrow\langle\surd,s\rangle}$$
        $$\frac{}{\langle e,s\rangle\xrightarrow{e}\langle\underline{e},s'\rangle}\textrm{ if }s'\in effect(e,s)$$
        $$\frac{}{\langle\phi,s\rangle\rightarrow\langle\surd,s\rangle}\textrm{ if }test(\phi,s)$$
        $$\frac{\langle x,s\rangle\xrightarrow{e}\langle\underline{e},s'\rangle}{\langle x+ y,s\rangle\xrightarrow{e}\langle\underline{e},s'\rangle} \quad\frac{\langle x,s\rangle\xrightarrow{e}\langle x',s'\rangle}{\langle x+ y,s\rangle\xrightarrow{e}\langle x',s'\rangle}$$
        $$\frac{\langle y,s\rangle\xrightarrow{e}\langle\underline{e},s'\rangle}{\langle x+ y,s\rangle\xrightarrow{e}\langle\underline{e},s'\rangle} \quad\frac{\langle y,s\rangle\xrightarrow{e}\langle y',s'\rangle}{\langle x+ y,s\rangle\xrightarrow{e}\langle y',s'\rangle}$$
        $$\frac{\langle x,s\rangle\xrightarrow{e}\langle\underline{e},s'\rangle}{\langle x\cdot y,s\rangle\xrightarrow{e} \langle \underline{e}\cdot y,s'\rangle} \quad\frac{\langle x,s\rangle\xrightarrow{e}\langle x',s'\rangle}{\langle x\cdot y,s\rangle\xrightarrow{e}\langle x'\cdot y,s'\rangle}$$
        \caption{Single event transition rules of $BATC_G$}
        \label{SETRForBATCG}
    \end{table}
\end{center}

Note that, we replace the single atomic event $e\in\mathbb{E}$ by $X\subseteq\mathbb{E}$, we can obtain the pomset transition rules of $BATC_G$, and omit them.

\begin{theorem}[Congruence of $BATC_G$ with respect to truly concurrent bisimulation equivalences]\label{CBATCG}
(1) Pomset bisimulation equivalence $\sim_{p}$ is a congruence with respect to $BATC_G$.

(2) Step bisimulation equivalence $\sim_{s}$ is a congruence with respect to $BATC_G$.

(3) Hp-bisimulation equivalence $\sim_{hp}$ is a congruence with respect to $BATC_G$.

(4) Hhp-bisimulation equivalence $\sim_{hhp}$ is a congruence with respect to $BATC_G$.
\end{theorem}

\begin{proof}
(1) It is easy to see that pomset bisimulation is an equivalent relation on $BATC_G$ terms, we only need to prove that $\sim_{p}$ is preserved by the operators $\cdot$ and $+$. Since the TSS in Table \ref{SETRForBATCG} is in Panth format, by Theorem \ref{tcbacpn}, pomset bisimulation equivalence $\sim_{p}$ is a congruence with respect to $BATC_G$.

(2) It is easy to see that step bisimulation is an equivalent relation on $BATC_G$ terms, we only need to prove that $\sim_{s}$ is preserved by the operators $\cdot$ and $+$. Since the TSS in Table \ref{SETRForBATCG} is in Panth format, by Theorem \ref{tcbacpn}, step bisimulation equivalence $\sim_{s}$ is a congruence with respect to $BATC_G$.

(3) It is easy to see that hp-bisimulation is an equivalent relation on $BATC_G$ terms, we only need to prove that $\sim_{hp}$ is preserved by the operators $\cdot$ and $+$. Since the TSS in Table \ref{SETRForBATCG} is in Panth format, by Theorem \ref{tcbacpn}, hp-bisimulation equivalence $\sim_{hp}$ is a congruence with respect to $BATC_G$.

(4) It is easy to see that hhp-bisimulation is an equivalent relation on $BATC_G$ terms, we only need to prove that $\sim_{hhp}$ is preserved by the operators $\cdot$ and $+$. Since the TSS in Table \ref{SETRForBATCG} is in Panth format, by Theorem \ref{tcbacpn}, hhp-bisimulation equivalence $\sim_{hhp}$ is a congruence with respect to $BATC_G$.
\end{proof}

\begin{theorem}[Soundness of $BATC_G$ modulo truly concurrent bisimulation equivalences]\label{SBATCG}
(1) Let $x$ and $y$ be $BATC_G$ terms. If $BATC\vdash x=y$, then $x\sim_{p} y$.

(2) Let $x$ and $y$ be $BATC_G$ terms. If $BATC\vdash x=y$, then $x\sim_{s} y$.

(3) Let $x$ and $y$ be $BATC_G$ terms. If $BATC\vdash x=y$, then $x\sim_{hp} y$.

(4) Let $x$ and $y$ be $BATC_G$ terms. If $BATC\vdash x=y$, then $x\sim_{hhp} y$.
\end{theorem}

\begin{proof}
(1) Since pomset bisimulation $\sim_p$ is both an equivalent and a congruent relation, we only need to check if each axiom in Table \ref{AxiomsForBATCG} is sound modulo pomset bisimulation equivalence. We leave the proof as an exercise for the readers.

(2) Since step bisimulation $\sim_s$ is both an equivalent and a congruent relation, we only need to check if each axiom in Table \ref{AxiomsForBATCG} is sound modulo step bisimulation equivalence. We leave the proof as an exercise for the readers.

(3) Since hp-bisimulation $\sim_{hp}$ is both an equivalent and a congruent relation, we only need to check if each axiom in Table \ref{AxiomsForBATCG} is sound modulo hp-bisimulation equivalence. We leave the proof as an exercise for the readers.

(4) Since hhp-bisimulation $\sim_{hhp}$ is both an equivalent and a congruent relation, we only need to check if each axiom in Table \ref{AxiomsForBATCG} is sound modulo hhp-bisimulation equivalence. We leave the proof as an exercise for the readers.
\end{proof}

\begin{theorem}[Completeness of $BATC_G$ modulo truly concurrent bisimulation equivalences]\label{CBATCG}
(1) Let $p$ and $q$ be closed $BATC_G$ terms, if $p\sim_{p} q$ then $p=q$.

(2) Let $p$ and $q$ be closed $BATC_G$ terms, if $p\sim_{s} q$ then $p=q$.

(3) Let $p$ and $q$ be closed $BATC_G$ terms, if $p\sim_{hp} q$ then $p=q$.

(4) Let $p$ and $q$ be closed $BATC_G$ terms, if $p\sim_{hhp} q$ then $p=q$.
\end{theorem}

\begin{proof}
(1) Firstly, by the elimination theorem of $BATC_G$, we know that for each closed $BATC_G$ term $p$, there exists a closed basic $BATC_G$ term $p'$, such that $BATC_G\vdash p=p'$, so, we only need to consider closed basic $BATC_G$ terms.

The basic terms (see Definition \ref{BTBATCG}) modulo associativity and commutativity (AC) of conflict $+$ (defined by axioms $A1$ and $A2$ in Table \ref{AxiomsForBATCG}), and this equivalence is denoted by $=_{AC}$. Then, each equivalence class $s$ modulo AC of $+$ has the following normal form

$$s_1+\cdots+ s_k$$

with each $s_i$ either an atomic event, or an atomic guard, or of the form $t_1\cdot t_2$, and each $s_i$ is called the summand of $s$.

Now, we prove that for normal forms $n$ and $n'$, if $n\sim_{p} n'$ then $n=_{AC}n'$. It is sufficient to induct on the sizes of $n$ and $n'$.

\begin{itemize}
  \item Consider a summand $e$ of $n$. Then $\langle n,s\rangle\xrightarrow{e}\langle \underline{e},s'\rangle$, so $n\sim_p n'$ implies $\langle n',s\rangle\xrightarrow{e}\langle \underline{e},s\rangle$, meaning that $n'$ also contains the summand $e$.
  \item Consider a summand $\phi$ of $n$. Then $\langle n,s\rangle\rightarrow\langle \surd,s\rangle$, if $test(\phi,s)$ holds, so $n\sim_p n'$ implies $\langle n',s\rangle\rightarrow\langle \surd,s\rangle$, if $test(\phi,s)$ holds, meaning that $n'$ also contains the summand $\phi$.
  \item Consider a summand $t_1\cdot t_2$ of $n$. Then $\langle n,s\rangle\xrightarrow{t_1}\langle t_2,s'\rangle$, so $n\sim_p n'$ implies $\langle n',s\rangle\xrightarrow{t_1}\langle t_2',s'\rangle$ with $t_2\sim_p t_2'$, meaning that $n'$ contains a summand $t_1\cdot t_2'$. Since $t_2$ and $t_2'$ are normal forms and have sizes smaller than $n$ and $n'$, by the induction hypotheses $t_2\sim_p t_2'$ implies $t_2=_{AC} t_2'$.
\end{itemize}

So, we get $n=_{AC} n'$.

Finally, let $s$ and $t$ be basic terms, and $s\sim_p t$, there are normal forms $n$ and $n'$, such that $s=n$ and $t=n'$. The soundness theorem of $BATC_G$ modulo pomset bisimulation equivalence (see Theorem \ref{SBATCG}) yields $s\sim_p n$ and $t\sim_p n'$, so $n\sim_p s\sim_p t\sim_p n'$. Since if $n\sim_p n'$ then $n=_{AC}n'$, $s=n=_{AC}n'=t$, as desired.

(2) It can be proven similarly as (1).

(3) It can be proven similarly as (1).

(4) It can be proven similarly as (1).
\end{proof}

\subsubsection{$APTC$ with Guards}{\label{aptcg}}

In this subsection, we will extend $APTC$ with guards, which is abbreviated $APTC_G$. The set of basic guards $\mathbb{B}$ with element $\phi,\psi,\cdots$, which is extended by the following formation rules:

$$\phi::=0|1|\neg\phi|\psi\in \mathbb{B}_{at}|\phi+\psi|\phi\cdot\psi|\phi\parallel\psi$$

The set of axioms of $APTC_G$ including axioms of $BATC_G$ in Table \ref{AxiomsForBATCG} and the axioms are shown in Table \ref{AxiomsForAPTCG}.

\begin{center}
    \begin{table}
        \begin{tabular}{@{}ll@{}}
            \hline No. &Axiom\\
            $P1$ & $x\between y = x\parallel y + x\mid y$\\
            $P2$ & $e_1\parallel (e_2\cdot y) = (e_1\parallel e_2)\cdot y$\\
            $P3$ & $(e_1\cdot x)\parallel e_2 = (e_1\parallel e_2)\cdot x$\\
            $P4$ & $(e_1\cdot x)\parallel (e_2\cdot y) = (e_1\parallel e_2)\cdot (x\between y)$\\
            $P5$ & $(x+ y)\parallel z = (x\parallel z)+ (y\parallel z)$\\
            $P6$ & $x\parallel (y+ z) = (x\parallel y)+ (x\parallel z)$\\
            $P7$ & $0\parallel x = 0$\\
            $P8$ & $x\parallel 0 = 0$\\
            $P9$ & $1\parallel x = x$\\
            $P10$ & $x\parallel 1 = x$\\
            $C1$ & $e_1\mid e_2 = \gamma(e_1,e_2)$\\
            $C2$ & $e_1\mid (e_2\cdot y) = \gamma(e_1,e_2)\cdot y$\\
            $C3$ & $(e_1\cdot x)\mid e_2 = \gamma(e_1,e_2)\cdot x$\\
            $C4$ & $(e_1\cdot x)\mid (e_2\cdot y) = \gamma(e_1,e_2)\cdot (x\between y)$\\
            $C5$ & $(x+ y)\mid z = (x\mid z) + (y\mid z)$\\
            $C6$ & $x\mid (y+ z) = (x\mid y)+ (x\mid z)$\\
            $C7$ & $0\mid x = 0$\\
            $C8$ & $x\mid0 = 0$\\
            $C9$ & $1\mid x = 0$\\
            $C10$ & $x\mid1 = 0$\\
            $D1$ & $e\notin H\quad\partial_H(e) = e$\\
            $D2$ & $e\in H\quad \partial_H(e) = 0$\\
            $D3$ & $\partial_H(0) = 0$\\
            $D4$ & $\partial_H(x+ y) = \partial_H(x)+\partial_H(y)$\\
            $D5$ & $\partial_H(x\cdot y) = \partial_H(x)\cdot\partial_H(y)$\\
            $D6$ & $\partial_H(x\parallel y) = \partial_H(x)\parallel\partial_H(y)$\\
            $G12$ & $\phi(x\parallel y) =\phi x\parallel \phi y$\\
            $G13$ & $\phi(x\mid y) =\phi x\mid \phi y$\\
            $G14$ & $\phi\parallel 0 = 0$\\
            $G15$ & $0\parallel \phi = 0$\\
            $G16$ & $\phi\mid 0 = 0$\\
            $G17$ & $0\mid \phi = 0$\\
            $G18$ & $\phi\parallel 1 = \phi$\\
            $G19$ & $1\parallel \phi = \phi$\\
            $G20$ & $\phi\mid 1 = 0$\\
            $G21$ & $1\mid \phi = 0$\\
            $G22$ & $\phi\parallel\neg\phi = 0$\\
            $G23$ & $\partial_H(\phi) = \phi$\\
            $G24$ & $\phi_0\parallel\cdots\parallel\phi_n = 0$ if $\forall s_0,\cdots,s_n\in S,\exists i\leq n.test(\neg\phi_i,s_0\cup\cdots\cup s_n)$\\
        \end{tabular}
        \caption{Axioms of $APTC_G$}
        \label{AxiomsForAPTCG}
    \end{table}
\end{center}

\begin{definition}[Basic terms of $APTC_G$]\label{BTAPTCG}
The set of basic terms of $APTC_G$, $\mathcal{B}(APTC_G)$, is inductively defined as follows:
\begin{enumerate}
    \item $\mathbb{E}\subset\mathcal{B}(APTC_G)$;
    \item $\mathbb{B}\subset\mathcal{B}(APTC_G)$;
    \item if $e\in \mathbb{E}, t\in\mathcal{B}(APTC_G)$ then $e\cdot t\in\mathcal{B}(APTC_G)$;
    \item if $\phi\in \mathbb{B}, t\in\mathcal{B}(APTC_G)$ then $\phi\cdot t\in\mathcal{B}(APTC_G)$;
    \item if $t,s\in\mathcal{B}(APTC_G)$ then $t+ s\in\mathcal{B}(APTC_G)$.
    \item if $t,s\in\mathcal{B}(APTC_G)$ then $t\parallel s\in\mathcal{B}(APTC_G)$.
\end{enumerate}
\end{definition}

Based on the definition of basic terms for $APTC_G$ (see Definition \ref{BTAPTCG}) and axioms of $APTC_G$, we can prove the elimination theorem of $APTC_G$.

\begin{theorem}[Elimination theorem of $APTC_G$]\label{ETAPTCG}
Let $p$ be a closed $APTC_G$ term. Then there is a basic $APTC_G$ term $q$ such that $APTC_G\vdash p=q$.
\end{theorem}

\begin{proof}
(1) Firstly, suppose that the following ordering on the signature of $APTC_G$ is defined: $\parallel > \cdot > +$ and the symbol $\parallel$ is given the lexicographical status for the first argument, then for each rewrite rule $p\rightarrow q$ in Table \ref{TRSForAPTCG} relation $p>_{lpo} q$ can easily be proved. We obtain that the term rewrite system shown in Table \ref{TRSForAPTCG} is strongly normalizing, for it has finitely many rewriting rules, and $>$ is a well-founded ordering on the signature of $APTC_G$, and if $s>_{lpo} t$, for each rewriting rule $s\rightarrow t$ is in Table \ref{TRSForAPTCG} (see Theorem \ref{SN}).

\begin{center}
    \begin{table}
        \begin{tabular}{@{}ll@{}}
            \hline No. &Rewriting Rule\\
            $RP1$ & $x\between y \rightarrow x\parallel y + x\mid y$\\
            $RP2$ & $e_1\parallel (e_2\cdot y) \rightarrow (e_1\parallel e_2)\cdot y$\\
            $RP3$ & $(e_1\cdot x)\parallel e_2 \rightarrow (e_1\parallel e_2)\cdot x$\\
            $RP4$ & $(e_1\cdot x)\parallel (e_2\cdot y) \rightarrow (e_1\parallel e_2)\cdot (x\between y)$\\
            $RP5$ & $(x+ y)\parallel z \rightarrow (x\parallel z)+ (y\parallel z)$\\
            $RP6$ & $x\parallel (y+ z) \rightarrow (x\parallel y)+ (x\parallel z)$\\
            $RP7$ & $0\parallel x \rightarrow 0$\\
            $RP8$ & $x\parallel 0 \rightarrow 0$\\
            $RP9$ & $1\parallel x \rightarrow x$\\
            $RP10$ & $x\parallel 1 \rightarrow x$\\
            $RC1$ & $e_1\mid e_2 \rightarrow \gamma(e_1,e_2)$\\
            $RC2$ & $e_1\mid (e_2\cdot y) \rightarrow \gamma(e_1,e_2)\cdot y$\\
            $RC3$ & $(e_1\cdot x)\mid e_2 \rightarrow \gamma(e_1,e_2)\cdot x$\\
            $RC4$ & $(e_1\cdot x)\mid (e_2\cdot y) \rightarrow \gamma(e_1,e_2)\cdot (x\between y)$\\
            $RC5$ & $(x+ y)\mid z \rightarrow (x\mid z) + (y\mid z)$\\
            $RC6$ & $x\mid (y+ z) \rightarrow (x\mid y)+ (x\mid z)$\\
            $RC7$ & $0\mid x \rightarrow 0$\\
            $RC8$ & $x\mid0 \rightarrow 0$\\
            $RC9$ & $1\mid x \rightarrow 0$\\
            $RC10$ & $x\mid1 \rightarrow 0$\\
            $RD1$ & $e\notin H\quad\partial_H(e) \rightarrow e$\\
            $RD2$ & $e\in H\quad \partial_H(e) \rightarrow 0$\\
            $RD3$ & $\partial_H(0) \rightarrow 0$\\
            $RD4$ & $\partial_H(x+ y) \rightarrow \partial_H(x)+\partial_H(y)$\\
            $RD5$ & $\partial_H(x\cdot y) \rightarrow \partial_H(x)\cdot\partial_H(y)$\\
            $RD6$ & $\partial_H(x\parallel y) \rightarrow \partial_H(x)\parallel\partial_H(y)$\\
            $RG12$ & $\phi(x\parallel y) \rightarrow\phi x\parallel \phi y$\\
            $RG13$ & $\phi(x\mid y) \rightarrow\phi x\mid \phi y$\\
            $RG14$ & $\phi\parallel 0 \rightarrow 0$\\
            $RG15$ & $0\parallel \phi \rightarrow 0$\\
            $RG16$ & $\phi\mid 0 \rightarrow 0$\\
            $RG17$ & $0\mid \phi \rightarrow 0$\\
            $RG18$ & $\phi\parallel 1 \rightarrow \phi$\\
            $RG19$ & $1\parallel \phi \rightarrow \phi$\\
            $RG20$ & $\phi\mid 1 \rightarrow 0$\\
            $RG21$ & $1\mid \phi \rightarrow 0$\\
            $RG22$ & $\phi\parallel\neg\phi \rightarrow 0$\\
            $RG23$ & $\partial_H(\phi) \rightarrow \phi$\\
            $RG24$ & $\phi_0\parallel\cdots\parallel\phi_n \rightarrow 0$ if $\forall s_0,\cdots,s_n\in S,\exists i\leq n.test(\neg\phi_i,s_0\cup\cdots\cup s_n)$\\
        \end{tabular}
        \caption{Term rewrite system of $APTC_G$}
        \label{TRSForAPTCG}
    \end{table}
\end{center}

(2) Then we prove that the normal forms of closed $APTC_G$ terms are basic $APTC_G$ terms.

Suppose that $p$ is a normal form of some closed $APTC_G$ term and suppose that $p$ is not a basic $APTC_G$ term. Let $p'$ denote the smallest sub-term of $p$ which is not a basic $APTC_G$ term. It implies that each sub-term of $p'$ is a basic $APTC_G$ term. Then we prove that $p$ is not a term in normal form. It is sufficient to induct on the structure of $p'$:

\begin{itemize}
  \item Case $p'\equiv e, e\in \mathbb{E}$. $p'$ is a basic $APTC_G$ term, which contradicts the assumption that $p'$ is not a basic $APTC_G$ term, so this case should not occur.
  \item Case $p'\equiv \phi, \phi\in \mathbb{B}$. $p'$ is a basic term, which contradicts the assumption that $p'$ is not a basic term, so this case should not occur.
  \item Case $p'\equiv p_1\cdot p_2$. By induction on the structure of the basic $APTC_G$ term $p_1$:
      \begin{itemize}
        \item Subcase $p_1\in \mathbb{E}$. $p'$ would be a basic $APTC_G$ term, which contradicts the assumption that $p'$ is not a basic $APTC_G$ term;
        \item Subcase $p_1\in \mathbb{B}$. $p'$ would be a basic term, which contradicts the assumption that $p'$ is not a basic term;
        \item Subcase $p_1\equiv e\cdot p_1'$. $RA5$ or $RA9$ rewriting rules in Table \ref{TRSForBATCG} can be applied. So $p$ is not a normal form;
        \item Subcase $p_1\equiv \phi\cdot p_1'$. $RG1$, $RG3$, $RG4$, $RG5$, $RG7$, or $RG8-9$ rewriting rules can be applied. So $p$ is not a normal form;
        \item Subcase $p_1\equiv p_1'+ p_1''$. $RA4$, $RA6$, $RG2$, or $RG6$ rewriting rules in Table \ref{TRSForBATCG} can be applied. So $p$ is not a normal form;
        \item Subcase $p_1\equiv p_1'\parallel p_1''$. $RP2$-$RP10$ rewrite rules in Table \ref{TRSForAPTCG} can be applied. So $p$ is not a normal form;
        \item Subcase $p_1\equiv p_1'\mid p_1''$. $RC1$-$RC10$ rewrite rules in Table \ref{TRSForAPTCG} can be applied. So $p$ is not a normal form;
        \item Subcase $p_1\equiv \partial_H(p_1')$. $RD1$-$RD6$ rewrite rules in Table \ref{TRSForAPTCG} can be applied. So $p$ is not a normal form.
      \end{itemize}
  \item Case $p'\equiv p_1+ p_2$. By induction on the structure of the basic $APTC_G$ terms both $p_1$ and $p_2$, all subcases will lead to that $p'$ would be a basic $APTC_G$ term, which contradicts the assumption that $p'$ is not a basic $APTC_G$ term.
  \item Case $p'\equiv p_1\parallel p_2$. By induction on the structure of the basic $APTC_G$ terms both $p_1$ and $p_2$, all subcases will lead to that $p'$ would be a basic $APTC_G$ term, which contradicts the assumption that $p'$ is not a basic $APTC_G$ term.
  \item Case $p'\equiv p_1\mid p_2$. By induction on the structure of the basic $APTC_G$ terms both $p_1$ and $p_2$, all subcases will lead to that $p'$ would be a basic $APTC_G$ term, which contradicts the assumption that $p'$ is not a basic $APTC_G$ term.
  \item Case $p'\equiv \partial_H(p_1)$. By induction on the structure of the basic $APTC_G$ terms of $p_1$, all subcases will lead to that $p'$ would be a basic $APTC_G$ term, which contradicts the assumption that $p'$ is not a basic $APTC_G$ term.
\end{itemize}
\end{proof}

We will define a term-deduction system which gives the operational semantics of $APTC_G$..

\begin{center}
    \begin{table}
        $$\frac{}{\langle e_1\parallel\cdots \parallel e_n,s\rangle\xrightarrow{\{e_1,\cdots,e_n\}}\langle \underline{e_1}\parallel\cdots \parallel \underline{e_n},s'\rangle}\textrm{ if }s'\in effect(e_1,s)\cup\cdots\cup effect(e_n,s)$$
        $$\frac{}{\langle\phi_1\parallel\cdots\parallel \phi_n,s\rangle\rightarrow\langle\surd,s\rangle}\textrm{ if }test(\phi_1,s),\cdots,test(\phi_n,s)$$

        $$\frac{\langle x,s\rangle\xrightarrow{e_1}\langle\underline{e_1},s'\rangle\quad \langle y,s\rangle\xrightarrow{e_2}\langle\underline{e_2},s''\rangle}{\langle x\parallel y,s\rangle\xrightarrow{\{e_1,e_2\}}\langle\underline{e_1}\parallel\underline{e_2},s'\cup s''\rangle} \quad\frac{\langle x,s\rangle\xrightarrow{e_1}\langle x',s'\rangle\quad \langle y,s\rangle\xrightarrow{e_2}\langle\underline{e_2},s''\rangle}{\langle x\parallel y,s\rangle\xrightarrow{\{e_1,e_2\}}\langle x'\parallel\underline{e_2},s'\cup s''\rangle}$$

        $$\frac{\langle x,s\rangle\xrightarrow{e_1}\langle\underline{e_1},s'\rangle\quad \langle y,s\rangle\xrightarrow{e_2}\langle y',s''\rangle}{\langle x\parallel y,s\rangle\xrightarrow{\{e_1,e_2\}}\langle \underline{e_1}\parallel y',s'\cup s''\rangle} \quad\frac{\langle x,s\rangle\xrightarrow{e_1}\langle x',s'\rangle\quad \langle y,s\rangle\xrightarrow{e_2}\langle y',s''\rangle}{\langle x\parallel y,s\rangle\xrightarrow{\{e_1,e_2\}}\langle x'\between y',s'\cup s''\rangle}$$

        $$\frac{\langle x,s\rangle\xrightarrow{e_1}\langle\underline{e_1},s'\rangle\quad \langle y,s\rangle\xrightarrow{e_2}\langle\underline{e_2},s''\rangle}{\langle x\mid y,s\rangle\xrightarrow{\gamma(e_1,e_2)}\langle\underline{\gamma(e_1,e_2)},effect(\gamma(e_1,e_2),s)\rangle} \quad\frac{\langle x,s\rangle\xrightarrow{e_1}\langle x',s'\rangle\quad \langle y,s\rangle\xrightarrow{e_2}\langle\underline{e_2},s''\rangle}{\langle x\mid y,s\rangle\xrightarrow{\gamma(e_1,e_2)}\langle x'\between\underline{\gamma(e_1,e_2)},effect(\gamma(e_1,e_2),s)\rangle}$$

        $$\frac{\langle x,s\rangle\xrightarrow{e_1}\langle\underline{e_1},s'\rangle\quad \langle y,s\rangle\xrightarrow{e_2}\langle y',s''\rangle}{\langle x\mid y,s\rangle\xrightarrow{\gamma(e_1,e_2)}\langle y'\between\underline{\gamma(e_1,e_2)},effect(\gamma(e_1,e_2),s)\rangle} \quad\frac{\langle x,s\rangle\xrightarrow{e_1}\langle x',s'\rangle\quad \langle y,s\rangle\xrightarrow{e_2}\langle y',s''\rangle}{\langle x\mid y,s\rangle\xrightarrow{\gamma(e_1,e_2)}\langle x'\between y',effect(\gamma(e_1,e_2),s)\rangle}$$

        $$\frac{\langle x,s\rangle\xrightarrow{e}\langle\underline{e},s'\rangle}{\langle\partial_H(x),s\rangle\xrightarrow{e}\langle\underline{e},s'\rangle}\quad (e\notin H)\quad\frac{\langle x,s\rangle\xrightarrow{e}\langle x',s'\rangle}{\langle\partial_H(x),s\rangle\xrightarrow{e}\langle\partial_H(x'),s'\rangle}\quad(e\notin H)$$
        \caption{Transition rules of $APTC_G$}
        \label{TRForAPTCG}
    \end{table}
\end{center}

\begin{theorem}[Generalization of $APTC_G$ with respect to $BATC_G$]
$APTC_G$ is a generalization of $BATC_G$.
\end{theorem}

\begin{proof}
It follows from the following three facts.

\begin{enumerate}
  \item The transition rules of $BATC_G$ are all source-dependent;
  \item The sources of the transition rules $APTC_G$ contain an occurrence of $\between$, or $\parallel$, or $\mid$;
  \item The transition rules of $APTC_G$ are all source-dependent.
\end{enumerate}

So, $APTC_G$ is a generalization of $BATC_G$, that is, $BATC_G$ is an embedding of $APTC_G$, as desired.
\end{proof}

\begin{theorem}[Congruence of $APTC_G$ with respect to truly concurrent bisimulation equivalences]\label{CAPTCG}
(1) Pomset bisimulation equivalence $\sim_{p}$ is a congruence with respect to $APTC_G$.

(2) Step bisimulation equivalence $\sim_{s}$ is a congruence with respect to $APTC_G$.

(3) Hp-bisimulation equivalence $\sim_{hp}$ is a congruence with respect to $APTC_G$.

(4) Hhp-bisimulation equivalence $\sim_{hhp}$ is a congruence with respect to $APTC_G$.
\end{theorem}

\begin{proof}
(1) It is easy to see that pomset bisimulation is an equivalent relation on $APTC_G$ terms, we only need to prove that $\sim_{p}$ is preserved by the operators $\parallel$, $\mid$, $\partial_H$. Since the TSS in Table \ref{TRForAPTCG} is in Panth format, by Theorem \ref{tcbacpn}, pomset bisimulation equivalence $\sim_{p}$ is a congruence with respect to $APTC_G$.

(2) It is easy to see that step bisimulation is an equivalent relation on $APTC_G$ terms, we only need to prove that $\sim_{s}$ is preserved by the operators $\parallel$, $\mid$, $\partial_H$. Since the TSS in Table \ref{TRForAPTCG} is in Panth format, by Theorem \ref{tcbacpn}, step bisimulation equivalence $\sim_{s}$ is a congruence with respect to $APTC_G$.

(3) It is easy to see that hp-bisimulation is an equivalent relation on $APTC_G$ terms, we only need to prove that $\sim_{hp}$ is preserved by the operators $\parallel$, $\mid$, $\partial_H$. Since the TSS in Table \ref{TRForAPTCG} is in Panth format, by Theorem \ref{tcbacpn}, hp-bisimulation equivalence $\sim_{hp}$ is a congruence with respect to $APTC_G$.

(4) It is easy to see that hhp-bisimulation is an equivalent relation on $APTC_G$ terms, we only need to prove that $\sim_{hhp}$ is preserved by the operators $\parallel$, $\mid$, $\partial_H$. Since the TSS in Table \ref{TRForAPTCG} is in Panth format, by Theorem \ref{tcbacpn}, hhp-bisimulation equivalence $\sim_{hhp}$ is a congruence with respect to $APTC_G$.
\end{proof}

\begin{theorem}[Soundness of $APTC_G$ modulo truly concurrent bisimulation equivalences]\label{SAPTCG}
(1) Let $x$ and $y$ be $APTC_G$ terms. If $APTC\vdash x=y$, then $x\sim_{p} y$.

(2) Let $x$ and $y$ be $APTC_G$ terms. If $APTC\vdash x=y$, then $x\sim_{s} y$.

(3) Let $x$ and $y$ be $APTC_G$ terms. If $APTC\vdash x=y$, then $x\sim_{hp} y$.
\end{theorem}

\begin{proof}
(1) Since pomset bisimulation $\sim_p$ is both an equivalent and a congruent relation, we only need to check if each axiom in Table \ref{AxiomsForAPTCG} is sound modulo pomset bisimulation equivalence. We leave the proof as an exercise for the readers.

(2) Since step bisimulation $\sim_s$ is both an equivalent and a congruent relation, we only need to check if each axiom in Table \ref{AxiomsForAPTCG} is sound modulo step bisimulation equivalence. We leave the proof as an exercise for the readers.

(3) Since hp-bisimulation $\sim_{hp}$ is both an equivalent and a congruent relation, we only need to check if each axiom in Table \ref{AxiomsForAPTCG} is sound modulo hp-bisimulation equivalence. We leave the proof as an exercise for the readers.
\end{proof}

\begin{theorem}[Completeness of $APTC_G$ modulo truly concurrent bisimulation equivalences]\label{CAPTCG}
(1) Let $p$ and $q$ be closed $APTC_G$ terms, if $p\sim_{p} q$ then $p=q$.

(2) Let $p$ and $q$ be closed $APTC_G$ terms, if $p\sim_{s} q$ then $p=q$.

(3) Let $p$ and $q$ be closed $APTC_G$ terms, if $p\sim_{hp} q$ then $p=q$.
\end{theorem}

\begin{proof}
(1) Firstly, by the elimination theorem of $APTC_G$ (see Theorem \ref{ETAPTCG}), we know that for each closed $APTC_G$ term $p$, there exists a closed basic $APTC_G$ term $p'$, such that $APTC\vdash p=p'$, so, we only need to consider closed basic $APTC_G$ terms.

The basic terms (see Definition \ref{BTAPTCG}) modulo associativity and commutativity (AC) of conflict $+$ (defined by axioms $A1$ and $A2$ in Table \ref{AxiomsForBATCG}), and these equivalences is denoted by $=_{AC}$. Then, each equivalence class $s$ modulo AC of $+$ has the following normal form

$$s_1+\cdots+ s_k$$

with each $s_i$ either an atomic event, or an atomic guard, or of the form

$$t_1\cdot\cdots\cdot t_m$$

with each $t_j$ either an atomic event, or an atomic guard, or of the form

$$u_1\parallel\cdots\parallel u_l$$

with each $u_l$ an atomic event, or an atomic guard, and each $s_i$ is called the summand of $s$.

Now, we prove that for normal forms $n$ and $n'$, if $n\sim_{p} n'$ then $n=_{AC}n'$. It is sufficient to induct on the sizes of $n$ and $n'$.

\begin{itemize}
  \item Consider a summand $e$ of $n$. Then $\langle n,s\rangle\xrightarrow{e}\langle \underline{e},s'\rangle$, so $n\sim_p n'$ implies $\langle n',s\rangle\xrightarrow{e}\langle \underline{e},s\rangle$, meaning that $n'$ also contains the summand $e$.
  \item Consider a summand $\phi$ of $n$. Then $\langle n,s\rangle\rightarrow\langle \surd,s\rangle$, if $test(\phi,s)$ holds, so $n\sim_p n'$ implies $\langle n',s\rangle\rightarrow\langle \surd,s\rangle$, if $test(\phi,s)$ holds, meaning that $n'$ also contains the summand $\phi$.
  \item Consider a summand $t_1\cdot t_2$ of $n$,
  \begin{itemize}
    \item if $t_1\equiv e'$, then $\langle n,s\rangle\xrightarrow{e'}\langle t_2,s'\rangle$, so $n\sim_p n'$ implies $\langle n',s\rangle\xrightarrow{e'}\langle t_2',s'\rangle$ with $t_2\sim_p t_2'$, meaning that $n'$ contains a summand $e'\cdot t_2'$. Since $t_2$ and $t_2'$ are normal forms and have sizes smaller than $n$ and $n'$, by the induction hypotheses if $t_2\sim_p t_2'$ then $t_2=_{AC} t_2'$;
    \item if $t_1\equiv \phi'$, then $\langle n,s\rangle\rightarrow\langle t_2,s\rangle$, if $test(\phi',s)$ holds, so $n\sim_p n'$ implies $\langle n',s\rangle\rightarrow\langle t_2',s\rangle$ with $t_2\sim_p t_2'$, if $test(\phi',s)$ holds, meaning that $n'$ contains a summand $\phi'\cdot t_2'$. Since $t_2$ and $t_2'$ are normal forms and have sizes smaller than $n$ and $n'$, by the induction hypotheses if $t_2\sim_p t_2'$ then $t_2=_{AC} t_2'$;
    \item if $t_1\equiv e_1\parallel\cdots\parallel e_l$, then $\langle n,s\rangle\xrightarrow{\{e_1,\cdots,e_l\}}\langle t_2,s'\rangle$, so $n\sim_p n'$ implies $\langle n',s\rangle\xrightarrow{\{e_1,\cdots,e_l\}}\langle t_2',s'\rangle$ with $t_2\sim_p t_2'$, meaning that $n'$ contains a summand $(e_1\parallel\cdots\parallel e_l)\cdot t_2'$. Since $t_2$ and $t_2'$ are normal forms and have sizes smaller than $n$ and $n'$, by the induction hypotheses if $t_2\sim_p t_2'$ then $t_2=_{AC} t_2'$;
    \item if $t_1\equiv \phi_1\parallel\cdots\parallel \phi_l$, then $\langle n,s\rangle\rightarrow\langle t_2,s\rangle$, if $test(\phi_1,s),\cdots,test(\phi_l,s)$ hold, so $n\sim_p n'$ implies $\langle n',s\rangle\rightarrow\langle t_2',s\rangle$ with $t_2\sim_p t_2'$, if $test(\phi_1,s),\cdots,test(\phi_l,s)$ hold, meaning that $n'$ contains a summand $(\phi_1\parallel\cdots\parallel \phi_l)\cdot t_2'$. Since $t_2$ and $t_2'$ are normal forms and have sizes smaller than $n$ and $n'$, by the induction hypotheses if $t_2\sim_p t_2'$ then $t_2=_{AC} t_2'$.
  \end{itemize}
\end{itemize}

So, we get $n=_{AC} n'$.

Finally, let $s$ and $t$ be basic $APTC_G$ terms, and $s\sim_p t$, there are normal forms $n$ and $n'$, such that $s=n$ and $t=n'$. The soundness theorem of $APTC_G$ modulo pomset bisimulation equivalence (see Theorem \ref{SAPTCG}) yields $s\sim_p n$ and $t\sim_p n'$, so $n\sim_p s\sim_p t\sim_p n'$. Since if $n\sim_p n'$ then $n=_{AC}n'$, $s=n=_{AC}n'=t$, as desired.

(2) It can be proven similarly as (1).

(3) It can be proven similarly as (1).
\end{proof}

\subsubsection{Recursion}{\label{recg}}

In this subsection, we introduce recursion to capture infinite processes based on $APTC_G$. In the following, $E,F,G$ are recursion specifications, $X,Y,Z$ are recursive variables.

\begin{definition}[Guarded recursive specification]
A recursive specification

$$X_1=t_1(X_1,\cdots,X_n)$$
$$...$$
$$X_n=t_n(X_1,\cdots,X_n)$$

is guarded if the right-hand sides of its recursive equations can be adapted to the form by applications of the axioms in $APTC$ and replacing recursion variables by the right-hand sides of their recursive equations,

$$(a_{11}\parallel\cdots\parallel a_{1i_1})\cdot s_1(X_1,\cdots,X_n)+\cdots+(a_{k1}\parallel\cdots\parallel a_{ki_k})\cdot s_k(X_1,\cdots,X_n)+(b_{11}\parallel\cdots\parallel b_{1j_1})+\cdots+(b_{1j_1}\parallel\cdots\parallel b_{lj_l})$$

where $a_{11},\cdots,a_{1i_1},a_{k1},\cdots,a_{ki_k},b_{11},\cdots,b_{1j_1},b_{1j_1},\cdots,b_{lj_l}\in \mathbb{E}$, and the sum above is allowed to be empty, in which case it represents the deadlock $0$. And there does not exist an infinite sequence of $1$-transitions $\langle X|E\rangle\rightarrow\langle X'|E\rangle\rightarrow\langle X''|E\rangle\rightarrow\cdots$.
\end{definition}

\begin{center}
    \begin{table}
        $$\frac{\langle t_i(\langle X_1|E\rangle,\cdots,\langle X_n|E\rangle),s\rangle\xrightarrow{\{e_1,\cdots,e_k\}}\langle\underline{e_1}\parallel\cdots\parallel\underline{e_k},s'\rangle}{\langle\langle X_i|E\rangle,s\rangle\xrightarrow{\{e_1,\cdots,e_k\}}\langle\underline{e_1}\parallel\cdots\parallel\underline{e_k},s'\rangle}$$
        $$\frac{\langle t_i(\langle X_1|E\rangle,\cdots,\langle X_n|E\rangle),s\rangle\xrightarrow{\{e_1,\cdots,e_k\}} \langle y,s'\rangle}{\langle\langle X_i|E\rangle,s\rangle\xrightarrow{\{e_1,\cdots,e_k\}} \langle y,s'\rangle}$$
        \caption{Transition rules of guarded recursion}
        \label{TRForGRG}
    \end{table}
\end{center}

\begin{theorem}[Conservitivity of $APTC_G$ with guarded recursion]
$APTC_G$ with guarded recursion is a conservative extension of $APTC_G$.
\end{theorem}

\begin{proof}
Since the transition rules of $APTC_G$ are source-dependent, and the transition rules for guarded recursion in Table \ref{TRForGRG} contain only a fresh constant in their source, so the transition rules of $APTC_G$ with guarded recursion are conservative extensions of those of $APTC_G$.
\end{proof}

\begin{theorem}[Congruence theorem of $APTC_G$ with guarded recursion]
Truly concurrent bisimulation equivalences $\sim_{p}$, $\sim_s$ and $\sim_{hp}$ are all congruences with respect to $APTC_G$ with guarded recursion.
\end{theorem}

\begin{proof}
It is easy to see that pomset bisimulation $\sim_{p}$, step bisimulation $\sim_{s}$ and hp-bisimulation $\sim_{hp}$ are all equivalent relations on $APTC_G$ with guarded recursion terms. Since the TSS in Table \ref{TRForGRG} is in Panth format, by Theorem \ref{tcbacpn}, pomset bisimulation $\sim_{p}$, step bisimulation $\sim_{s}$ and hp-bisimulation $\sim_{hp}$ are all congruence with respect to $APTC_G$ with guarded recursion.
\end{proof}

\begin{theorem}[Elimination theorem of $APTC_G$ with linear recursion]\label{ETRecursionG}
Each process term in $APTC_G$ with linear recursion is equal to a process term $\langle X_1|E\rangle$ with $E$ a linear recursive specification.
\end{theorem}

\begin{proof}
By applying structural induction with respect to term size, each process term $t_1$ in $APTC_G$ with linear recursion generates a process can be expressed in the form of equations

$$t_i=(a_{i11}\parallel\cdots\parallel a_{i1i_1})t_{i1}+\cdots+(a_{ik_i1}\parallel\cdots\parallel a_{ik_ii_k})t_{ik_i}+(b_{i11}\parallel\cdots\parallel b_{i1i_1})+\cdots+(b_{il_i1}\parallel\cdots\parallel b_{il_ii_l})$$

for $i\in\{1,\cdots,n\}$. Let the linear recursive specification $E$ consist of the recursive equations

$$X_i=(a_{i11}\parallel\cdots\parallel a_{i1i_1})X_{i1}+\cdots+(a_{ik_i1}\parallel\cdots\parallel a_{ik_ii_k})X_{ik_i}+(b_{i11}\parallel\cdots\parallel b_{i1i_1})+\cdots+(b_{il_i1}\parallel\cdots\parallel b_{il_ii_l})$$

for $i\in\{1,\cdots,n\}$. Replacing $X_i$ by $t_i$ for $i\in\{1,\cdots,n\}$ is a solution for $E$, $RSP$ yields $t_1=\langle X_1|E\rangle$.
\end{proof}

\begin{theorem}[Soundness of $APTC_G$ with guarded recursion]\label{SAPTC_GRG}
Let $x$ and $y$ be $APTC_G$ with guarded recursion terms. If $APTC_G\textrm{ with guarded recursion}\vdash x=y$, then

(1) $x\sim_{s} y$.

(2) $x\sim_{p} y$.

(3) $x\sim_{hp} y$.
\end{theorem}

\begin{proof}
(1) Since step bisimulation $\sim_s$ is both an equivalent and a congruent relation with respect to $APTC_G$ with guarded recursion, we only need to check if each axiom in Table \ref{RDPRSP} is sound modulo step bisimulation equivalence. We leave them as exercises to the readers.

(2) Since pomset bisimulation $\sim_{p}$ is both an equivalent and a congruent relation with respect to the guarded recursion, we only need to check if each axiom in Table \ref{RDPRSP} is sound modulo pomset bisimulation equivalence. We leave them as exercises to the readers.

(3) Since hp-bisimulation $\sim_{hp}$ is both an equivalent and a congruent relation with respect to guarded recursion, we only need to check if each axiom in Table \ref{RDPRSP} is sound modulo hp-bisimulation equivalence. We leave them as exercises to the readers.
\end{proof}

\begin{theorem}[Completeness of $APTC_G$ with linear recursion]\label{CAPTC_GRG}
Let $p$ and $q$ be closed $APTC_G$ with linear recursion terms, then,

(1) if $p\sim_{s} q$ then $p=q$.

(2) if $p\sim_{p} q$ then $p=q$.

(3) if $p\sim_{hp} q$ then $p=q$.
\end{theorem}

\begin{proof}
Firstly, by the elimination theorem of $APTC_G$ with guarded recursion (see Theorem \ref{ETRecursionG}), we know that each process term in $APTC_G$ with linear recursion is equal to a process term $\langle X_1|E\rangle$ with $E$ a linear recursive specification. And for the simplicity, without loss of generalization, we do not consider empty event $1$, just because recursion with $1$ are similar to that with silent event $\tau$, please refer to the proof of Theorem \ref{CAPTCTAU} for details.

It remains to prove the following cases.

(1) If $\langle X_1|E_1\rangle \sim_s \langle Y_1|E_2\rangle$ for linear recursive specification $E_1$ and $E_2$, then $\langle X_1|E_1\rangle = \langle Y_1|E_2\rangle$.

Let $E_1$ consist of recursive equations $X=t_X$ for $X\in \mathcal{X}$ and $E_2$
consists of recursion equations $Y=t_Y$ for $Y\in\mathcal{Y}$. Let the linear recursive specification $E$ consist of recursion equations $Z_{XY}=t_{XY}$, and $\langle X|E_1\rangle\sim_s\langle Y|E_2\rangle$, and $t_{XY}$ consists of the following summands:

\begin{enumerate}
  \item $t_{XY}$ contains a summand $(a_1\parallel\cdots\parallel a_m)Z_{X'Y'}$ iff $t_X$ contains the summand $(a_1\parallel\cdots\parallel a_m)X'$ and $t_Y$ contains the summand $(a_1\parallel\cdots\parallel a_m)Y'$ such that $\langle X'|E_1\rangle\sim_s\langle Y'|E_2\rangle$;
  \item $t_{XY}$ contains a summand $b_1\parallel\cdots\parallel b_n$ iff $t_X$ contains the summand $b_1\parallel\cdots\parallel b_n$ and $t_Y$ contains the summand $b_1\parallel\cdots\parallel b_n$.
\end{enumerate}

Let $\sigma$ map recursion variable $X$ in $E_1$ to $\langle X|E_1\rangle$, and let $\pi$ map recursion variable $Z_{XY}$ in $E$ to $\langle X|E_1\rangle$. So, $\sigma((a_1\parallel\cdots\parallel a_m)X')\equiv(a_1\parallel\cdots\parallel a_m)\langle X'|E_1\rangle\equiv\pi((a_1\parallel\cdots\parallel a_m)Z_{X'Y'})$, so by $RDP$, we get $\langle X|E_1\rangle=\sigma(t_X)=\pi(t_{XY})$. Then by $RSP$, $\langle X|E_1\rangle=\langle Z_{XY}|E\rangle$, particularly, $\langle X_1|E_1\rangle=\langle Z_{X_1Y_1}|E\rangle$. Similarly, we can obtain $\langle Y_1|E_2\rangle=\langle Z_{X_1Y_1}|E\rangle$. Finally, $\langle X_1|E_1\rangle=\langle Z_{X_1Y_1}|E\rangle=\langle Y_1|E_2\rangle$, as desired.

(2) If $\langle X_1|E_1\rangle \sim_p \langle Y_1|E_2\rangle$ for linear recursive specification $E_1$ and $E_2$, then $\langle X_1|E_1\rangle = \langle Y_1|E_2\rangle$.

It can be proven similarly to (1), we omit it.

(3) If $\langle X_1|E_1\rangle \sim_{hp} \langle Y_1|E_2\rangle$ for linear recursive specification $E_1$ and $E_2$, then $\langle X_1|E_1\rangle = \langle Y_1|E_2\rangle$.

It can be proven similarly to (1), we omit it.
\end{proof}

\subsubsection{Abstraction}{\label{absg}}

To abstract away from the internal implementations of a program, and verify that the program exhibits the desired external behaviors, the silent step $\tau$ and abstraction operator $\tau_I$ are introduced, where $I\subseteq \mathbb{E}\cup \mathbb{B}_{at}$ denotes the internal events or guards. The silent step $\tau$ represents the internal events, and $\tau_{\phi}$ for internal guards, when we consider the external behaviors of a process, $\tau$ steps can be removed, that is, $\tau$ steps must keep silent. The transition rule of $\tau$ is shown in Table \ref{TRForTauG}. In the following, let the atomic event $e$ range over $\mathbb{E}\cup\{1\}\cup\{0\}\cup\{\tau\}$, and $\phi$ range over $\mathbb{B}\cup \{\tau\}$, and let the communication function $\gamma:\mathbb{E}\cup\{\tau\}\times \mathbb{E}\cup\{\tau\}\rightarrow \mathbb{E}\cup\{0\}$, with each communication involved $\tau$ resulting in $0$. We use $\tau(s)$ to denote $effect(\tau,s)$, for the fact that $\tau$ only change the state of internal data environment, that is, for the external data environments, $s=\tau(s)$.

\begin{center}
    \begin{table}
        $$\frac{}{\langle\tau_{\phi},s\rangle\rightarrow\langle\surd,s\rangle}\textrm{ if }test(\tau_{\phi},s)$$
        $$\frac{}{\langle\tau,s\rangle\xrightarrow{\tau}\langle\surd,\tau(s)\rangle}$$
        \caption{Transition rule of the silent step}
        \label{TRForTauG}
    \end{table}
\end{center}

\begin{definition}[Guarded linear recursive specification]\label{GLRSG}
A linear recursive specification $E$ is guarded if there does not exist an infinite sequence of $\tau$-transitions $\langle X|E\rangle\xrightarrow{\tau}\langle X'|E\rangle\xrightarrow{\tau}\langle X''|E\rangle\xrightarrow{\tau}\cdots$, and there does not exist an infinite sequence of $1$-transitions $\langle X|E\rangle\rightarrow\langle X'|E\rangle\rightarrow\langle X''|E\rangle\rightarrow\cdots$.
\end{definition}

\begin{theorem}[Conservitivity of $APTC_G$ with silent step and guarded linear recursion]
$APTC_G$ with silent step and guarded linear recursion is a conservative extension of $APTC_G$ with linear recursion.
\end{theorem}

\begin{proof}
Since the transition rules of $APTC_G$ with linear recursion are source-dependent, and the transition rules for silent step in Table \ref{TRForTauG} contain only a fresh constant $\tau$ in their source, so the transition rules of $APTC_G$ with silent step and guarded linear recursion is a conservative extension of those of $APTC_G$ with linear recursion.
\end{proof}

\begin{theorem}[Congruence theorem of $APTC_G$ with silent step and guarded linear recursion]
Rooted branching truly concurrent bisimulation equivalences $\approx_{rbp}$, $\approx_{rbs}$ and $\approx_{rbhp}$ are all congruences with respect to $APTC_G$ with silent step and guarded linear recursion.
\end{theorem}

\begin{proof}
It is easy to see that rooted branching pomset bisimulation $\approx_{rbp}$, rooted branching step bisimulation $\approx_{rbs}$ and rooted branching hp-bisimulation $\approx_{rbhp}$ are all equivalent relations on $APTC_G$ with silent step and guarded linear recursion terms. Since the TSS in Table \ref{TRForTauG} is in RBB cool format, by Theorem \ref{rbtcbacpn}, rooted branching pomset bisimulation $\approx_{rbp}$, rooted branching step bisimulation $\approx_{rbs}$ and rooted branching hp-bisimulation $\approx_{rbhp}$ are all congruence with respect to $APTC_G$ with silent step and guarded linear recursion.
\end{proof}

We design the axioms for the silent step $\tau$ in Table \ref{AxiomsForTauG}.

\begin{center}
\begin{table}
  \begin{tabular}{@{}ll@{}}
\hline No. &Axiom\\
  $B1$ & $e\cdot\tau=e$\\
  $B2$ & $e\cdot(\tau\cdot(x+y)+x)=e\cdot(x+y)$\\
  $B3$ & $x\parallel\tau=x$\\
  $G25$ & $\tau_{\phi}\cdot x=x$\\
  $G26$ & $x\cdot \tau_{\phi} = x$\\
  $G27$ & $x\parallel\tau_{\phi} = x$\\
\end{tabular}
\caption{Axioms of silent step}
\label{AxiomsForTauG}
\end{table}
\end{center}

\begin{theorem}[Elimination theorem of $APTC_G$ with silent step and guarded linear recursion]\label{ETTauG}
Each process term in $APTC_G$ with silent step and guarded linear recursion is equal to a process term $\langle X_1|E\rangle$ with $E$ a guarded linear recursive specification.
\end{theorem}

\begin{proof}
By applying structural induction with respect to term size, each process term $t_1$ in $APTC_G$ with silent step and guarded linear recursion generates a process can be expressed in the form of equations

$$t_i=(a_{i11}\parallel\cdots\parallel a_{i1i_1})t_{i1}+\cdots+(a_{ik_i1}\parallel\cdots\parallel a_{ik_ii_k})t_{ik_i}+(b_{i11}\parallel\cdots\parallel b_{i1i_1})+\cdots+(b_{il_i1}\parallel\cdots\parallel b_{il_ii_l})$$

for $i\in\{1,\cdots,n\}$. Let the linear recursive specification $E$ consist of the recursive equations

$$X_i=(a_{i11}\parallel\cdots\parallel a_{i1i_1})X_{i1}+\cdots+(a_{ik_i1}\parallel\cdots\parallel a_{ik_ii_k})X_{ik_i}+(b_{i11}\parallel\cdots\parallel b_{i1i_1})+\cdots+(b_{il_i1}\parallel\cdots\parallel b_{il_ii_l})$$

for $i\in\{1,\cdots,n\}$. Replacing $X_i$ by $t_i$ for $i\in\{1,\cdots,n\}$ is a solution for $E$, $RSP$ yields $t_1=\langle X_1|E\rangle$.
\end{proof}

\begin{theorem}[Soundness of $APTC_G$ with silent step and guarded linear recursion]\label{SAPTC_GTAUG}
Let $x$ and $y$ be $APTC_G$ with silent step and guarded linear recursion terms. If $APTC_G$ with silent step and guarded linear recursion $\vdash x=y$, then

(1) $x\approx_{rbs} y$.

(2) $x\approx_{rbp} y$.

(3) $x\approx_{rbhp} y$.
\end{theorem}

\begin{proof}
(1) Since rooted branching step bisimulation $\approx_{rbs}$ is both an equivalent and a congruent relation with respect to $APTC_G$ with silent step and guarded linear recursion, we only need to check if each axiom in Table \ref{AxiomsForTauG} is sound modulo rooted branching step bisimulation equivalence. We leave them as exercises to the readers.

(2) Since rooted branching pomset bisimulation $\approx_{rbp}$ is both an equivalent and a congruent relation with respect to $APTC_G$ with silent step and guarded linear recursion, we only need to check if each axiom in Table \ref{AxiomsForTauG} is sound modulo rooted branching pomset bisimulation $\approx_{rbp}$. We leave them as exercises to the readers.

(3) Since rooted branching hp-bisimulation $\approx_{rbhp}$ is both an equivalent and a congruent relation with respect to $APTC_G$ with silent step and guarded linear recursion, we only need to check if each axiom in Table \ref{AxiomsForTauG} is sound modulo rooted branching hp-bisimulation $\approx_{rbhp}$. We leave them as exercises to the readers.
\end{proof}

\begin{theorem}[Completeness of $APTC_G$ with silent step and guarded linear recursion]\label{CAPTC_GTAUG}
Let $p$ and $q$ be closed $APTC_G$ with silent step and guarded linear recursion terms, then,

(1) if $p\approx_{rbs} q$ then $p=q$.

(2) if $p\approx_{rbp} q$ then $p=q$.

(3) if $p\approx_{rbhp} q$ then $p=q$.
\end{theorem}

\begin{proof}
Firstly, by the elimination theorem of $APTC_G$ with silent step and guarded linear recursion (see Theorem \ref{ETTauG}), we know that each process term in $APTC_G$ with silent step and guarded linear recursion is equal to a process term $\langle X_1|E\rangle$ with $E$ a guarded linear recursive specification.

It remains to prove the following cases.

(1) If $\langle X_1|E_1\rangle \approx_{rbs} \langle Y_1|E_2\rangle$ for guarded linear recursive specification $E_1$ and $E_2$, then $\langle X_1|E_1\rangle = \langle Y_1|E_2\rangle$.

Firstly, the recursive equation $W=\tau+\cdots+\tau$ with $W\nequiv X_1$ in $E_1$ and $E_2$, can be removed, and the corresponding summands $aW$ are replaced by $a$, to get $E_1'$ and $E_2'$, by use of the axioms $RDP$, $A3$ and $B1$, and $\langle X|E_1\rangle = \langle X|E_1'\rangle$, $\langle Y|E_2\rangle = \langle Y|E_2'\rangle$.

Let $E_1$ consists of recursive equations $X=t_X$ for $X\in \mathcal{X}$ and $E_2$
consists of recursion equations $Y=t_Y$ for $Y\in\mathcal{Y}$, and are not the form $\tau+\cdots+\tau$. Let the guarded linear recursive specification $E$ consists of recursion equations $Z_{XY}=t_{XY}$, and $\langle X|E_1\rangle\approx_{rbs}\langle Y|E_2\rangle$, and $t_{XY}$ consists of the following summands:

\begin{enumerate}
  \item $t_{XY}$ contains a summand $(a_1\parallel\cdots\parallel a_m)Z_{X'Y'}$ iff $t_X$ contains the summand $(a_1\parallel\cdots\parallel a_m)X'$ and $t_Y$ contains the summand $(a_1\parallel\cdots\parallel a_m)Y'$ such that $\langle X'|E_1\rangle\approx_{rbs}\langle Y'|E_2\rangle$;
  \item $t_{XY}$ contains a summand $b_1\parallel\cdots\parallel b_n$ iff $t_X$ contains the summand $b_1\parallel\cdots\parallel b_n$ and $t_Y$ contains the summand $b_1\parallel\cdots\parallel b_n$;
  \item $t_{XY}$ contains a summand $\tau Z_{X'Y}$ iff $XY\nequiv X_1Y_1$, $t_X$ contains the summand $\tau X'$, and $\langle X'|E_1\rangle\approx_{rbs}\langle Y|E_2\rangle$;
  \item $t_{XY}$ contains a summand $\tau Z_{XY'}$ iff $XY\nequiv X_1Y_1$, $t_Y$ contains the summand $\tau Y'$, and $\langle X|E_1\rangle\approx_{rbs}\langle Y'|E_2\rangle$.
\end{enumerate}

Since $E_1$ and $E_2$ are guarded, $E$ is guarded. Constructing the process term $u_{XY}$ consist of the following summands:

\begin{enumerate}
  \item $u_{XY}$ contains a summand $(a_1\parallel\cdots\parallel a_m)\langle X'|E_1\rangle$ iff $t_X$ contains the summand $(a_1\parallel\cdots\parallel a_m)X'$ and $t_Y$ contains the summand $(a_1\parallel\cdots\parallel a_m)Y'$ such that $\langle X'|E_1\rangle\approx_{rbs}\langle Y'|E_2\rangle$;
  \item $u_{XY}$ contains a summand $b_1\parallel\cdots\parallel b_n$ iff $t_X$ contains the summand $b_1\parallel\cdots\parallel b_n$ and $t_Y$ contains the summand $b_1\parallel\cdots\parallel b_n$;
  \item $u_{XY}$ contains a summand $\tau \langle X'|E_1\rangle$ iff $XY\nequiv X_1Y_1$, $t_X$ contains the summand $\tau X'$, and $\langle X'|E_1\rangle\approx_{rbs}\langle Y|E_2\rangle$.
\end{enumerate}

Let the process term $s_{XY}$ be defined as follows:

\begin{enumerate}
  \item $s_{XY}\triangleq\tau\langle X|E_1\rangle + u_{XY}$ iff $XY\nequiv X_1Y_1$, $t_Y$ contains the summand $\tau Y'$, and $\langle X|E_1\rangle\approx_{rbs}\langle Y'|E_2\rangle$;
  \item $s_{XY}\triangleq\langle X|E_1\rangle$, otherwise.
\end{enumerate}

So, $\langle X|E_1\rangle=\langle X|E_1\rangle+u_{XY}$, and $(a_1\parallel\cdots\parallel a_m)(\tau\langle X|E_1\rangle+u_{XY})=(a_1\parallel\cdots\parallel a_m)((\tau\langle X|E_1\rangle+u_{XY})+u_{XY})=(a_1\parallel\cdots\parallel a_m)(\langle X|E_1\rangle+u_{XY})=(a_1\parallel\cdots\parallel a_m)\langle X|E_1\rangle$, hence, $(a_1\parallel\cdots\parallel a_m)s_{XY}=(a_1\parallel\cdots\parallel a_m)\langle X|E_1\rangle$.

Let $\sigma$ map recursion variable $X$ in $E_1$ to $\langle X|E_1\rangle$, and let $\pi$ map recursion variable $Z_{XY}$ in $E$ to $s_{XY}$. It is sufficient to prove $s_{XY}=\pi(t_{XY})$ for recursion variables $Z_{XY}$ in $E$. Either $XY\equiv X_1Y_1$ or $XY\nequiv X_1Y_1$, we all can get $s_{XY}=\pi(t_{XY})$. So, $s_{XY}=\langle Z_{XY}|E\rangle$ for recursive variables $Z_{XY}$ in $E$ is a solution for $E$. Then by $RSP$, particularly, $\langle X_1|E_1\rangle=\langle Z_{X_1Y_1}|E\rangle$. Similarly, we can obtain $\langle Y_1|E_2\rangle=\langle Z_{X_1Y_1}|E\rangle$. Finally, $\langle X_1|E_1\rangle=\langle Z_{X_1Y_1}|E\rangle=\langle Y_1|E_2\rangle$, as desired.

(2) If $\langle X_1|E_1\rangle \approx_{rbp} \langle Y_1|E_2\rangle$ for guarded linear recursive specification $E_1$ and $E_2$, then $\langle X_1|E_1\rangle = \langle Y_1|E_2\rangle$.

It can be proven similarly to (1), we omit it.

(3) If $\langle X_1|E_1\rangle \approx_{rbhb} \langle Y_1|E_2\rangle$ for guarded linear recursive specification $E_1$ and $E_2$, then $\langle X_1|E_1\rangle = \langle Y_1|E_2\rangle$.

It can be proven similarly to (1), we omit it.
\end{proof}

The unary abstraction operator $\tau_I$ ($I\subseteq \mathbb{E}\cup \mathbb{B}_{at}$) renames all atomic events or atomic guards in $I$ into $\tau$. $APTC_G$ with silent step and abstraction operator is called $APTC_{G_{\tau}}$. The transition rules of operator $\tau_I$ are shown in Table \ref{TRForAbstractionG}.

\begin{center}
    \begin{table}
        $$\frac{\langle x,s\rangle\xrightarrow{e}\langle\underline{e},s'\rangle}{\langle\tau_I(x),s\rangle\xrightarrow{e}\langle\underline{e},s'\rangle}\quad e\notin I
        \quad\quad\frac{\langle x,s\rangle\xrightarrow{e}\langle x',s'\rangle}{\langle\tau_I(x),s\rangle\xrightarrow{e}\langle\tau_I(x'),s'\rangle}\quad e\notin I$$

        $$\frac{\langle x,s\rangle\xrightarrow{e}\langle\surd,s'\rangle}{\langle\tau_I(x),s\rangle\xrightarrow{\tau}\langle\surd,\tau(s)\rangle}\quad e\in I
        \quad\quad\frac{\langle x,s\rangle\xrightarrow{e}\langle x',s'\rangle}{\langle\tau_I(x),s\rangle\xrightarrow{\tau}\langle\tau_I(x'),\tau(s)\rangle}\quad e\in I$$
        \caption{Transition rule of the abstraction operator}
        \label{TRForAbstractionG}
    \end{table}
\end{center}

\begin{theorem}[Conservitivity of $APTC_{G_{\tau}}$ with guarded linear recursion]
$APTC_{G_{\tau}}$ with guarded linear recursion is a conservative extension of $APTC_G$ with silent step and guarded linear recursion.
\end{theorem}

\begin{proof}
Since the transition rules of $APTC_G$ with silent step and guarded linear recursion are source-dependent, and the transition rules for abstraction operator in Table \ref{TRForAbstractionG} contain only a fresh operator $\tau_I$ in their source, so the transition rules of $APTC_{G_{\tau}}$ with guarded linear recursion is a conservative extension of those of $APTC_G$ with silent step and guarded linear recursion.
\end{proof}

\begin{theorem}[Congruence theorem of $APTC_{G_{\tau}}$ with guarded linear recursion]
Rooted branching truly concurrent bisimulation equivalences $\approx_{rbp}$, $\approx_{rbs}$ and $\approx_{rbhp}$ are all congruences with respect to $APTC_{G_{\tau}}$ with guarded linear recursion.
\end{theorem}

\begin{proof}
It is easy to see that rooted branching pomset bisimulation $\approx_{rbp}$, rooted branching step bisimulation $\approx_{rbs}$ and rooted branching hp-bisimulation $\approx_{rbhp}$ are all equivalent relations on $APTC_{G_{\tau}}$ with guarded linear recursion terms. Since the TSS in Table \ref{TRForAbstractionG} is in RBB cool format, by Theorem \ref{rbtcbacpn}, rooted branching pomset bisimulation $\approx_{rbp}$, rooted branching step bisimulation $\approx_{rbs}$ and rooted branching hp-bisimulation $\approx_{rbhp}$ are all congruence with respect to $APTC_{G_{\tau}}$ with guarded linear recursion.
\end{proof}

We design the axioms for the abstraction operator $\tau_I$ in Table \ref{AxiomsForAbstractionG}.

\begin{center}
\begin{table}
  \begin{tabular}{@{}ll@{}}
\hline No. &Axiom\\
  $TI1$ & $e\notin I\quad \tau_I(e)=e$\\
  $TI2$ & $e\in I\quad \tau_I(e)=\tau$\\
  $TI3$ & $\tau_I(0)=0$\\
  $TI4$ & $\tau_I(x+y)=\tau_I(x)+\tau_I(y)$\\
  $TI5$ & $\tau_I(x\cdot y)=\tau_I(x)\cdot\tau_I(y)$\\
  $TI6$ & $\tau_I(x\parallel y)=\tau_I(x)\parallel\tau_I(y)$\\
  $G28$ & $\phi\notin I\quad \tau_I(\phi)=\phi$\\
  $G29$ & $\phi\in I\quad \tau_I(\phi)=\tau_{\phi}$\\
\end{tabular}
\caption{Axioms of abstraction operator}
\label{AxiomsForAbstractionG}
\end{table}
\end{center}

\begin{theorem}[Soundness of $APTC_{G_{\tau}}$ with guarded linear recursion]\label{SAPTC_GABSG}
Let $x$ and $y$ be $APTC_{G_{\tau}}$ with guarded linear recursion terms. If $APTC_{G_{\tau}}$ with guarded linear recursion $\vdash x=y$, then

(1) $x\approx_{rbs} y$.

(2) $x\approx_{rbp} y$.

(3) $x\approx_{rbhp} y$.
\end{theorem}

\begin{proof}
(1) Since rooted branching step bisimulation $\approx_{rbs}$ is both an equivalent and a congruent relation with respect to $APTC_{G_{\tau}}$ with guarded linear recursion, we only need to check if each axiom in Table \ref{AxiomsForAbstractionG} is sound modulo rooted branching step bisimulation equivalence. We leave them as exercises to the readers.

(2) Since rooted branching pomset bisimulation $\approx_{rbp}$ is both an equivalent and a congruent relation with respect to $APTC_{G_{\tau}}$ with guarded linear recursion, we only need to check if each axiom in Table \ref{AxiomsForAbstractionG} is sound modulo rooted branching pomset bisimulation $\approx_{rbp}$. We leave them as exercises to the readers.

(3) Since rooted branching hp-bisimulation $\approx_{rbhp}$ is both an equivalent and a congruent relation with respect to $APTC_{G_{\tau}}$ with guarded linear recursion, we only need to check if each axiom in Table \ref{AxiomsForAbstractionG} is sound modulo rooted branching hp-bisimulation $\approx_{rbhp}$. We leave them as exercises to the readers.
\end{proof}

Though $\tau$-loops are prohibited in guarded linear recursive specifications (see Definition \ref{GLRSG}) in a specifiable way, they can be constructed using the abstraction operator, for example, there exist $\tau$-loops in the process term $\tau_{\{a\}}(\langle X|X=aX\rangle)$. To avoid $\tau$-loops caused by $\tau_I$ and ensure fairness, the concept of cluster and $CFAR$ (Cluster Fair Abstraction Rule) \cite{CFAR} (see chapter \ref{pe}) are still needed.

\begin{theorem}[Completeness of $APTC_{G_{\tau}}$ with guarded linear recursion and $CFAR$]\label{CCFARG}
Let $p$ and $q$ be closed $APTC_{G_{\tau}}$ with guarded linear recursion and $CFAR$ terms, then,

(1) if $p\approx_{rbs} q$ then $p=q$.

(2) if $p\approx_{rbp} q$ then $p=q$.

(3) if $p\approx_{rbhp} q$ then $p=q$.
\end{theorem}

\begin{proof}
(1) For the case of rooted branching step bisimulation, the proof is following.

Firstly, in the proof the Theorem \ref{CAPTC_GTAUG}, we know that each process term $p$ in $APTC_G$ with silent step and guarded linear recursion is equal to a process term $\langle X_1|E\rangle$ with $E$ a guarded linear recursive specification. And we prove if $\langle X_1|E_1\rangle\approx_{rbs}\langle Y_1|E_2\rangle$, then $\langle X_1|E_1\rangle=\langle Y_1|E_2\rangle$

The only new case is $p\equiv\tau_I(q)$. Let $q=\langle X|E\rangle$ with $E$ a guarded linear recursive specification, so $p=\tau_I(\langle X|E\rangle)$. Then the collection of recursive variables in $E$ can be divided into its clusters $C_1,\cdots,C_N$ for $I$. Let

$(a_{1i1}\parallel\cdots\parallel a_{k_{i1}i1}) Y_{i1}+\cdots+(a_{1im_i}\parallel\cdots\parallel a_{k_{im_i}im_i}) Y_{im_i}+b_{1i1}\parallel\cdots\parallel b_{l_{i1}i1}+\cdots+b_{1im_i}\parallel\cdots\parallel b_{l_{im_i}im_i}$

be the conflict composition of exits for the cluster $C_i$, with $i\in\{1,\cdots,N\}$.

For $Z\in C_i$ with $i\in\{1,\cdots,N\}$, we define

$s_Z\triangleq (\hat{a_{1i1}}\parallel\cdots\parallel \hat{a_{k_{i1}i1}}) \tau_I(\langle Y_{i1}|E\rangle)+\cdots+(\hat{a_{1im_i}}\parallel\cdots\parallel \hat{a_{k_{im_i}im_i}}) \tau_I(\langle Y_{im_i}|E\rangle)+\hat{b_{1i1}}\parallel\cdots\parallel \hat{b_{l_{i1}i1}}+\cdots+\hat{b_{1im_i}}\parallel\cdots\parallel \hat{b_{l_{im_i}im_i}}$

For $Z\in C_i$ and $a_1,\cdots,a_j\in \mathbb{E}\cup\{\tau\}$ with $j\in\mathbb{N}$, we have

$(a_1\parallel\cdots\parallel a_j)\tau_I(\langle Z|E\rangle)$

$=(a_1\parallel\cdots\parallel a_j)\tau_I((a_{1i1}\parallel\cdots\parallel a_{k_{i1}i1}) \langle Y_{i1}|E\rangle+\cdots+(a_{1im_i}\parallel\cdots\parallel a_{k_{im_i}im_i}) \langle Y_{im_i}|E\rangle+b_{1i1}\parallel\cdots\parallel b_{l_{i1}i1}+\cdots+b_{1im_i}\parallel\cdots\parallel b_{l_{im_i}im_i})$

$=(a_1\parallel\cdots\parallel a_j)s_Z$

Let the linear recursive specification $F$ contain the same recursive variables as $E$, for $Z\in C_i$, $F$ contains the following recursive equation

$Z=(\hat{a_{1i1}}\parallel\cdots\parallel \hat{a_{k_{i1}i1}}) Y_{i1}+\cdots+(\hat{a_{1im_i}}\parallel\cdots\parallel \hat{a_{k_{im_i}im_i}})  Y_{im_i}+\hat{b_{1i1}}\parallel\cdots\parallel \hat{b_{l_{i1}i1}}+\cdots+\hat{b_{1im_i}}\parallel\cdots\parallel \hat{b_{l_{im_i}im_i}}$

It is easy to see that there is no sequence of one or more $\tau$-transitions from $\langle Z|F\rangle$ to itself, so $F$ is guarded.

For

$s_Z=(\hat{a_{1i1}}\parallel\cdots\parallel \hat{a_{k_{i1}i1}}) Y_{i1}+\cdots+(\hat{a_{1im_i}}\parallel\cdots\parallel \hat{a_{k_{im_i}im_i}}) Y_{im_i}+\hat{b_{1i1}}\parallel\cdots\parallel \hat{b_{l_{i1}i1}}+\cdots+\hat{b_{1im_i}}\parallel\cdots\parallel \hat{b_{l_{im_i}im_i}}$

is a solution for $F$. So, $(a_1\parallel\cdots\parallel a_j)\tau_I(\langle Z|E\rangle)=(a_1\parallel\cdots\parallel a_j)s_Z=(a_1\parallel\cdots\parallel a_j)\langle Z|F\rangle$.

So,

$\langle Z|F\rangle=(\hat{a_{1i1}}\parallel\cdots\parallel \hat{a_{k_{i1}i1}}) \langle Y_{i1}|F\rangle+\cdots+(\hat{a_{1im_i}}\parallel\cdots\parallel \hat{a_{k_{im_i}im_i}}) \langle Y_{im_i}|F\rangle+\hat{b_{1i1}}\parallel\cdots\parallel \hat{b_{l_{i1}i1}}+\cdots+\hat{b_{1im_i}}\parallel\cdots\parallel \hat{b_{l_{im_i}im_i}}$

Hence, $\tau_I(\langle X|E\rangle=\langle Z|F\rangle)$, as desired.

(2) For the case of rooted branching pomset bisimulation, it can be proven similarly to (1), we omit it.

(3) For the case of rooted branching hp-bisimulation, it can be proven similarly to (1), we omit it.
\end{proof} 

%% file: section5.tex
\section{Process Algebra vs. Automata and Kleene Algebra}\label{pak}

In this chapter, we discuss the relationship between process algebra and automata, Kleene algebra. Firstly, we introduce the relationship between process algebra vs. automata and Kleene algebra in section \ref{paka}. Then, discuss the relationship between truly concurrent process algebra vs. truly concurrent automata and concurrent Kleene algebra in section \ref{tcpacka}.

\subsection{Process Algebra and Kleene Algebra}\label{paka}

\subsubsection{Automata}

\begin{definition}[Automaton]
An automaton is a tuple $\mathcal{A}=(Q,A,E,I,F)$ where $Q$ is a finite set of states, $A$ is an alphabet, $I\subseteq Q$ is the set of initial states, $F\subseteq Q$ is the set of final states, and $E$ is the finite set of the transitions of $\mathcal{A}$ and $E\subseteq Q\times A\times Q$.
\end{definition}

It is well-known that automata recognize regular languages.

\begin{definition}[Transition relation]
Let $p,q\in Q$. We define the transition relation $\xrightarrow[\mathcal{A}]{}\subseteq Q\times A\times Q$ on $\mathcal{A}$ as the smallest relation satisfying:

\begin{enumerate}
  \item $p\xrightarrow[\mathcal{A}]{\epsilon}p$ for all $p\in Q$;
  \item $p\xrightarrow[\mathcal{A}]{a}q$ if and only if $(p,a,q)\in E$.
\end{enumerate}
\end{definition}

\begin{definition}[Bisimulation]
Let $\mathcal{A}=(Q,A,E,I,F)$ and $\mathcal{A}'=(Q',A',E',I',F')$ be two automata with the same alphabet. The automata $\mathcal{A}$ and $\mathcal{A}'$ are bisimilar, $\mathcal{A}\underline{\leftrightarrow}\mathcal{A}'$, iff there is a relation $R$ between their reachable states that preserves transitions and termination:

\begin{enumerate}
  \item $R$ relate reachable states, i.e., every reachable state of $\mathcal{A}$ is related to a reachable state of $\mathcal{A}'$ and every reachable state of $\mathcal{A}'$ is related to a reachable state of $\mathcal{A}$;
  \item $I$ are related to $I'$ by $R$; 
  \item whenever $p$ is related to $p'$, $pRp'$ and $p\xrightarrow[\mathcal{A}]{a}q$, then there is state $q'$ in $\mathcal{A}'$ with $p'\xrightarrow[\mathcal{A}']{a}q'$ and $qRq'$;
  \item whenever $p$ is related to $p'$, $pRp'$ and $p'\xrightarrow[\mathcal{A}']{a}q'$, then there is state $q$ in $\mathcal{A}$ with $p\xrightarrow[\mathcal{A}]{a}q$ and $qRq'$;
  \item whenever $pRp'$, then $p\in F$ iff $p'\in F'$.
\end{enumerate}
\end{definition}

\subsubsection{Kleene Algebra}\label{ka} 

Here we give the version of Kleene algebra (KA) of Kozen \cite{KA3} which has a sound and complete finitary axiomatization.

Traditionally, we note Kleene algebra in the context of regular language and expressions. We fix a finite alphabet $\Sigma$ and a word formed over $\Sigma$ is a finite sequence of symbols from $\Sigma$, and the empty word is denoted $\epsilon$. Let $\Sigma^*$ denote the set of all words over $\Sigma$ and a language is a set of words. For words $u,v\in\Sigma^*$, we define $u\cdot v$ as the concatenation of $u$ and $v$, $u\cdot v=uv$. Then for $U,V\subseteq\Sigma^*$, we define $U\cdot V=\{uv|u\in U,v\in V\}$, $U+V=U\cup V$, $U^*=\bigcup_{n\in\mathbb{N}}U^n$ where $U^0=\{\epsilon\}$ and $U^{n+1}=U\cdot U^n$.

\begin{definition}[Syntax of regular expressions]
We define the set of regular expressions $\mathcal{T}_{KA}$ as follows.

$$\mathcal{T}_{KA}\ni x,y::=0|1|a\in\Sigma|x+y|x\cdot y|x^*$$
\end{definition}

\begin{definition}[Semantics of regular expressions]
We define the interpretation of regular expressions $\sembrack{-}:\mathcal{T}_{KA}\rightarrow\mathcal{P}(\Sigma^*)$ inductively as Table \ref{SRE} shows.
\end{definition}

\begin{center}
    \begin{table}
        $$\sembrack{0}_{KA}=\emptyset \quad \sembrack{a}_{KA}=\{a\} \quad \sembrack{x\cdot y}_{KA}=\sembrack{x}_{KA}\cdot \sembrack{y}_{KA}$$
        $$\sembrack{1}_{KA}=\{\epsilon\} \quad \sembrack{x+y}_{KA}=\sembrack{x}_{KA}+\sembrack{y}_{KA} \quad\sembrack{x^*}_{KA}=\sembrack{x}^*_{KA}$$
        \caption{Semantics of regular expressions}
        \label{SRE}
    \end{table}
\end{center}

We define a Kleene algebra as a tuple $(A,+,\cdot,^*,0,1)$, where $A$ is a set, $^*$ is a unary operator, $+$ and $\cdot$ are binary operators, and $0$ and $1$ are constants, which satisfies the axioms in Table \ref{AxiomsForKA} for all $x,y,z\in \mathcal{T}_{KA}$, where $x\preceq y$ means $x+y=y$.

\begin{center}
    \begin{table}
        \begin{tabular}{@{}ll@{}}
            \hline No. &Axiom\\
            $A1$ & $x+y=y+z$\\
            $A2$ & $x+(y+z)=(x+y)+z$\\
            $A3$ & $x+x=x$\\
            $A4$ & $(x+y)\cdot z=x\cdot z+y\cdot z$\\
            $A5$ & $x\cdot(y+z)=x\cdot y+x\cdot z$\\
            $A6$ & $x\cdot(y\cdot z)=(x\cdot y)\cdot z$\\
            $A7$ & $x+0=x$\\
            $A8$ & $0\cdot x=0$\\
            $A9$ & $x\cdot 0=0$\\
            $A10$ & $x\cdot 1=x$\\
            $A11$ & $1\cdot x=x$\\
            $A12$ & $1+x\cdot x^*=x^*$\\
            $A13$ & $1+x^*\cdot x=x^*$\\
            $A14$ & $x+y\cdot z\preceq z\Rightarrow y^*\cdot x\preceq z$\\
            $A15$ & $x+y\cdot z\preceq y\Rightarrow x\cdot z^*\preceq y$\\
        \end{tabular}
        \caption{Axioms of Kleene algebra}
        \label{AxiomsForKA}
    \end{table}
\end{center}

Let $\equiv_{KA}$ denote the smallest congruence on $\mathcal{T}_{KA}$ induced by the axioms of the Kleene algebra. Then we can get the following soundness and completeness theorem, which is proven by Kozen \cite{KA3}.

\begin{theorem}[Soundness and completeness of Kleene algebra]
For all $x,y\in\mathcal{T}_{KA}$, $x\equiv_{KA} y$ iff $\sembrack{x}_{KA}=\sembrack{y}_{KA}$.
\end{theorem}

\subsubsection{Milner's Proof System for Regular Expressions Modulo Bisimilarity}{\label{mf}}

Process algebras CCS \cite{CC} \cite{CCS} and ACP \cite{ACP} have a bisimilarity-based operational semantics, on operational semantics, please refer to subsection \ref{os}.

Milner wanted to give regular expressions a bisimilarity-based semantic foundation and designed a proof system \cite{MF1} denoted Mil. Similarly to Kleene algebra in subsection \ref{ka}, the signature of Mil as a tuple $(A,+,\cdot,^*,0,1)$ includes a set of atomic actions $A$ and $a,b,c,\cdots\in A$, two special constants with inaction or deadlock denoted $0$ and empty action denoted $1$, two binary functions with sequential composition denoted $\cdot$ and alternative composition denoted $+$, and also a unary function iteration denoted $^*$. 

\begin{definition}[Syntax of Mil]
The expressions (terms) set $\mathcal{T}_{Mil}$ is defined inductively by the following grammar.

$$\mathcal{T}_{Mil}\ni x,y::=0|1|a\in A|x+y|x\cdot y|x^*$$
\end{definition}

Note that Kleene algebra KA in section \ref{ka} and Mil have almost the same grammar structures to express regular language and expressions, but different backgrounds for the former usually initialized to axiomatize the regular expressions and the latter came from process algebra to capture computation.

\begin{definition}[Semantics of Mil]
Let the symbol $\downarrow$ denote the successful termination predicate. Then we give the TSS of Mil as Table \ref{SMil} shows, where $a,b,c,\cdots\in A$, $x,x_1,x_2,x',x_1',x_2'\in\mathcal{T}_{Mil}$.
\end{definition}

\begin{center}
    \begin{table}
        $$\frac{}{1\downarrow}\quad\frac{}{a\xrightarrow{a}1}$$
        $$\frac{x_1\downarrow}{(x_1+x_2)\downarrow}\quad\frac{x_2\downarrow}{(x_1+x_2)\downarrow}\quad\frac{x_1\xrightarrow{a}x_1'}{x_1+x_2\xrightarrow{a}x_1'}\quad\frac{x_2\xrightarrow{b}x_2'}{x_1+x_2\xrightarrow{b}x_2'}$$
        $$\frac{x_1\downarrow\quad x_2\downarrow}{(x_1\cdot x_2)\downarrow} \quad\frac{x_1\xrightarrow{a}x_1'}{x_1\cdot x_2\xrightarrow{a}x_1'\cdot x_2} \quad\frac{x_1\downarrow\quad x_2\xrightarrow{b}x_2'}{x_1\cdot x_2\xrightarrow{b}x_2'}$$
        $$\frac{x\downarrow}{(x^*)\downarrow} \quad\frac{x\xrightarrow{a}x'}{x^*\xrightarrow{a}x'\cdot x^*}$$
        \caption{Semantics of Mil}
        \label{SMil}
    \end{table}
\end{center}

Note that there is no any transition rules related to the constant $0$. Then the axiomatic system of Mil is shown in Table \ref{AxiomsForMil}.

\begin{center}
    \begin{table}
        \begin{tabular}{@{}ll@{}}
            \hline No. &Axiom\\
            $A1$ & $x+y=y+z$\\
            $A2$ & $x+(y+z)=(x+y)+z$\\
            $A3$ & $x+x=x$\\
            $A4$ & $(x+y)\cdot z=x\cdot z+y\cdot z$\\
            $A5$ & $x\cdot(y\cdot z)=(x\cdot y)\cdot z$\\
            $A6$ & $x+0=x$\\
            $A7$ & $0\cdot x=0$\\
            $A8$ & $x\cdot 1=x$\\
            $A9$ & $1\cdot x=x$\\
            $A10$ & $1+x\cdot x^*=x^*$\\
            $A11$ & $(1+x)^*=x^*$\\
            $A12$ & $x=y\cdot x+z\Rightarrow x=y^*\cdot z\textrm{ if }\neg y\downarrow$\\
        \end{tabular}
        \caption{Axioms of Mil}
        \label{AxiomsForMil}
    \end{table}
\end{center}

Note that there are two significant differences between the axiomatic systems of Mil and KA, the axioms $x\cdot 0=0$ and $x\cdot(y+z)=x\cdot y+x\cdot z$ of KA do not hold in Mil.

As mentioned in chapter \ref{intro}, Milner proved the soundness of Mil and remained the completeness open. Just very recently, Grabmayer \cite{MF6} claimed to have proven that Mil is complete with respect to a specific kind of process graphs called LLEE-1-charts which is equal to regular expressions, modulo the corresponding kind of bisimilarity called 1-bisimilarity.

\begin{theorem}[Soundness and completeness of Mil]
Mil is sound and complete with respect to process semantics equality of regular expressions.
\end{theorem} 

\subsection{Truly Concurrent Process Algebra vs. Concurrent Kleene Algebra}\label{tcpacka}

\subsubsection{Truly Concurrent Automata}\label{tcap}

There exist two types of automata recognizing finite N-free pomsets: the branching automata \cite{BA1} \cite{BA2} \cite{BA3} \cite{BA4} and the pomset automata \cite{PA1} \cite{PA2}, and in general case, they are equivalent in terms of expressive power \cite{PA3}.

From the background of (true) concurrency, the truly concurrent automata (TCA) are just the combinations of the branching automata and the pomset automata, although they bring redundancies. The following is the definition of TCA.

\begin{definition}[Truly concurrent automaton]
A truly concurrent automaton (TCA) is a tuple $\mathcal{A}=(Q,A,E,I,F)$ where $Q$ is a finite set of states, $A$ is an alphabet, $I\subseteq Q$ is the set of initial states, $F\subseteq Q$ is the set of final states, and $E$ is the finite set of the transitions of $\mathcal{A}$, and is partitioned into $E=(E_{seq}, E_{fork}, E_{join}, E_{par})$:

\begin{itemize}
  \item $E_{seq}\subseteq Q\times A\times Q$ contains the sequential transitions, which are the transition of traditional Kleene automata;
  \item $E_{fork}\subseteq Q\times\mathcal{M}^{>1}(Q)$ and $E_{join}\subseteq\mathcal{M}^{>1}(Q)\times Q$ are the set of fork and join transitions respectively, where $\mathcal{M}^{>1}(Q)$ is the set of multi-subsets of $Q$ with at least two elements;
  \item $E_{par}:Q\times Q\times Q\rightarrow Q$ is the parallel transition.
\end{itemize}
\end{definition}

From \cite{CKA7}, with the assumptions in section \ref{pesap}, i.e., the causalities among parallel branches are all communications, we know that a pomset with N-shape can be structured and then is N-free (without N-shapes), so it is the so-called series-parallel\cite{CKA3} \cite{CKA4} \cite{CKA7}. We can define the TCA language and its composition.

\begin{definition}[TCA language]
We write the set of pomset with N-shape as $\textsf{TCA}$. A TCA language $L$ is a set of pomset with $L\subseteq\textsf{TCA}$. 
\end{definition}

\begin{definition}[TCA language composition]
For $L,K\subseteq\textsf{TCA}$, we define $L+K=L\cup K$, $L\cdot K=\{l\cdot k|l\in L,k\in K\}$, $L^*=\bigcup_{n\in\mathbb{N}}L^n$ where $L^0=\{1\}$ and $L^{n+1}=L\cdot L^n$, $L\parallel K=\{l\parallel k|l\in L,k\in K\}$, $L^{\dagger}=||_{n\in\mathbb{N}}L^n$ where $L^0=\{1\}$ and $L^{n+1}=L\parallel L^n$, $L\between K=\{l\between k|l\in L,k\in K\}$, $L\mid K=\{l\mid k|l\in L,k\in K\}$.
\end{definition}

Note that, parallel transition $E_{par}$ is corresponding to the composition of parallel branches, and communications lead to join-fork pairs. 

\begin{definition}[Transition relation]
Let $SP(A)$ be the series-parallel pomset of $A$, $p,q,r,s,t\in Q$ and $P,P_1,P_2,P_i\subseteq SP(A)$. We define the transition relation $\xrightarrow[\mathcal{A}]{}\subseteq Q\times SP(A)\times Q$ on a TCA $\mathcal{A}$ as the smallest relation satisfying:

\begin{enumerate}
  \item $p\xrightarrow[\mathcal{A}]{\epsilon}p$ for all $p\in Q$;
  \item $p\xrightarrow[\mathcal{A}]{a}q$ iff $(p,a,q)\in E_{seq}$;
  \item if $p\xrightarrow[\mathcal{A}]{P_1}q$ and $q\xrightarrow[\mathcal{A}]{P_2}r$, then $p\xrightarrow[\mathcal{A}]{P_1P_2}r$;
  \item for all $n>1$, if $p_i\xrightarrow[\mathcal{A}]{P_i}q_i$ for $i\in\{1,\cdots,n\}$, $(p,\{p_1,\cdots,p_n\})\in E_{fork}$, $(\{q_1,\cdots,q_n\},q)\in E_{join}$, then $p\xrightarrow[\mathcal{A}]{\parallel P_i}q$;
  \item if $p\xrightarrow[\mathcal{A}]{P_1}q\in F$ and $r\xrightarrow[\mathcal{A}]{P_2}s\in F$, then $t\xrightarrow[\mathcal{A}]{P_1\parallel P_2}E_{par}(t,p,r)$;
  \item for all $n>1$, if $p\xrightarrow[\mathcal{A}]{a}q$, $(\{p_1,\cdots,p_n\},p)\in E_{join}$, $(q,\{q_1,\cdots,q_n\})\in E_{fork}$, then $\parallel p_i\xrightarrow[\mathcal{A}]{a}\parallel q_i$ for $i\in\{1,\cdots,n\}$.
\end{enumerate}
\end{definition}

\begin{definition}[Pomset bisimulation]
Let $\mathcal{A}=(Q,A,E,I,F)$ and $\mathcal{A}'=(Q',A',E',I',F')$ be two truly concurrent automata with the same alphabet. The automata $\mathcal{A}$ and $\mathcal{A}'$ are pomset bisimilar, $\mathcal{A}\sim_p\mathcal{A}'$, iff there is a relation $R$ between their reachable states that preserves transitions and termination:

\begin{enumerate}
  \item $R$ relate reachable states, i.e., every reachable state of $\mathcal{A}$ is related to a reachable state of $\mathcal{A}'$ and every reachable state of $\mathcal{A}'$ is related to a reachable state of $\mathcal{A}$;
  \item $I$ are related to $I'$ by $R$; 
  \item whenever $p$ is related to $p'$, $pRp'$ and $p\xrightarrow[\mathcal{A}]{P}q$ with $P\subseteq SP(A)$, then there is state $q'$ in $\mathcal{A}'$ with $p'\xrightarrow[\mathcal{A}']{P}q'$ and $qRq'$;
  \item whenever $p$ is related to $p'$, $pRp'$ and $p'\xrightarrow[\mathcal{A}']{P}q'$ with $P\subseteq SP(A)$, then there is state $q$ in $\mathcal{A}$ with $p\xrightarrow[\mathcal{A}]{P}q$ and $qRq'$;
  \item whenever $pRp'$, then $p\in F$ iff $p'\in F'$.
\end{enumerate}
\end{definition}

\subsubsection{Concurrent Kleene Algebra with Communication}{\label{ckacc}}

As we point out, concurrent Kleene algebra (CKA) is based on true concurrency. In this section, we extend the pomset language based CKA in the mostly recent work \cite{CKA7} \cite{CKA70} to the prime event structure based one, that is, CKA with communication. 

We define the syntax and semantics of the TCA expressions.

\begin{definition}[Syntax of TCA expressions]
We define the set of TCA expressions $\mathcal{T}_{CKA}$ as follows.

$$\mathcal{T}_{CKA}\ni x,y::=0|1|a,b\in\mathbb{E}|\gamma(a,b)|x+y|x\cdot y|x^*|x\parallel y|x^{\dagger}|x\mid y|x\between y$$
\end{definition}

In the definition of TCA expressions, the atomic actions include actions in $a,b\in\mathbb{E}$, the constant $0$ denoted inaction without any behaviour, the constant $1$ denoted empty action which terminates immediately and successfully, and also the communication actions $\gamma(a,b)$. The operator $+$ is the alternative composition, i.e., the program $x+y$ either executes $x$ or $y$ alternatively. The operator $\cdot$ is the sequential composition, i.e., the program $x\cdot y$ firstly executes $x$ followed $y$. The Kleene star $x^*$ can execute $x$ for some number of times sequentially (maybe zero). The operator $\parallel$ is the parallel composition, i.e., the program $x\parallel y$ executes $x$ and $y$ in parallel. The parallel star $x^{\dagger}$ can execute $x$ for some number of times in parallel (maybe zero). The program $x\mid y$ executes with synchronous communications. The program $x\between y$ means $x$ and $y$ execute concurrently, i.e., in parallel but maybe with unstructured communications.

\begin{definition}[Semantics of TCA expressions]
Then we define the interpretation of TCA expressions $\sembrack{-}:\mathcal{T}_{CKA}\rightarrow 2^{\textsf{TCA}}$ inductively as Table \ref{STCAE} shows.
\end{definition}

\begin{center}
    \begin{table}
        $$\sembrack{0}_{CKA}=\emptyset \quad \sembrack{a}_{CKA}=\{a\} \quad \sembrack{x\cdot y}_{CKA}=\sembrack{x}_{CKA}\cdot \sembrack{y}_{CKA}$$
        $$\sembrack{1}_{CKA}=\{1\} \quad \sembrack{x+y}_{CKA}=\sembrack{x}_{CKA}+\sembrack{y}_{CKA} \quad\sembrack{x^*}_{CKA}=\sembrack{x}^*_{CKA}$$
        $$\sembrack{x\parallel y}_{CKA}=\sembrack{x}_{CKA}\parallel\sembrack{y}_{CKA} \quad\sembrack{x^{\dagger}}_{CKA}=\sembrack{x}^{\dagger}_{CKA}$$
        $$\sembrack{x\mid y}_{CKA}=\sembrack{x}_{CKA}\mid\sembrack{y}_{CKA}\quad \sembrack{x\between y}_{CKA}=\sembrack{x}_{CKA}\between\sembrack{y}_{CKA}$$
        \caption{Semantics of TCA expressions}
        \label{STCAE}
    \end{table}
\end{center}

We define a concurrent Kleene algebra with communication (CKA) as a tuple $(\mathbb{E},+,\cdot,^*,\parallel,^{\dagger},\between,\mid,0,1)$, where $\mathbb{E}$ is a set, $^*$ and $^{\dagger}$, $+$, $\cdot$, $\parallel$, $\between$ and $\mid$ are binary operators, and $0$ and $1$ are constants, which satisfies the axioms in Table \ref{AxiomsForCKA} for all $x,y,z,h\in \mathcal{T}_{CKA}$ and $a,b,a_0,a_1,a_2,a_3\in\mathbb{E}$, where $x\preceq y$ means $x+y=y$.

\begin{center}
    \begin{table}
        \begin{tabular}{@{}ll@{}}
            \hline No. &Axiom\\
            $A1$ & $x+y=y+z$\\
            $A2$ & $x+(y+z)=(x+y)+z$\\
            $A3$ & $x+x=x$\\
            $A4$ & $(x+y)\cdot z=x\cdot z+y\cdot z$\\
            $A5$ & $x\cdot(y+z)=x\cdot y+x\cdot z$\\
            $A6$ & $x\cdot(y\cdot z)=(x\cdot y)\cdot z$\\
            $A7$ & $x+0=x$\\
            $A8$ & $0\cdot x=0$\\
            $A9$ & $x\cdot 0=0$\\
            $A10$ & $x\cdot 1=x$\\
            $A11$ & $1\cdot x=x$\\
            $P1$ & $x\between y=x\parallel y+x\mid y$\\
            $P2$ & $x\parallel y=y\parallel x$\\
            $P3$ & $x\parallel(y\parallel z)=(x\parallel y)\parallel z$\\
            $P4$ & $(x+y)\parallel z=x\parallel z+y\parallel z$\\
            $P5$ & $x\parallel(y+z)=x\parallel y+x\parallel z$\\
            $P6$ & $(x\parallel y)\cdot (z\between h)\preceq(x\cdot z)\parallel(y\cdot h)$\\
            $P7$ & $x\parallel 0=0$\\
            $P8$ & $0\parallel x=0$\\
            $P9$ & $x\parallel 1=x$\\
            $P10$ & $1\parallel x=x$\\
            $C1$ & $x\mid y=y\mid x$\\
            $C2$ & $(x+y)\mid z=x\mid z+y\mid z$\\
            $C3$ & $x\mid(y+z)=x\mid y+x\mid z$\\
            $C4$ & $\gamma(a,b)\cdot (x\between y)\preceq(a\cdot x)\mid(b\cdot y)$\\
            $C5$ & $x\mid 0=0$\\
            $C6$ & $0\mid x=0$\\
            $C7$ & $x\mid 1=0$\\
            $C8$ & $1\mid x=0$\\
            $A12$ & $1+x\cdot x^*=x^*$\\
            $A13$ & $1+x^*\cdot x=x^*$\\
            $A14$ & $x+y\cdot z\preceq z\Rightarrow y^*\cdot x\preceq z$\\
            $A15$ & $x+y\cdot z\preceq y\Rightarrow x\cdot z^*\preceq y$\\
            $P11$ & $1+x\parallel x^{\dagger}=x^{\dagger}$\\
            $P12$ & $1+x^{\dagger}\parallel x=x^{\dagger}$\\
            $P13$ & $x+y\parallel z\preceq z\Rightarrow y^{\dagger}\parallel x\preceq z$\\
            $P14$ & $x+y\parallel z\preceq y\Rightarrow x\parallel z^{\dagger}\preceq y$\\
        \end{tabular}
        \caption{Axioms of concurrent Kleene algebra}
        \label{AxiomsForCKA}
    \end{table}
\end{center}

By use of the axiomatic system of CKA, we can get the extended Milner's expansion law.

\begin{proposition}[Extended Milner's expansion law]
For $a,b\in \mathbb{E}$, $a\between b=a\cdot b+b\cdot a+a\parallel b+a\mid b$.
\end{proposition}

Let $\equiv_{CKA}$ denote the smallest congruence on $\mathcal{T}_{CKA}$ induced by the axioms of the concurrent Kleene algebra CKA. Since $\equiv_{CKA}$ is the smallest congruence on $\mathcal{T}_{CKA}$ induced by the axioms of the concurrent Kleene algebra CKA, we can only check the soundness of each axiom according to the definition of semantics of TCA expressions. And also by use of communication merge, the TCA expressions are been transformed into the so-called series-parallel ones \cite{CKA3} \cite{CKA4} \cite{CKA7} free of N-shapes. Then we can get the following soundness and completeness theorem \cite{CKA3} \cite{CKA4} \cite{CKA7}.

\begin{theorem}[Soundness and completeness of CKA]
For all $x,y\in\mathcal{T}_{CKA}$, $x\equiv_{CKA} y$ iff $\sembrack{x}_{CKA}=\sembrack{y}_{CKA}$.
\end{theorem}


\subsubsection{Process Algebra with Kleene Star and Parallel Star}{\label{paks}}

As we point out, truly concurrent process algebra APTC \cite{APTC} is also based on true concurrency. In this section, we extend APTC to the one with Kleene star and parallel star. 


Similarly to concurrent Kleene algebra in section \ref{ckacc}, the signature of process algebra APTC with Kleene star and parallel star as a tuple $(\mathbb{E},+,\cdot,^*,\parallel,^{\dagger},\between,\mid,0,1)$ includes a set of atomic actions $\mathbb{E}$ and $a,b,c,\cdots\in \mathbb{E}$, two special constants with inaction or deadlock denoted $0$ and empty action denoted $1$, six binary functions with sequential composition denoted $\cdot$, alternative composition denoted $+$, parallel composition denoted $\parallel$, concurrent composition $\between$ and communication merge $\mid$, and also three unary functions with sequential iteration denoted $^*$ and parallel iteration denoted $^{\dagger}$. 

\begin{definition}[Syntax of APTC with Kleene star and parallel star]
The expressions (terms) set $\mathcal{T}_{APTC}$ is defined inductively by the following grammar.

$$\mathcal{T}_{APTC}\ni x,y::=0|1|a,b\in\mathbb{E}|\gamma(a,b)|x+y|x\cdot y|x^*|x\parallel y|x^{\dagger}|x\mid y|x\between y$$
\end{definition}

Note that concurrent Kleene algebra CKA in section \ref{ckacc} and APTC with Kleene star and parallel star have almost the same grammar structures to express TCA language and expressions, but different semantics.

\begin{definition}[Semantics of APTC with Kleene star and parallel star]
Let the symbol $\downarrow$ denote the successful termination predicate. Then we give the TSS of APTC with Kleene star and parallel star as Table \ref{SAPTCStar}, where $a,b,c,\cdots\in \mathbb{E}$, $x,x_1,x_2,x',x_1',x_2'\in\mathcal{T}_{APTC}$.
\end{definition}

\begin{center}
    \begin{table}
        $$\frac{}{1\downarrow}\quad\frac{}{a\xrightarrow{a}1}$$
        $$\frac{x_1\downarrow}{(x_1+x_2)\downarrow}\quad\frac{x_2\downarrow}{(x_1+x_2)\downarrow}\quad\frac{x_1\xrightarrow{a}x_1'}{x_1+x_2\xrightarrow{a}x_1'}\quad\frac{x_2\xrightarrow{b}x_2'}{x_1+x_2\xrightarrow{b}x_2'}$$
        $$\frac{x_1\downarrow\quad x_2\downarrow}{(x_1\cdot x_2)\downarrow} \quad\frac{x_1\xrightarrow{a}x_1'}{x_1\cdot x_2\xrightarrow{a}x_1'\cdot x_2} \quad\frac{x_1\downarrow\quad x_2\xrightarrow{b}x_2'}{x_1\cdot x_2\xrightarrow{b}x_2'}$$
        $$\frac{x_1\downarrow\quad x_2\downarrow}{(x_1\parallel x_2)\downarrow} \quad\frac{x_1\xrightarrow{a}x_1'\quad x_2\xrightarrow{b}x_2'}{x_1\parallel x_2\xrightarrow{\{a,b\}}x_1'\between x_2'}$$
        $$\frac{x_1\downarrow\quad x_2\downarrow}{(x_1\mid x_2)\downarrow} \quad\frac{x_1\xrightarrow{a}x_1'\quad x_2\xrightarrow{b}x_2'}{x_1\mid x_2\xrightarrow{\gamma(a,b)}x_1'\between x_2'}$$
        $$\frac{x\downarrow}{(x^*)\downarrow} \quad\frac{x\xrightarrow{a}x'}{x^*\xrightarrow{a}x'\cdot x^*}$$
        $$\frac{x\downarrow}{(x^{\dagger})\downarrow} \quad\frac{x\xrightarrow{a}x'}{x^{\dagger}\xrightarrow{a}x'\parallel x^*}$$
        \caption{Semantics of APTC with Kleene star and parallel star}
        \label{SAPTCStar}
    \end{table}
\end{center}

Note that there is no any transition rules related to the constant $0$. Then the axiomatic system of APTC with Kleene star and parallel star is shown in Table \ref{AxiomsForAPTCStar}.

\begin{center}
    \begin{table}
        \begin{tabular}{@{}ll@{}}
            \hline No. &Axiom\\
            $A1$ & $x+y=y+z$\\
            $A2$ & $x+(y+z)=(x+y)+z$\\
            $A3$ & $x+x=x$\\
            $A4$ & $(x+y)\cdot z=x\cdot z+y\cdot z$\\
            $A5$ & $x\cdot(y\cdot z)=(x\cdot y)\cdot z$\\
            $A6$ & $x+0=x$\\
            $A7$ & $0\cdot x=0$\\
            $A8$ & $x\cdot 1=x$\\
            $A9$ & $1\cdot x=x$\\
            $P1$ & $x\between y=x\parallel y+x\mid y$\\
            $P2$ & $x\parallel y=y\parallel x$\\
            $P3$ & $x\parallel(y\parallel z)=(x\parallel y)\parallel z$\\
            $P4$ & $(x+y)\parallel z=x\parallel z+y\parallel z$\\
            $P5$ & $x\parallel(y+z)=x\parallel y+x\parallel z$\\
            $P6$ & $(x\parallel y)\cdot (z\between h)=(x\cdot z)\parallel(y\cdot h)$\\
            $P7$ & $x\parallel 0=0$\\
            $P8$ & $0\parallel x=0$\\
            $P9$ & $x\parallel 1=x$\\
            $P10$ & $1\parallel x=x$\\
            $C1$ & $x\mid y=y\mid x$\\
            $C2$ & $(x+y)\mid z=x\mid z+y\mid z$\\
            $C3$ & $x\mid(y+z)=x\mid y+x\mid z$\\
            $C4$ & $\gamma(a,b)\cdot (x\between y)=(a\cdot x)\mid(b\cdot y)$\\
            $C5$ & $x\mid 0=0$\\
            $C6$ & $0\mid x=0$\\
            $C7$ & $x\mid 1=0$\\
            $C8$ & $1\mid x=0$\\
            $A10$ & $1+x\cdot x^*=x^*$\\
            $A11$ & $(1+x)^*=x^*$\\
            $A12$ & $x=y\cdot x+z\Rightarrow x=y^*\cdot z\textrm{ if }\neg y\downarrow$\\
            $P11$ & $1+x\parallel x^{\dagger}=x^{\dagger}$\\
            $P12$ & $(1+x)^{\dagger}=x^{\dagger}$\\
            $P13$ & $x=y\parallel x+z\Rightarrow x=y^{\dagger}\cdot z\textrm{ if }\neg y\downarrow$\\
        \end{tabular}
        \caption{Axioms of APTC with Kleene star and parallel star}
        \label{AxiomsForAPTCStar}
    \end{table}
\end{center}

Note that there are two significant differences between the axiomatic systems of APTC with Kleene star and parallel star and CKA, the axioms $x\cdot 0=0$ and $x\cdot(y+z)=x\cdot y+x\cdot z$ of CKA do not hold in APTC with Kleene star and parallel star.

By use of the axiomatic system of APTC, we can get the extended Milner's expansion law.

\begin{proposition}[Extended Milner's expansion law]
For $a,b\in \mathbb{E}$, $a\between b=a\cdot b+b\cdot a+a\parallel b+a\mid b$.
\end{proposition}

Then there are two questions: (R) the problem of recognizing whether a given process graph is pomset bisimilar to one in the image of the process interpretation of a $\mathcal{T}_{APTC}$ expression, and (A) whether a natural adaptation of Salomaa’s complete proof system for language equivalence of $\mathcal{T}_{APTC}$ expressions is complete for pomset bisimilarity of the process interpretation of $\mathcal{T}_{APTC}$ expressions. While (R) is decidable in principle, it is just a pomset extension to the problem of recognizing whether a given process graph is bisimilar to one in the image of the process interpretation of a star expression \cite{AP8}.

Since $\equiv_{APTC}$ is the smallest congruence on $\mathcal{T}_{APTC}$ induced by the axioms of the APTC with Kleene star and parallel star, we can only check the soundness of each axiom according to the definition of TSS of TCA expressions. As mentioned in the section \ref{mf}, just very recently, Grabmayer \cite{MF6} claimed to have proven that Mil is complete with respect to a specific kind of process graphs called LLEE-1-charts which is equal to regular expressions, modulo the corresponding kind of bisimilarity called 1-bisimilarity. Based on this work, we believe that we can get the following conclusion based on the corresponding pomset bisimilarity and let the proof of the completeness be open.

\begin{theorem}[Soundness and completeness of APTC with Kleene star and parallel star]
APTC with Kleene star and parallel star is sound and complete with respect to process semantics equality of TCA expressions.
\end{theorem}